\documentclass[sigconf,nonacm]{acmart}

\usepackage{amsmath}
\usepackage{amsthm}
\usepackage{mathtools}
\usepackage{enumitem}
\usepackage{algorithm}
\usepackage{algpseudocode}
\usepackage{thmtools}
\usepackage{thm-restate}

\usepackage{xcolor}
\usepackage[capitalise, noabbrev]{cleveref}
\usepackage{color-edits}
\usepackage{graphicx,tikz,booktabs}
\usetikzlibrary{arrows.meta}
\usepackage{thmtools}

\declaretheorem[style=plain, name=Theorem]{theorem}
\declaretheorem[style=plain, name=Lemma, numberwithin=section]{lemma}
\declaretheorem[style=plain, name=Proposition, numberwithin=section]{proposition}
\declaretheorem[style=plain, name=Corollary, numberwithin=section]{corollary}
\declaretheorem[style=definition, name=Definition, numberwithin=section]{definition}
\declaretheorem[style=definition, name=Example, numberwithin=section]{example}
\declaretheorem[style=remark, name=Remark, numberwithin=section]{remark}

\newcommand{\game}{\mathcal{G}}                 %
\newcommand{\players}{\Sigma}                   %
\newcommand{\team}{\Pi}                         %
\newcommand{\opp}{\mathsf{O}}%
\newcommand{\block}{\mathsf{scan}} 
\newcommand{\states}{S}                         %
\newcommand{\target}{T}                         %
\newcommand{\sinit}{s_0}                        %
\newcommand{\obj}{\mathcal{W}}                  %

\newcommand{\st}{\operatorname{st}}             %
\newcommand{\Dist}{\Delta}                      %
\newcommand{\Supp}{\mathsf{Supp}}               %
\newcommand{\val}{\mathrm{val}}                 %
\DeclarePairedDelimiter{\abs}{\lvert}{\rvert}   %

\newcommand{\stratp}[1][p]{\sigma_{#1}}                 %
\newcommand{\stratteam}{\sigma_{\team}}                 %
\newcommand{\stratrest}{\tau_{\opp}}  %
\newcommand{\stratopp}{\tau}                            %
\newcommand{\Strats}[1]{\Lambda_{#1}}                   %

\newcommand{\Prb}[2]{\mathbb{P}_{#1}^{#2}}      %
\newcommand{\Reach}[1]{\Diamond #1}             %

\newcommand{\tuple}[1]{\langle #1 \rangle}

\newcommand{\imdp}{\mathcal{M}}                 %
\newcommand{\mergedmdp}{\widehat{\imdp}}        %
\newcommand{\Trap}{\mathsf{Trap}}               %
\newcommand{\rk}{\rho}                         %
\newcommand{\tsink}{t^{*}}                      %
\newcommand{\fsink}{f^{*}}                      %
\newcommand{\Forests}{\mathcal{F}}              %
\newcommand{\wt}{w}                             %

\usepackage{complexity}
\renewcommand{\K}{\mathbb{K}}                     %
\usepackage{tikz}
\usetikzlibrary{arrows.meta}
\newcommand{\rankbox}[1]{{\setlength{\fboxsep}{1.1pt}\fcolorbox{gray!60}{gray!12}{\scriptsize$#1$}}} %

\newcommand{\sgoal}{s_{\mathsf{goal}}}
\newcommand{\sfail}{s_{\mathsf{fail}}}
\usepackage[framemethod=TikZ]{mdframed}

\title{Limit-sure Winning in Team Games without Shared Randomness}

\author{L\'eonard Brice}
\affiliation{
  \institution{Institute of Science and Technology Austria}
  \city{Klosterneuburg}
  \country{Austria}}
\email{leonard.brice@ista.ac.at}

\author{Thomas A. Henzinger}
\affiliation{
  \institution{Institute of Science and Technology Austria}
  \city{Klosterneuburg}
  \country{Austria}}
\email{tah@ista.ac.at}

\author{Alipasha Montaseri}
\affiliation{
  \institution{\'Ecole Polytechnique F\'ed\'erale de Lausanne}
  \city{Lausanne}
  \country{Switzerland}}
\email{alipasha.montaseri@epfl.ch}

\author{Ali Shafiee}
\affiliation{
  \institution{Institute of Science and Technology Austria}
  \city{Klosterneuburg}
  \country{Austria}}
\email{ali.shafiee@ista.ac.at}

\author{K. S. Thejaswini}
\affiliation{
  \institution{Universit\'e libre de Bruxelles}
  \city{Brussels}
  \country{Belgium}}
\email{thejaswini.raghavan@ulb.be}

\begin{abstract}
We study concurrent stochastic games in which a team of agents plays against an adversarial environment. If the agents could access a shared source of randomness, then these games reduce to two-player games.
Despite being common in the literature on coalition games, this assumption does not accurately model real-world systems.
In this paper, we consider the setting where the agents only have access to individual sources of randomness, and must satisfy a given reachability or safety objective.
The limit-sure problem asks whether the team can ensure that the objective is satisfied with probability at least $1-\lambda$ for every $\lambda>0$.
For reachability objectives, we show that this problem is $\NP$-complete.
For safety objectives, we show that limit-sure winning coincides with almost-sure winning, which is known to lie in $\P$.
As a corollary of these results, we conclude that model checking of the qualitative fragment of the logic Individually Randomised ATL---introduced in the recent work of Brice et al.---is $\P^\NP$-complete.
\end{abstract}


\begin{document}

\maketitle

\section{Introduction}\label{sec:intro}
Planning for a group of agents in an uncertain and reactive world is a fundamental problem in multi-agent systems. A standard model is the
\emph{concurrent stochastic game} in which, at every round, each agent simultaneously picks an
action, and so does the environment. The next state is selected following a distribution that depends on the state and on the chosen profile. Agents that pursue a common goal against such an environment form a \emph{team}, and the environment is modelled as a single adversarial opponent.

Randomisation is often necessary to achieve the optimal value in concurrent games. If the team players chose their action deterministically, then an adversarial environment would be able to thwart the attempt of this team.
This can already be observed in two-player games, like rock-paper-scissors.

The usual semantics that describe the winning of a coalition in such graph games include 
alternating-time temporal logic
(ATL)~\cite{AHK02}, and its probabilistic variants: Randomised ATL (RATL)~\cite{dAHK98} and Probabilistic ATL (PATL)~\cite{ChenLu07}. Such logics however model coalitions as a single player, making the implicit assumption that it can resort to a common source of randomness. This fails to accurately model several real-world agents, which may not be able to coordinate their random choices. 
For example, different sensors cannot coordinate their random choices after a network partition, and different components of a system that were programmed together and then
shipped apart do not share a common source of randomness. The team is then restricted to using only \emph{product} distributions.
Recently, the work of Brice et al.~\cite{BHMST26} initiated a study of team games with individual randomisation, building on previous work of de Alfaro and Faella~\cite{dAF03}.
We work with the same semantics as in their work and in the work of de Alfaro and Faella. 

This setting formulates questions that are relevant to both planning and verification. 
The standard architecture in multi-agent planning is centralised planning followed by decentralised execution. A planner that knows the model designs one controller per agent offline, and at run time each agent executes its own controller with its own source of randomness. For planning scenarios, a planner is given a reliability requirement in addition to a goal, say ``reach the target with probability at least $1-\lambda$''.
There are two ways that these questions  can be formulated. 
Can the requirement be met for one given $\lambda$, and can it be
met for \emph{every} $\lambda$, thus ensuring that the plan succeeds with probability arbitrarily close to~1?
We focus on the second question, known as \emph{limit-sure} winning.
A positive instance ensures that one can construct a plan with as much accuracy as needed.

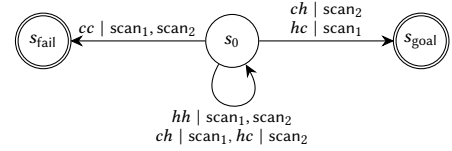
\begin{figure}[t]
  \centering
  \newcommand{\ja}[2]{#1\,{\color{black!45}\textbar}\,\textit{#2}}%
    \begin{tikzpicture}[x=1cm,y=1cm,font=\footnotesize,
      >={Stealth[length=1.7mm,width=1.5mm]}]
    \tikzset{
      gstate/.style={circle,draw,fill=white,minimum size=7mm,inner sep=0.5pt},
      gabs/.style={gstate,double,double distance=0.6pt},
      elab/.style={font=\scriptsize,align=center,inner sep=1.5pt}
    }
    \node[gabs] (fail) at (-2.5,0) {$\sfail$};
    \node[gstate] (s0) at (0,0) {$\sinit$};
    \node[gabs] (goal) at (2.5,0) {$\sgoal$};
    \draw[->] (s0) -- node[elab,above] {$cc\mid \block_1,\block_2$} (fail);
    \draw[->] (s0) -- node[elab,above] {$ch\mid \block_2$\\$hc\mid \block_1$} (goal);
    \draw[->] (s0) to[out=240,in=300,looseness=6,min distance=8mm]
      node[elab,below] {$hh\mid \block_1,\block_2$\\$ch\mid \block_1$, $hc\mid \block_2$} (s0);
  \end{tikzpicture}
  \caption{The team game described in \cref{example:introexample}.  }

  \Description{A game graph with three states: a failure state on the left, the initial state in the middle, and a goal state on the right. The failure and goal states are absorbing and drawn with double circles. From the initial state, the play moves to the failure state if both agents cross and the adversary scans. It moves to the goal state if both agents cross and the adversary waits, or if both agents hide and the adversary scans. It stays in the initial state if both agents hide and the adversary waits, or if the agents disagree, whatever the adversary does.}
  \label{fig:snowballexample}
\end{figure}

\begin{example}
\label{example:introexample}
Two robots must cross a guarded door.
At each round, each agent chooses to cross ($c$) or hide ($h$), while the opponent chooses to scan either the first robot ($\block_1$) or the second ($\block_2$).
If both robots choose to cross, since the door is narrow, the robots break and reach $\sfail$ irrespective of the opponent's action. 
If robot~$i$ crosses while the opponent plays $\block_i$, robot~$i$ is caught, but the team can try again.  
If exactly one robot, say, robot~$i$ attempts to cross, it can reach $\sgoal$ only if the opponent does not choose $\block_i$.  If both hide, then nothing changes, and the game is tried again. 
The transitions are shown in~\Cref{fig:snowballexample}.

If the team had a shared source of randomness, then they could choose between $(c,h)$ and $(h,c)$ uniformly at random, and the team crosses almost surely irrespective of the opponent strategy, since they end up either in $s_0$ or in $\sgoal$ with positive probability at each attempt. 
However, if the team had no shared source of randomness, then by choosing between $(c,h)$ and $(h,c)$, they also choose $(c,c)$ and reach the failure state with positive probability.

However, even without a shared source of randomness, the team can win limit surely.
Let both robots cross with probability $\epsilon\in (0,1)$ at every round. Therefore, they can reach the state $\sgoal$ with probability $\epsilon(1-\epsilon)$ for each opponent action.
The probability that they both play cross is $\epsilon^2$. Hence, they end up in the same state $s_0$ with probability $1-\epsilon$.
This therefore gives us 
\[
 \Pr(\text{reach }\sgoal)
 =\frac{\epsilon(1-\epsilon)}
        {\epsilon(1-\epsilon)+\epsilon^2}
 =1-\epsilon.
\]
For every $1>\lambda>0$, choosing
$0<\epsilon<\lambda$ ensures that the team can succeed with probability greater
than $1-\lambda$.
\end{example}

\paragraph{Our contributions.} 
Our main contribution is in showing the computational complexity of limit-sure reachability. We also show the computational complexity of limit-sure safety games.
\begin{enumerate}[itemsep=2pt,topsep=3pt]
  \item We show that deciding whether a team without shared randomness can satisfy reachability limit surely
    is $\mathsf{NP}$-complete (\cref{thm:np-upper-bound}).  The lower bound follows from the hardness reduction of almost-sure reachability, since both limit-sure and almost-sure winning coincide on those instances~\cite{BHMST26}.
    We provide an $\NP$ upper bound by giving a polynomial certificate that represents the family of strategies, parametrised by $\lambda$, that can be verified in polynomial time.
    
    \item For safety objectives, we show that limit-sure winning and almost-sure winning
    coincide (\cref{thm:safety-limit-almost}). Thus, deciding whether a team can limit surely win a safety game is in~$\P$. 
\end{enumerate}

As a corollary, we show in \cref{thm:iratl-model-checking} that model checking the qualitative fragment of the logic IRATL, introduced by Brice et al.~\cite{BHMST26}, is in $\P^\NP$.
We also show that it is complete for this class.

To prove the $\NP$ upper bound for reachability, we represent a family of team strategies by one team strategy
whose action probabilities are algebraic Puiseux series in~$\epsilon$
(\Cref{lem:puiseux}).
This follows the approach of Asadi et al.~\cite{AsadiEtAl2025},
who study POMDPs using
Puiseux methods originating with Bewley and Kohlberg~\cite{BewleyKohlberg76}. Only the 
leading exponent of each action probability matters for our case, defining a notion of \emph{rank} of an action. 
Thanks to a standard result linking absorption probabilities to spanning forests~\cite{FreidlinWentzell12}, we reduce our problem to guessing the ranks
and then comparing minimum-rank spanning forests
(\cref{lem:forest}). We also show that these ranks only require polynomially many bits using Cramer's rule and Hadamard's
inequality (\cref{lem:small-ranks}).

Guessing the team ranks alone gives an $\NP^{\NP}$ upper bound by using an $\NP$ oracle that detects a violating opponent response using
Edmonds' algorithm~\cite{Edmonds67}.
To improve this and obtain an $\NP$ upper bound, we construct an encoded MDP whose rank encoding breaks ties.
We guess a response and certify its optimality through local
inequalities%
~\cite{Frank81,Schrijver03}.
This reduces verification to polynomial time.
We give the details in~\Cref{sec:np-upper-bound}.

\subsection*{Related work} 
\paragraph{Concurrent games and qualitative objectives.}
Concurrent stochastic games are the standard model of a controller acting against a reactive environment~\cite{dAHK98}. For two players and reachability, the sure, almost-sure and limit-sure winning sets are computable in polynomial time~\cite{dAHK07}, and the quantitative threshold decision problem is in the complexity class $\exists\mathbb{R}$ (the complexity class associated with the existential theory of the reals~\cite{ETRComplexityclass}) with $\mathsf{SQRTSUM}$-hardness as the best known lower bound~\cite{CdAH12,EY15}. %
The difference between almost-sure and limit-sure reachability appears already in recursive games~\cite{Everett57}.
The asymptotic analysis of families of strategies ensuring limit-sure winning, by algebraic means, goes back to Bewley and Kohlberg and is an active line of research. 
For example, Solan used them to analyse perturbed Markov chains, communicating classes, and exit
distributions~\cite{Solan2003}.
A related use of leading transition exponents appears in Young's minimum-resistance-tree characterisation of stochastic stability~\cite{Young93}.

\paragraph{Teams, coalitions and sharing randomness.}
If the team instead had a shared source of randomness, it could be interpreted as a single player, and then the sure, almost-sure and limit-sure reachability sets are all computable in polynomial time~\cite{dAHK07}. This setting reduces to two-player concurrent games, and there are several works that study such games~\cite{dAHK98,CdAH12,HKM09}. 
Brice et al.~\cite{BHMST26}, following the work of de Alfaro and Faella~\cite{dAF03}, developed team games where players need not have access to shared randomness. The work of de Alfaro and Faella already established that even for  almost-sure reachability (asking whether there is a strategy to reach a goal with probability~1), the problem is $\NP$-complete.
In recent work, Brice et al.~\cite{BHMST26} extended the problem to its quantitative variant and considered the threshold problem, which asks if the target can be reached with a probability strictly greater than a given threshold.
They show that memoryless team strategies are sufficient for the threshold problem, which ensures that the threshold problem is in $\exists\mathbb{R}$. 
The complexity of the limit-sure problem for teams was left open (\cite[Problem~3]{BHMST26}), together with the limit-sure safety problem
(\cite[Problem~4]{BHMST26}). We settle both.

\paragraph{Equilibria vs team games.} The gap between shared and individual sources of randomness is the gap that the extensive-form literature calls \emph{ex ante} versus \emph{no
coordination}. Team-maxmin equilibria were introduced by von Stengel and Koller~\cite{vonStengelK97}. 
Celli and Gatti~\cite{CG18} consider three different communication models, namely correlation both before and during
play, correlation only before play, and none at all, and show that the
loss from restricting communication can be unbounded in the size of the
game tree. Our setting is similar to their third setting, but on infinite plays on a graph. 
Follow-up works of the above give algorithms for ex ante coordinated team strategies~\cite{ FarinaCGS18,ZhangAC21}. 
Basilico et al.~\cite{BCDG17} study team max-min equilibria for one-shot games in normal form and give efficiency bounds and algorithms for teams that cannot correlate their strategy. Kalogiannis et al.~\cite{KAPVC23} consider stationary approximate Nash equilibria in discounted adversarial team Markov games with independent team policies. 
The literature above works with one-shot games or finite trees, or equilibria, whereas we work on 
qualitative, limit-sure winning on graphs using different techniques.

\paragraph{POMDPs and decentralised planning.} Centralised planning with decentralised execution is the standard setting for Dec-POMDPs, where the general
problem with finite horizon is $\mathsf{NEXP}$-complete~\cite{BGIZ02,OliehoekAmato16}.  Such Dec-POMDPs are useful in cooperative multi-agent reinforcement
learning~\cite{ZhangYB21}. Our model is orthogonal to partial observation, which it replaces with an active adversary. 
Under partial observation, with infinite horizon, limit-sure reachability is undecidable in general, already for blind POMDPs~\cite{GO10}; but Asadi et al. show that it is $\mathsf{NP}$-complete for policies with a fixed memory bound~\cite{AsadiEtAl2025}. 
Their proof uses Puiseux strategies to obtain polynomial-size
rank witnesses, and our proof builds on this approach.
However, fixing a team strategy leaves an adversarial MDP.
Our novel ingredient is a polynomial-time verification procedure
that certifies the condition against all opponent responses
by checking only one response.

\paragraph{Logics for multi-agent coalition.}
Alternating-time Temporal Logic (ATL) \cite{AHK02} is used to model multi-agent systems that navigate a graph. 
Some extensions of ATL, like PATL~\cite{ChenLu07} or RATL~\cite{dAHK98}, also consider randomised strategies. However, their semantics treat a coalition as a meta-player,  implicitly assuming that they share randomness.
Model checkers for stochastic multi-agent systems~\cite{PRISMgames20} also assume shared randomisation.
Individually Randomised ATL (IRATL)~\cite{BHMST26} was proposed to consider settings where teams do not communicate or share randomness during play. Our results settle the complexity of limit-sure reachability and safety in this logic, 
which leads to a tight complexity bound for model checking the
qualitative fragment of IRATL. Probabilistic Strategy Logic~\cite{AKMMR19} is expressive enough to model individual randomisation, but is undecidable in general.

\section{Preliminaries}\label{sec:prelims}
\paragraph{Notation} We write $\mathbb{N}$ for the set of nonnegative integers. For a finite set $X$, $\Dist(X)$ is the set of probability distributions over $X$, and $\Supp(d) = \{x \in X \mid d(x) > 0\}$ is the \emph{support} of $d \in \Dist(X)$. For sets $X$ and $Y$, $Y^{X}$ is the set of functions from $X$ to $Y$. $X^{+}$ and $X^{\omega}$ are the sets of nonempty finite sequences and of infinite sequences over $X$.

\begin{definition}[Game]
\label{def:game}
We define a \emph{game} $\game$ as a tuple
$\game=\tuple{\players,\states,(A_p)_{p\in\players},
\delta,\team,\sinit,\obj}$.
The sets $\players$ and $\states$ are finite sets of players and states.
Each player $p$ has a finite nonempty action set $A_p$.
The transition function is
$\delta\colon\states\times\prod_{p\in\players}A_p\to\Dist(\states)$.
The team is $\team\subseteq\players$, and the initial state is
$\sinit\in\states$.
The objective $\obj\subseteq\states^\omega$ is a set of infinite
state sequences.
\end{definition}

We call the players in $\players\setminus\team$ the \emph{opponents}.
For $X\subseteq\players$, the set
$A_X=\prod_{p\in X}A_p$ contains its \emph{action profiles}.
We write a complete profile as
$(\bar a,\bar b)\in A_{\team}\times A_{\players\setminus\team}$.
Here, $\bar a=(a_p)_{p\in\team}$ is the team profile,
and $\bar b$ is the opponents' profile.

The game starts in $\sinit$. At each round, every player picks an action simultaneously and independently. If the current state is $s$ and the players pick the profile $(\bar{a}, \bar{b})$, the next state is determined according to $\delta(s, (\bar{a}, \bar{b}))$. We call the resulting infinite sequence of states a \emph{play}. 
The game is zero-sum, meaning that the team wins the plays in $\obj$, and the opponents try to keep the play out of $\obj$.

\begin{definition}[Strategies, team strategies]
  \label{def:strategies}
  A \emph{history} is a non-empty finite sequence of states, an element of $\states^{+}$. A \emph{strategy} for player $p \in \players$ is a mapping $\stratp \colon \states^{+} \to \Dist(A_p)$. It is \emph{deterministic} if $\stratp(h)$ assigns probability~$1$ to one action at every history $h$. It is \emph{memoryless} if $\stratp(h)$ depends only on the last state of $h$, and we write $\stratp(s)$ for its value at the histories ending in $s$. A \emph{team strategy} is a tuple $\stratteam = (\stratp)_{p \in \team}$ of strategies, one per team player. At a history $h$, it induces the product distribution $\stratteam(h)$ over $A_{\team}$, defined for every team profile $(a_p)_{p \in \team}$ by
  \begin{equation}
    \label{eq:product}
    \stratteam(h)\bigl((a_p)_{p \in \team}\bigr) \;=\; \prod_{p \in \team} \stratp(h)(a_p).
  \end{equation}
  Generally, for a set $\players'$ of players we write
  $
    \Strats{\players'}^{\mathrm h} \;=\; \prod_{p \in \players'} \bigl(\Dist(A_p)\bigr)^{\states^{+}}
  $
  for the set of tuples $(\stratp)_{p \in \players'}$ of strategies, and $\Strats{\players'} \subseteq \Strats{\players'}^{\mathrm h}$ for those that are memoryless. A team strategy is thus an element of $\Strats{\team}^{\mathrm h}$.
\end{definition}

\begin{remark}
  \label{rem:no-shared-randomness}
  The product form of \cref{eq:product} is the defining feature of our setting: each team player randomises privately, so the team can only play product distributions over $A_{\team}$. By contrast, a team with shared randomness can play any joint distribution in $\Dist(A_{\team})$. %
  The game then reduces to a two-player zero-sum concurrent %
  game.
\end{remark}

\paragraph{Objectives}
\label{sec:objectives}

The objective of a game specifies which plays the team wins. We consider reachability and safety objectives. Reachability expresses that the play eventually enters a target state, and safety expresses that the play never leaves the set of safe states.

\begin{definition}[Reachability objective]
  \label{def:reachability-objective}
  For a set $\target \subseteq \states$ of \emph{target states}, the \emph{reachability objective} requires to visit $\target$ at least once:
  \[
    \Reach{\target}
    = \{s^0 s^1 \cdots \in \states^{\omega} \mid \exists k \ge 0,\ s^k \in \target\}.
  \]
\end{definition}

\begin{definition}[Safety objective]
  \label{def:safety-objective}
  For a set $C \subseteq \states$ of \emph{safe states}, the \emph{safety objective} requires that a play remain in $C$ forever:
  \[
    \Box C
    = \{s^0 s^1 \cdots \in \states^{\omega} \mid \forall k \ge 0,\ s^k \in C\}.
  \]
\end{definition}

\begin{definition}[Markov decision process]
  \label{def:mdp}
  A \emph{Markov decision process} (MDP) is a tuple $(V,A,P)$,
  where $V$ is a finite nonempty set of states, $A$ is a finite
  nonempty set of actions, and
  $P\colon V\times A\to\Dist(V)$ is a transition function.
  Its \emph{support} is the set of triples $(u,a,v)$ with
  $P(u,a)(v)>0$.
  Fixing a memoryless team strategy
  $\stratteam\in\Strats{\team}$ induces the MDP
  $(\states,A_{\opp},P_{\stratteam})$, where

\[
  P_{\stratteam}(s,b)(t)
  =\sum_{\bar a\in A_{\team}}
    \stratteam(s)(\bar a)\cdot\delta(s,(\bar a,b))(t)
  \qquad
  \forall s,t\in\states,\ b\in A_{\opp}.
\]

\end{definition}

\begin{definition}[Markov chain]
  \label{def:markov-chain}
  A \emph{Markov chain} (MC) is a pair $(V,P)$,
  where $V$ is a finite nonempty set of states and
  $P\colon V\times V\to[0,1]$ is a \emph{transition matrix}.
  A state $u\in V$ is \emph{absorbing} if $P(u,u)=1$.
  For $u\in V$ and $R\subseteq V$, we write
  $\Prb{u}{P}(\Reach{R})$ for the probability that the chain
  started in $u$ eventually visits $R$. 
  Fixing memoryless strategies
  $\stratteam\in\Strats{\team}$ and
  $\stratopp\in\Strats{\opp}$ induces the Markov chain
  $(\states,P_{\stratteam,\stratopp})$, where

\[
  P_{\stratteam,\stratopp}(s,t)
  =\sum_{\substack{\bar a\in A_{\team}, b\in A_{\opp}}}
    \stratteam(s)(\bar a)\cdot\stratopp(s)(b)
    \cdot\delta(s,(\bar a,b))(t)
  \quad
  \forall s,t\in\states.
\]

\end{definition}

By a classical result~\cite{Norris98}, the probabilities
$\Prb{v}{P}(\Reach{R})$ for $v\in V$ form the coordinatewise
least solution $q\in\mathbb{R}_{\geq0}^V$ of the system
\begin{equation}
  \label{eq:reach-system}
  q_v=1\ \text{ for }v\in R,
  \qquad
  q_v=\sum_{w\in V}P(v,w)\,q_w\
  \text{ for }v\in V\setminus R.
\end{equation}

\paragraph{Winning conditions}
\label{sec:winning-conditions}

Every team strategy $\stratteam$ and tuple $\stratrest$ of strategies for the opponents induces a probability measure on plays starting in $\sinit$, which we denote by $\Prb{\sinit}{\stratteam,\stratrest}$. In particular, we write $\Prb{\sinit}{\stratteam,\stratrest}(\obj)$ for the probability of satisfying the objective. The team chooses its strategy first and seeks to maximise this probability against every response by the opponents.

\begin{definition}[Almost-sure winning]
  \label{def:almost-sure}
  The team \emph{wins almost surely} if some team strategy ensures the objective with probability $1$ against every tuple of strategies for the opponents:
  \[
    \exists \stratteam \in \Strats{\team}^{\mathrm h}\; : \;
    \forall \stratrest \in \Strats{\players \setminus \team}^{\mathrm h},
    \qquad \Prb{\sinit}{\stratteam,\stratrest}(\obj)=1.
  \]
\end{definition}

\begin{definition}[Limit-sure winning]
  \label{def:limit-sure}
  The \emph{value} of a game $\game$ is
  \[
    \val(\game) \;=\; \sup_{\stratteam \in \Strats{\team}^{\mathrm h}} \, \inf_{\stratrest \in \Strats{\players \setminus \team}^{\mathrm h}} \, \Prb{\sinit}{\stratteam, \stratrest}(\obj).
  \]
  The team \emph{wins limit surely} if $\val(\game)=1$, or equivalently, if
  \[
    \forall \lambda \in (0,1)\; : \;
    \exists \stratteam \in \Strats{\team}^{\mathrm h}\; : \;
    \forall \stratrest \in \Strats{\players \setminus \team}^{\mathrm h},
    \qquad \Prb{\sinit}{\stratteam,\stratrest}(\obj)>1-\lambda.
  \]
\end{definition}

\begin{remark}
  \label{rem:single-opponent}
  We may replace the opponents by a single opponent $\opp$ with
  action set $A_{\opp}=A_{\players\setminus\team}$.
  For each fixed team strategy, this does not change the infimum of the
  winning probability. See the full version. Thus, we assume that there is only one opponent $\opp$.
\end{remark}

We write $\Strats{\opp}^{\mathrm h}$ for the strategies of the
single opponent and $\Strats{\opp}$ for its memoryless strategies.
Let $D$ be the finite set of deterministic
memoryless strategies, which we identify with maps
$\stratopp\colon\states\to A_{\opp}$. Since deterministic strategies are sufficient, throughout the paper we assume an opponent response is an element of $D$.

We study the complexity of deciding whether the team wins a given game limit surely. We assume that games are given explicitly. Transition probabilities are rationals encoded in binary, $\delta$ is represented by a table indexed by states and complete action profiles, and the objective is specified by its set of target or safe states. The \emph{size} of the game is the total length of this encoding.

\section{Reachability}\label{sec:reach}
We now focus on the reachability objective $\obj=\Reach{\target}$.
We assume without loss of generality that each target state is \emph{absorbing}, that is,
$\delta(t,(\bar a,\bar b))(t)=1$ for every $t\in\target$ and action
profile $(\bar a,\bar b)\in A_{\players}$.

\begin{definition}[Limit-sure reachability problem]
\label{def:problem}
Given a reachability game, does the team win limit surely?
\end{definition}

Brice et al.~\cite{BHMST26} as well as De Alfaro and Faella~\cite{dAF03} show that almost-sure reachability is $\mathsf{NP}$-complete.
The reduction of Brice et al. from the $k$-clique problem, in particular, produces a game where
the team wins almost surely or has value at most
$3/4$~\cite[Theorem~4]{BHMST26}, which also proves $\mathsf{NP}$-hardness for limit-sure reachability.

We show in four steps that the problem belongs to $\mathsf{NP}$.
First, we show that limit-sure winning can be achieved with team
action probabilities given by Puiseux series in $\epsilon$.
Second, we show that winning depends only on their leading exponents,
called \emph{ranks}, which can be represented with polynomially many bits.
Third, for fixed team and opponent strategies, we check winning by
comparing minimum-rank spanning forests.
Finally, we give a certificate that uses the team ranks and one
opponent response.
We check this response through local inequalities in an encoded MDP.
These inequalities certify optimality without checking all opponent strategies.
We give the full NP proof in the full version.

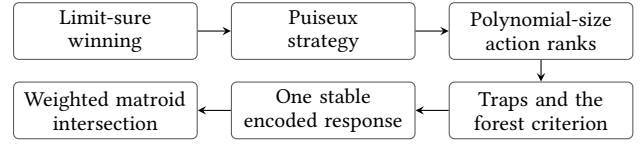
\begin{figure}[htb]
\centering
\begin{tikzpicture}[
    x=.34\linewidth,y=1cm,
    font=\fontsize{8}{9}\selectfont,
    >=stealth,
    stage/.style={
      draw=black!65,
      rounded corners=2pt,
      text width=.27\linewidth,
      minimum height=.75cm,
      inner xsep=2pt,
      inner ysep=2pt,
      align=center
    }
  ]
  \node[stage] (limit) at (0,0)
    {Limit-sure\\winning};
  \node[stage] (puiseux) at (1,0)
    {Puiseux\\strategy};
  \node[stage] (ranks) at (2,0)
    {Polynomial-size\\action ranks};

  \node[stage] (forests) at (2,-1.05)
    {Traps and the\\forest criterion};
  \node[stage] (response) at (1,-1.05)
    {One stable\\encoded response};
  \node[stage] (matroid) at (0,-1.05)
    {Weighted matroid\\intersection};

  \draw[->] (limit.east) -- (puiseux.west);
  \draw[->] (puiseux.east) -- (ranks.west);
  \draw[->] (ranks.south) -- (forests.north);
  \draw[->] (forests.west) -- (response.east);
  \draw[->] (response.west) -- (matroid.east);
\end{tikzpicture}
\caption{Proof outline for the $\NP$ upper bound.}
\Description{Six boxes form a sequence across two rows:
limit-sure winning, a Puiseux strategy, polynomial-size action ranks,
traps and the forest criterion, one stable encoded response,
and weighted matroid intersection.}
\label{fig:reachability-pipeline}
\end{figure}

\begin{proposition}[{\cite[Theorem~1]{BHMST26}}]
\label{prop:memoryless}
Let $t\in[0,1]$.
If some team strategy ensures reachability with probability greater
than $t$ against every opponent response, then some memoryless team
strategy does.
\end{proposition}

By \Cref{prop:memoryless}, restricting the team to memoryless
strategies preserves the value.
Fixing such a team strategy gives an MDP with an optimal deterministic
memoryless opponent strategy.
Thus, we restrict attention to memoryless team strategies and
deterministic memoryless opponent strategies.

\subsection{Puiseux Strategies}
\label{sec:puiseux}

Limit-sure winning allows the strategy to depend on the error $\lambda$.
In \Cref{example:introexample}, each agent crosses with probability
$\epsilon$ and hides with probability $1-\epsilon$.
The winning probability approaches $1$ as $\epsilon$ goes to zero.
Puiseux series describe such a family by a symbolic strategy.

\begin{definition}[Puiseux series]
\label{def:puiseux-series}
A \emph{Puiseux series} has the form
$x=\sum_{i\geq i_0}c_i\epsilon^{i/m}$, where
$i_0\in\mathbb Z$, $m\in\mathbb N_{>0}$, and $c_i\in\mathbb R$.
\end{definition}

We write $\K$ for the field of algebraic Puiseux series.
The sign of a nonzero series is the sign of its first nonzero
coefficient.
Thus, $\epsilon$ is a positive infinitesimal, with
$0<\epsilon<t$ for every real $t>0$.
For bounded $x\in\K$, its \emph{standard part} $\st(x)$ is the
coefficient of $\epsilon^0$, or zero if absent.
We use the following standard properties~\cite{BPR06}.
Here, first-order formulas use polynomial equalities and
inequalities, logical connectives, and quantifiers.

\begin{lemma}[Puiseux series properties~\cite{BPR06}]
\label{lem:puiseux-calculus}
The field $\K$ is a real closed ordered extension of $\mathbb R$.
\begin{enumerate}[label=(\roman*),itemsep=1pt,topsep=2pt]
  \item \emph{Transfer.}
  Every first-order sentence with real parameters holds over
  $\mathbb R$ if and only if it holds over $\K$.

  \item \emph{Evaluation.}
  Let $\bar x=(x_1,\ldots,x_n)\in\K^n$.
  Suppose a first-order formula $\Phi$ with real parameters
  satisfies $\Phi(\bar x)$ over $\K$.
  There is a real $r_0>0$ such that, for every real $0<r<r_0$,
  all series $x_i$ converge at $\epsilon=r$ and
  $\Phi(x_1(r),\ldots,x_n(r))$ holds over $\mathbb R$.
\end{enumerate}
\end{lemma}

A \emph{Puiseux team strategy} is a memoryless strategy with
probabilities in $\K$.
A Puiseux team strategy $\stratteam$ is \emph{Puiseux-winning} if $\st\bigl(\Prb{\sinit}{\stratteam,\stratopp}
       (\Reach{\target})\bigr)=1$ for every opponent response $\stratopp \in D$.
\begin{lemma}
\label{lem:puiseux}
The team wins limit surely if and only if some Puiseux team strategy
is Puiseux-winning.
\end{lemma}

\begin{proof}[Proof sketch]
We express strategy validity by polynomial constraints.
We describe reachability values by the least nonnegative solution
of \eqref{eq:reach-system}.
Thus, the limit-sure condition is first-order.
Transfer with error $\epsilon$ gives
$
  1-\epsilon
    <\Prb{\sinit}{\stratteam,\stratopp}(\Reach{\target})
    \leq1
$
for every
$\stratopp\in D$.

Each probability has standard part $1$.
Conversely, fix a positive real error $\lambda$ and a Puiseux-winning
strategy.
Each success probability is greater than $1-\lambda$ in $\K$.
Evaluation gives a real strategy with the same inequalities.
One sufficiently small value works for all responses, since $D$ is
finite.
\end{proof}

We may assume that all strategy exponents are nonnegative integers.
Bounded probabilities have no negative exponents.
We choose a common denominator $d$ and substitute $\epsilon=\eta^d$.
This substitution makes the exponents integers and preserves winning.

\subsubsection{Ranks}
\label{sec:ranks}

The \emph{rank} $\rk(x)$ of a nonzero series is its smallest exponent.
We put $\rk(0)=\infty$ and $\min\emptyset=\infty$.
For nonnegative series, leading coefficients cannot cancel.
Hence,
\begin{align*}
  \rk(xy)&=\rk(x)+\rk(y),\qquad 
  \rk(x+y)=\min\{\rk(x),\rk(y)\},\\
  \rk(x/y)&=\rk(x)-\rk(y)\qquad\text{if }y\neq0.
\end{align*}
We write $r_{p,s,a}=\rk(\stratp(s)(a))$ for an action rank.
Since action probabilities sum to $1$, their minimum rank at each
$(p,s)$ is zero.
Conversely, any such ranks are realised by
\begin{equation}
\label{eq:realisation}
  \stratp(s)(a)=
  \frac{\epsilon^{r_{p,s,a}}}
       {\sum_{a':\,r_{p,s,a'}<\infty}\epsilon^{r_{p,s,a'}}}
  \quad\text{for all }a\in A_p\text{ with }r_{p,s,a}<\infty,
\end{equation}
with probability zero for actions of rank $\infty$.
Independence makes profile ranks the sum of action ranks. 
Against an opponent action $b$, the transition rank is
\begin{equation}
\label{eq:transition-rank}
  \rk(s,b,t)=
  \min_{\bar a:\,\delta(s,(\bar a,b))(t)>0}
       \sum_{p\in\team}r_{p,s,a_p}.
\end{equation}

\begin{example}
\label{ex:ranks-robot}
The diagram below shows the strategy from
\Cref{example:introexample} and its induced MDP.
For either opponent action, reaching $\sgoal$ has rank
$\rk(\epsilon(1-\epsilon))=1$, whereas reaching $\sfail$
has rank $\rk(\epsilon^2)=2$.
Returning to $\sinit$ has rank $\rk(1-\epsilon)=0$.
Thus, success has smaller rank than failure and occurs first
with probability tending to $1$ as $\epsilon$ goes to zero.
Boxes show ranks, and black dots mark opponent actions.
Absorbing states are double-circled.

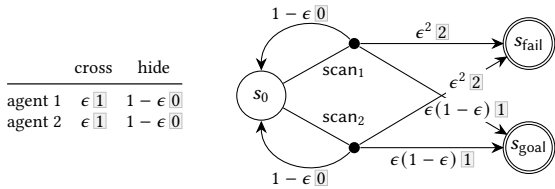
\begingroup
\centering
\newcommand{\reachrankbox}[1]{%
  \begingroup
  \setlength{\fboxsep}{0.5pt}%
  \setlength{\fboxrule}{0.3pt}%
  \fcolorbox{gray!50}{gray!10}{\ensuremath{#1}}%
  \endgroup}
\resizebox{.96\linewidth}{!}{%
\begin{tikzpicture}[
  x=1cm,y=1cm,font=\fontsize{8}{9}\selectfont,
  >={Stealth[length=1.7mm,width=1.5mm]}
]
  \tikzset{
    gstate/.style={
      circle,draw,fill=white,
      minimum size=6.6mm,inner sep=0.5pt
    },
    gabs/.style={gstate,double,double distance=0.6pt},
    act/.style={
      circle,fill=black,inner sep=0pt,minimum size=1.5mm
    },
    plab/.style={
      font=\fontsize{7}{8}\selectfont,inner sep=1.2pt
    },
    alab/.style={
      font=\fontsize{7}{8}\selectfont\itshape,inner sep=1pt
    }
  }

  \node[anchor=east] at (-0.95,0) {\fontsize{7}{8}\selectfont
    \setlength{\tabcolsep}{2.5pt}%
    \begin{tabular}{@{}lcc@{}}
      & cross & hide \\ \midrule
      agent 1
        & $\epsilon$\,\reachrankbox{1}
        & $1-\epsilon$\,\reachrankbox{0} \\
      agent 2
        & $\epsilon$\,\reachrankbox{1}
        & $1-\epsilon$\,\reachrankbox{0}
    \end{tabular}};

  \node[gstate] (s0)   at (0,0)       {$\sinit$};
  \node[act]    (sc)   at (1.25,0.7)  {};
  \node[act]    (wt)   at (1.25,-0.7) {};
  \node[gabs]   (fail) at (3.6,0.7)   {$\sfail$};
  \node[gabs]   (goal) at (3.6,-0.7)  {$\sgoal$};

  \draw (s0) --
    node[alab,below right,pos=0.55] {$\block_1$} (sc);
  \draw (s0) --
    node[alab,above right,pos=0.55] {$\block_2$} (wt);

  \draw[->] (sc) --
    node[plab,above] {$\epsilon^2$\,\reachrankbox{2}} (fail);
  \draw[->] (sc) --
    node[plab,fill=white,pos=0.72]
      {$\epsilon(1-\epsilon)$\,\reachrankbox{1}} (goal);
  \draw[->] (wt) --
    node[plab,fill=white,pos=0.72]
      {$\epsilon^2$\,\reachrankbox{2}} (fail);
  \draw[->] (wt) --
    node[plab,below]
      {$\epsilon(1-\epsilon)$\,\reachrankbox{1}} (goal);

  \draw[->] (sc) to[out=135,in=90,looseness=1.2]
    node[plab,above,pos=0.45]
      {$1-\epsilon$\,\reachrankbox{0}} (s0);
  \draw[->] (wt) to[out=225,in=270,looseness=1.2]
    node[plab,below,pos=0.45]
      {$1-\epsilon$\,\reachrankbox{0}} (s0);

  \path (fail) edge[loop right,draw=none]
    node[plab,right] {\phantom{$1$\,\reachrankbox{0}}} (fail);
\end{tikzpicture}%
}
\par
\endgroup
\end{example}

We use spanning forests to express absorption probabilities. 
Let $R$ be a set of absorbing states of a finite chain $P$ on $V$.
A \emph{spanning forest rooted at $R$} selects one outgoing
support edge at each vertex outside $R$.
We require that following the selected edges from any vertex
eventually reaches a root in $R$.
If this root is $r$, we say that the vertex \emph{leads to $r$}
in the forest.
Each root leads to itself.
We write $\Forests$ for these forests.
For $R'\subseteq R$, we write $\Forests_{R'}(u)$ for the forests
in which $u$ leads to a root in $R'$.
The weight and rank of a forest are
\begin{align*}
  \wt(F)&=\prod_{(v,w)\in F}P(v,w),
  &\rk(F)&=\sum_{(v,w)\in F}\rk(P(v,w)).
\end{align*}

The following lemma combines the classical forest
formula~\cite[Chapter~6, Lemma~3.3]{FreidlinWentzell12}
with the rank rules.

\begin{lemma}
\label{lem:forest-formula}
Let $P$ be a finite Puiseux Markov chain on $V$ with
absorbing states $R$.
Suppose every state has a support path to $R$.
Then $\sum_{F\in\Forests}\wt(F)>0$, and for all
$u\in V$ and $R'\subseteq R$,
\begin{align}
  \Prb{u}{P}(\Reach{R'})
    &=\frac{\sum_{F\in\Forests_{R'}(u)}\wt(F)}
            {\sum_{F\in\Forests}\wt(F)},
      \label{eq:forest-formula-short}\\
  \rk\bigl(\Prb{u}{P}(\Reach{R'})\bigr)
    &=\min_{F\in\Forests_{R'}(u)}\rk(F)
      -\min_{F\in\Forests}\rk(F).
      \label{eq:rank-criterion-short}
\end{align}
\end{lemma}

\begin{proof}[Proof sketch]
At each nonroot, we choose the first edge of a shortest support
path to $R$.
These edges form a forest of positive weight, so the denominator
is positive.
The classical forest formula~\cite{FreidlinWentzell12}
transfers to $\K$ by~\Cref{lem:puiseux-calculus} and
gives~\eqref{eq:forest-formula-short}.

Each forest weight has rank $\rk(F)$.
Their positive leading coefficients cannot cancel.
Thus, each nonzero sum has the minimum rank of its terms.
Division subtracts ranks and gives~\eqref{eq:rank-criterion-short}.
\end{proof}

\begin{example}
\label{ex:forest-formula}
The diagram below shows a chain with roots $3,4$ and its three
spanning forests.
Bold arrows mark the chosen edges.
The forests have weights $1/8$, $3/8$, and $1/8$.
Only the last two lead from $1$ to $3$.
By~\eqref{eq:forest-formula-short},
$\Prb{1}{P}(\Reach{\{3\}})
=\frac{3/8+1/8}{1/8+3/8+1/8}=\frac45$.

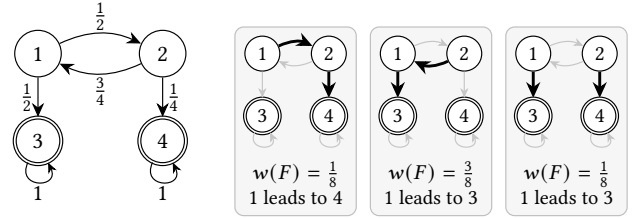
\begingroup
\centering
\resizebox{.96\linewidth}{!}{%
\begin{tikzpicture}[
    x=1cm,y=1cm,font=\fontsize{8}{9}\selectfont,
    >={Stealth[length=1.7mm,width=1.5mm]}]
  \tikzset{
    mcstate/.style={
      circle,draw,fill=white,
      minimum size=5.6mm,inner sep=0.5pt
    },
    mcroot/.style={mcstate,double,double distance=0.6pt},
    tstate/.style={
      circle,draw,fill=white,
      minimum size=4.2mm,inner sep=0pt,
      font=\fontsize{7}{8}\selectfont
    },
    troot/.style={tstate,double,double distance=0.5pt},
    tlab/.style={
      align=center,font=\fontsize{7}{8}\selectfont,inner sep=1pt
    },
    unused/.style={
      ->,gray!45,line width=0.4pt,
      >={Stealth[length=1.2mm,width=1mm]}
    },
    chosen/.style={
      ->,line width=1pt,
      >={Stealth[length=1.9mm,width=1.7mm]}
    },
    loopy/.style={
      ->,>={Stealth[length=1.3mm,width=1.1mm]}
    },
    box/.style={
      draw=gray!50,fill=gray!7,rounded corners=3pt
    }
  }

  \node[mcstate] (s1) at (0,0) {$1$};
  \node[mcstate] (s2) at (1.5,0) {$2$};
  \node[mcroot] (s3) at (0,-1.05) {$3$};
  \node[mcroot] (s4) at (1.5,-1.05) {$4$};

  \draw[->] (s1) to[bend left=25]
    node[above,inner sep=1pt] {$\tfrac12$} (s2);
  \draw[->] (s2) to[bend left=25]
    node[below,inner sep=1pt] {$\tfrac34$} (s1);
  \draw[->] (s1) --
    node[left,inner sep=1pt] {$\tfrac12$} (s3);
  \draw[->] (s2) --
    node[right,inner sep=1pt] {$\tfrac14$} (s4);
  \draw[loopy] (s3)
    to[out=-115,in=-65,min distance=3.5mm]
    node[below,inner sep=1pt] {$1$} (s3);
  \draw[loopy] (s4)
    to[out=-115,in=-65,min distance=3.5mm]
    node[below,inner sep=1pt] {$1$} (s4);

  \foreach \x/\name in {2.7/A,4.33/B,5.96/C} {
    \draw[box] (\x-0.33,0.32) rectangle (\x+1.13,-1.95);
    \node[tstate] (\name1) at (\x,0) {$1$};
    \node[tstate] (\name2) at (\x+0.8,0) {$2$};
    \node[troot] (\name3) at (\x,-0.75) {$3$};
    \node[troot] (\name4) at (\x+0.8,-0.75) {$4$};

    \draw[unused] (\name1) to[bend left=25] (\name2);
    \draw[unused] (\name2) to[bend left=25] (\name1);
    \draw[unused] (\name1) -- (\name3);
    \draw[unused] (\name2) -- (\name4);
    \draw[unused] (\name3)
      to[out=-115,in=-65,min distance=2.5mm] (\name3);
    \draw[unused] (\name4)
      to[out=-115,in=-65,min distance=2.5mm] (\name4);
  }

  \draw[chosen] (A1) to[bend left=25] (A2);
  \draw[chosen] (A2) -- (A4);
  \draw[chosen] (B2) to[bend left=25] (B1);
  \draw[chosen] (B1) -- (B3);
  \draw[chosen] (C1) -- (C3);
  \draw[chosen] (C2) -- (C4);

  \node[tlab] at (3.1,-1.55)
    {$\wt(F)=\tfrac18$\\$1$ leads to $4$};
  \node[tlab] at (4.73,-1.55)
    {$\wt(F)=\tfrac38$\\$1$ leads to $3$};
  \node[tlab] at (6.36,-1.55)
    {$\wt(F)=\tfrac18$\\$1$ leads to $3$};
\end{tikzpicture}%
}
\par
\endgroup
\end{example}

\subsubsection{Small ranks}
\label{sec:small-ranks}

A \emph{strongly connected component} (SCC) is a maximal set of mutually
reachable states.
A \emph{bottom} SCC has no outgoing support edge.
For each opponent response, we make the bottom SCCs of its chain absorbing.
Failure then means reaching a bottom SCC outside $\target$.
By \eqref{eq:rank-criterion-short}, this probability is infinitesimal
exactly when its forest numerator has a larger minimum rank than the
denominator.
Thus, winning depends only on the action ranks.

\begin{lemma}
\label{lem:small-ranks}
Let $N=|\{(p,s,a):p\in\team,\ s\in\states,\ a\in A_p\}|$.
If the team wins limit surely, it has a Puiseux-winning strategy whose
finite action ranks are integers between $0$ and $N^{N/2}$.
\end{lemma}

\begin{proof}[Proof sketch]
We prove it by a system of linear inequalities.
We fix the support of a winning strategy.
We choose one team profile for each forest edge.
The resulting rank is a sum of action ranks, each used at most once.
For each opponent response, fix a minimum forest term of the original
winning strategy.
We require every failure term to exceed this term by at least $1$.
These inequalities have coefficients in $\{-1,0,1\}$.
Their right sides are $1$.
We also keep nonnegativity and the zero ranks.
Their right sides are $0$.

We take a vertex of this nonempty system and scale it to an integer vector.
The scaled vector satisfies the same inequalities.
Cramer's rule and Hadamard's inequality bound its entries by
$N^{N/2}$.
Indeed, at a vertex we select an invertible matrix
$A\in\{-1,0,1\}^{k\times k}$ with $k\leq N$ and a right side
$b\in\{0,1\}^k$.
Multiplying the vertex by $|\det A|$ gives integer entries bounded
by $k^{k/2}$.
The multiplier is at least $1$, so the inequalities and zero ranks
are preserved.
Each rank needs $O(N\log(N+1))$ bits.
The realisation in \eqref{eq:realisation} is winning because only ranks matter.
\end{proof}

\subsection{Puiseux Markov Decision Processes}
\label{sec:puiseux-mdp}

Different opponent responses may produce different bottom SCCs.
We define a failure region that contains every bottom SCC
outside the target, under any opponent response.
Making this region and the target absorbing gives the same roots
for all forest comparisons.

\subsubsection{The induced Puiseux MDP}

We fix action ranks and their realisation as in \eqref{eq:realisation}.
Let $\imdp$ be the induced MDP, with transitions $P(s,b)(t)$ and
ranks as in \eqref{eq:transition-rank}.
Each state and action has a successor of rank zero, because its
transition probabilities sum to $1$.

\subsubsection{Traps and the merged MDP}
\label{sec:traps}
\label{sec:forests}

An \emph{end component} is a pair $(X,\Gamma)$, where $X$ is a nonempty set of states and $\Gamma(x)$ is a nonempty
set of opponent actions at each $x\in X$, such that (1) every action in $\Gamma(x)$ has all its successors in $X$, and (2)
these actions make $X$ strongly connected.
We call $(X,\Gamma)$ a \emph{maximal end component} (MEC) if
$X$ and every $\Gamma(x)$ are inclusion-maximal.
MECs can be computed from an MDP in polynomial
time~\cite{CH14}.
Put
\begin{align*}
  \Trap=\bigcup\{X:(X,\Gamma)\text{ is a MEC and }X\cap\target=\emptyset\}.
\end{align*}
The opponent can keep the play in $\Trap$ once it enters.
For example, a non-target state with a deterministic self-loop action
belongs to $\Trap$, even if another action reaches the target.

We make $\Trap$ absorbing and obtain $\imdp^\sharp$.
We write $\Prb{u}{\stratopp}$ for its probabilities under $\stratopp$.
The next lemma gives the winning criterion.

\begin{lemma}
\label{lem:trap}
Every response reaches $\target\cup\Trap$ with probability $1$ from
every state in $\imdp^\sharp$.
The strategy $\stratteam$ is Puiseux-winning if and only if
$
  \st\bigl(\Prb{\sinit}{\stratopp}(\Reach{\Trap})\bigr)=0$
for every $\stratopp\in D$.
\end{lemma}

\begin{proof}[Proof sketch]
Every non-target bottom SCC of a fixed response is contained in a
target-free MEC.
Thus, the play reaches $\target\cup\Trap$ with probability $1$.
Failure in the original MDP implies visiting $\Trap$.
Conversely, any response can be changed inside $\Trap$ to keep the
play there forever.
This change preserves the probability of entering $\Trap$.
Hence, a response with non-infinitesimal trap probability prevents
Puiseux winning.
\end{proof}

If $\sinit\in\target$ or $\Trap=\emptyset$, the strategy wins.
If $\sinit\in\Trap$, or no target is reachable in $\imdp^\sharp$,
it loses.
Otherwise, we merge $\target$ into $\tsink$ and $\Trap$ into $\fsink$.
The resulting MDP $\mergedmdp$ has states
$V=U\cup\{\tsink,\fsink\}$, where
$U=\states\setminus(\target\cup\Trap)$.
Ranks into a merged state are the minimum ranks into its members.
For example,
$ \rk(u,b,\fsink)=\min_{x\in\Trap}\rk(u,b,x)
$.

Its edge set $E$ contains the triples $e=(u,b,v)$ of finite rank.
Thus, the team ranks determine the graph for all responses.

We fix a response $\stratopp:U\to A_{\opp}$.
A root-forest is \emph{consistent} with $\stratopp$ if its edge at
$u$ uses $\stratopp(u)$.
It is \emph{trapped} if $\sinit$ leads to $\fsink$ and
\emph{targeted} if $\sinit$ leads to $\tsink$.
We write $\mathcal F_{\stratopp}$ for the consistent forests.
Their minimum trapped and targeted ranks are
\begin{align*}
  \rho_{\stratopp}^{\fsink}
    &=\min\{\rk(F):F\in\mathcal F_{\stratopp},\ \sinit\text{ leads to }\fsink\},\\
  \rho_{\stratopp}^{\tsink}
    &=\min\{\rk(F):F\in\mathcal F_{\stratopp},\ \sinit\text{ leads to }\tsink\}.
\end{align*}

\begin{lemma}
\label{lem:forest}
The strategy $\stratteam$ is Puiseux-winning if and only if
$\rho_{\stratopp}^{\fsink}>\rho_{\stratopp}^{\tsink}$ for every response $\stratopp$.
For a fixed response, both minima are computable in polynomial time.
\end{lemma}

The criterion follows from \eqref{eq:rank-criterion-short}, since the
failure probability has rank
$\rho_{\stratopp}^{\fsink}-\min\{\rho_{\stratopp}^{\fsink},\rho_{\stratopp}^{\tsink}\}$.
Each minimum reduces to a minimum-weight directed spanning tree
problem~\cite{Edmonds67,ChuLiu65,Tarjan77}.
For the trapped minimum, we delete $\tsink$ and redirect its
incoming edges to $\sinit$.
A consistent forest becomes a tree directed towards $\fsink$
exactly when $\sinit$ led to $\fsink$.
Otherwise, the redirection creates a cycle through $\sinit$.
We exchange the roots to compute the targeted minimum.
Both transformations preserve edge ranks.

\subsection{The \texorpdfstring{$\mathsf{NP}$}{NP} Upper Bound}
\label{sec:np-upper-bound}

There may be exponentially many opponent responses.
For fixed team ranks, an $\NP$ oracle can guess a violating
response and check it in polynomial time.
Guessing the team ranks and asking whether such a response exists gives an $\NP^{\NP}$ upper bound.
To obtain an $\NP$ upper bound, we guess the team ranks together
with one response in an encoded MDP.
We give a criterion that can be checked in polynomial time
and certifies that this response is optimal.

\subsubsection{The encoded MDP}
We index the edges of $\mergedmdp$ as $e_0,\ldots,e_{m-1}$.
For $J\subseteq E$, we put
\begin{align*}
  \widehat{\rk}(e_i)&=2^{m+1}\rk(e_i)+2^i,
  &\widehat{\rk}(J)&=\sum_{e\in J}\widehat{\rk}(e).
\end{align*}
The powers of two distinguish edge sets.
Their sum is less than $2^m$, so the encoding preserves strict
comparisons of original ranks.
All encoded ranks have polynomially many bits.

Let $E(u,b)$ be the edges with source $u$ and label $b$.
We keep the support of $\mergedmdp$ and replace its probabilities by
\begin{align}
  \widetilde P(u,b,v)=
      \epsilon^{\widehat{\rk}(u,b,v)}/Z_{u,b},
      \quad \text{where }
  Z_{u,b}=\sum_{e\in E(u,b)}\epsilon^{\widehat{\rk}(e)}.
      \notag
      \label{eq:generic-transition}
\end{align}
Other probabilities are zero and the roots are absorbing.
The encoded ranks are exponents before normalisation.
For a response $\stratopp$, let $q_v^{\stratopp}$ be its probability
of reaching $\fsink$ in this encoded MDP.
The opponent maximises this probability.

We write $\mathcal F_{\stratopp}(v)$ for the forests leading
from $v$ to $\fsink$.
We put
$
  N_v^{\stratopp}
    =\sum_{F\in\mathcal F_{\stratopp}(v)}
         \epsilon^{\widehat{\rk}(F)},
  \ D^{\stratopp}
    =\sum_{F\in\mathcal F_{\stratopp}}
         \epsilon^{\widehat{\rk}(F)}.
$
Each forest uses one edge per nonroot.
Thus, every forest weight has the same normalising factors,
which cancel in~\eqref{eq:forest-formula-short}.
Since $D^{\stratopp}>0$, we obtain
$q_v^{\stratopp}=N_v^{\stratopp}/D^{\stratopp}$.

\subsubsection{One stable response suffices}

A response $\tau$ is \emph{stable} if
\begin{equation}
\label{eq:stable-policy}
  q_u^{\stratopp}\geq
    \sum_{v\in V}\widetilde P(u,b,v)q_v^{\stratopp}
  \qquad\text{for all }u\in U\text{ and }b\in A_{\opp}.
\end{equation}
The right side is the value of playing $b$ once and then following
$\stratopp$.
Equality holds for $b=\stratopp(u)$.
These local inequalities suffice because every response has a support
path to a root from each state.

\begin{lemma}
\label{lem:np-stable-optimal}
A stable response exists.
Every stable response is optimal from every state.
\end{lemma}

\begin{proof}[Proof sketch]
We use a \emph{minimum principle}: if $d\in\K^V$ is zero at the roots and satisfies $d_u\geq\sum_v\widetilde P(u,\stratopp(u),v)d_v$ for all 
  $u\in U$, then $d\geq0$.
Otherwise, a negative minimum would propagate along a support
path to a root, where $d$ is zero.

Fix a stable response $\stratopp$ and any response $\stratopp'$.
Stability and the reachability equations show that
$d=q^{\stratopp}-q^{\stratopp'}$ satisfies the minimum principle
under $\stratopp'$.
Hence, $q^{\stratopp}\geq q^{\stratopp'}$, so $\stratopp$ is optimal.

For existence, choose a response maximising $\sum_v q_v^{\stratopp}$.
Such a response exists because there are finitely many responses.
If stability fails at $(u,b)$, switch the action at $u$ to $b$.
The minimum principle shows that no value decreases.
The strict improvement at $(u,b)$ gives a larger value at $u$.
This contradicts maximality of the sum.
\end{proof}

\begin{lemma}
\label{lem:np-one-policy-certificate}
The team strategy is Puiseux-winning if and only if
there exists a stable encoded response $\stratopp$ such that
$\rho_{\stratopp}^{\fsink}>\rho_{\stratopp}^{\tsink}$.
\end{lemma}

\begin{proof}[Proof sketch]
By \Cref{lem:forest}, the team strategy is Puiseux-winning
only when every response satisfies the forest criterion.
If this holds, we choose a stable response using
\Cref{lem:np-stable-optimal}.

For the converse, we show that any response violating the criterion
has a larger encoded failure probability than every response
satisfying it.
The encoding multiplies each forest rank by $2^{m+1}$ and adds
a number between $0$ and $2^m-1$.
Applying~\eqref{eq:rank-criterion-short} therefore gives
\begin{align*}
  \rho_{\stratopp}^{\fsink}\leq\rho_{\stratopp}^{\tsink}
  &\implies \rk(q_{\sinit}^{\stratopp})\leq 2^m-1,\\
  \rho_{\stratopp}^{\fsink}>\rho_{\stratopp}^{\tsink}
  &\implies \rk(q_{\sinit}^{\stratopp})\geq 2^m+1.
\end{align*}
Here, $\rk(0)=\infty$ also covers the case where no forest
leads from $\sinit$ to $\fsink$.
A positive Puiseux value with smaller rank is larger.
Thus, a response violating the criterion has a larger failure
probability than one satisfying it.

Now suppose a stable response satisfies the criterion.
If another response violated it, that response would have a larger
failure probability.
This contradicts the optimality of the stable response
from~\Cref{lem:np-stable-optimal}.
Hence every response satisfies the criterion, and the team
strategy is Puiseux-winning.
\end{proof}

\subsubsection{Checking stability in polynomial time}

We define the \emph{Bellman residual} as
$R^{\stratopp}(u,b)=q_u^{\stratopp}
-\sum_{v\in V}\widetilde P(u,b,v)q_v^{\stratopp}$.
Stability requires all residuals to be nonnegative.
The residual is zero for $b=\stratopp(u)$.
Fix $b\neq\stratopp(u)$.
For each forest $F$, we put $\chi_F(v)=1$ if $v$ leads to
$\fsink$, and $0$ otherwise.
Substituting $q_v^{\stratopp}=N_v^{\stratopp}/D^{\stratopp}$ gives
\begin{align*}
  Z_{u,b}D^{\stratopp}R^{\stratopp}(u,b)
  &=\sum_e\epsilon^{\widehat{\rk}(e)}
      (N_u^{\stratopp}-N_v^{\stratopp})\\
  &=\sum_{e,F}(\chi_F(u)-\chi_F(v))
      \epsilon^{\widehat{\rk}(F)+\widehat{\rk}(e)}.
\end{align*}
Here, $e=(u,b,v)$ ranges over $E(u,b)$ and
$F$ over $\mathcal F_{\stratopp}$.
If $u,v$ lead to the same root, the coefficient is zero.
Otherwise, adding $e$ joins the two trees of $F$ into one tree.
Its coefficient is $+1$ if $u$ leads to $\fsink$ and $-1$ otherwise.
Let $\mathcal T(\stratopp,u,b)$ contain the spanning trees
of the underlying undirected graph that select one edge
from each $E(z,\stratopp(z))$, for $z\in U$, and one from $E(u,b)$.
We retain edge labels.
Deleting the latter edge recovers the forest, so these trees
correspond to the nonzero terms above.
Writing $\operatorname{sgn}(T)$ for their coefficients gives
\begin{equation}
\label{eq:signed-tree-identity}
  Z_{u,b}D^{\stratopp}R^{\stratopp}(u,b)
  =\sum_{T\in\mathcal T(\stratopp,u,b)}
     \operatorname{sgn}(T)\epsilon^{\widehat{\rk}(T)}.
\end{equation}
Without encoding, minimum-rank trees may cancel out.
Encoded ranks are distinct, so the unique minimum tree determines
the residual sign, since $Z_{u,b}D^{\stratopp}>0$.
If no tree exists, the residual is zero.

\begin{example}
\label{ex:signed-tree-test}
Let $U=\{u,v\}$ and $\sinit=u$.
Action $a$ takes $u$ to $v$.
Action $b$ at $u$ and the only action $c$ at $v$ each have
both roots as successors.
Below, $e_i:d$ labels edge $e_i$ of action $d$.
Double circles mark absorbing roots, whose self-loops are omitted.

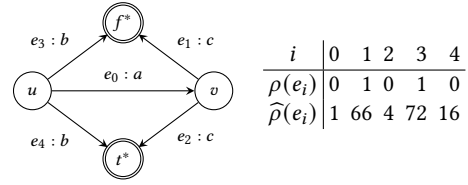
\begin{center}
\begingroup
\setlength{\arraycolsep}{2pt}
\begin{tikzpicture}[
    x=1cm,y=1cm,font=\scriptsize,>=stealth,
    baseline=(current bounding box.center),
    state/.style={circle,draw,minimum size=5mm,inner sep=1pt},
    root/.style={state,double,double distance=0.5pt}]
  \node[state] (u) at (0,0) {$u$};
  \node[state] (v) at (2.4,0) {$v$};
  \node[root]  (f) at (1.2,0.9) {$\fsink$};
  \node[root]  (t) at (1.2,-0.9) {$\tsink$};
  \draw[->] (u) -- node[above] {$e_0:a$} (v);
  \draw[->] (v) -- node[above right] {$e_1:c$} (f);
  \draw[->] (v) -- node[below right] {$e_2:c$} (t);
  \draw[->] (u) -- node[above left] {$e_3:b$} (f);
  \draw[->] (u) -- node[below left] {$e_4:b$} (t);
\end{tikzpicture}
\quad
$\begin{array}{c|rrrrr}
  i&0&1&2&3&4\\ \hline
  \rk(e_i)&0&1&0&1&0\\
  \widehat{\rk}(e_i)&1&66&4&72&16
\end{array}$
\endgroup
\end{center}

Here $m=5$, so $\widehat{\rk}(e_i)=64\rk(e_i)+2^i$.
Suppose $\stratopp(u)=a$ and $\stratopp(v)=c$.
Its two forests are $F_{\fsink}=\{e_0,e_1\}$ and
$F_{\tsink}=\{e_0,e_2\}$.
They send $u$ through $v$ to $\fsink$ and $\tsink$, respectively.

We check action $b$ at $u$.
Adding $e_3$ to $F_{\fsink}$ or $e_4$ to $F_{\tsink}$
creates an undirected cycle.
In both cases, the endpoints already lead to the same root,
so the term is zero.
The other pairs give the positive tree
$T_+=\{e_0,e_1,e_4\}$ of encoded rank $83$
and the negative tree $T_-=\{e_0,e_2,e_3\}$ of encoded rank $77$.
Both have original rank $1$, so their terms without encoding
cancel as $\epsilon-\epsilon=0$.
After encoding, \eqref{eq:signed-tree-identity} gives
$Z_{u,b}D^{\stratopp}R^{\stratopp}(u,b)
=\epsilon^{83}-\epsilon^{77}<0$.
Thus, $\stratopp$ is not stable.

Switching to $b$ at $u$ reverses the tree signs when we check $a$.
The new response $\stratopp'$ is therefore stable.
Its original forest ranks satisfy
$\rho_{\stratopp'}^{\fsink}=1>0=\rho_{\stratopp'}^{\tsink}$,
giving the certificate in~\Cref{lem:np-one-policy-certificate}.
\end{example}

\begin{lemma}
\label{lem:np-bellman-test}
Residual signs can be checked
in polynomial time.
\end{lemma}

\begin{proof}[Proof sketch]
We restrict the edges to the prescribed groups.
Selecting at most one edge per group is a partition matroid constraint.
Having no undirected cycle is a graphic matroid constraint.
Loops and pairs of parallel edges count as cycles.
There are $|U|+1=|V|-1$ groups.
Thus, common independent sets of this size select one edge per
group and form spanning trees.

Weighted matroid intersection finds a tree of minimum encoded
rank or reports none exists in polynomial
time~\cite{Schrijver03}.
We delete its edge from $E(u,b)$ and follow the forest from $u$
to determine its sign.
Equation~\eqref{eq:signed-tree-identity} gives the residual sign.
The encoded ranks have polynomial bit length. We check
polynomially many residuals.
\end{proof}

\begin{algorithm}[h]
\caption{Nondeterministic Algorithm for Limit-Sure Reachability}
\label{algo:limit-sure-reachability}
\begin{algorithmic}[1]
  \Statex \textbf{Input:} A game $\game$, an initial state $\sinit$, and targets $\target$.
  \Statex \textbf{Output:} Accept limit-sure winning or reject
  \State Guess $r_{p,s,a}\in\{0,\ldots,B\}\cup\{\infty\}$
         for every $(p,s,a)$.
  \State Reject if $\min_a r_{p,s,a}\neq0$ for some $(p,s)$.
  \State Build the support and transition ranks of $\imdp$
         using~\eqref{eq:transition-rank}.
  \State Compute $\Trap$ as the union of target-free MECs.
  \State Accept if $\sinit\in\target$ or $\Trap=\emptyset$.
  \State Make $\Trap$ absorbing to obtain $\imdp^\sharp$.
  \State Reject if $\sinit\in\Trap$ or no path from $\sinit$ reaches
         $\target$ in $\imdp^\sharp$.
  \State Merge $\target$ into $\tsink$ and $\Trap$ into $\fsink$ to obtain $\mergedmdp$.
  \State Compute the encoded edge ranks $\widehat{\rk}$ for the encoded MDP.
  \State Guess one action $\stratopp(u)\in A_{\opp}$ at each $u\in U$.
  \State Check all residual signs using~\Cref{lem:np-bellman-test}.
  \State Reject if $R^{\stratopp}(u,b)<0$ for any $(u,b)\in U\times A_{\opp}$.
  \State Using the original ranks $\rk$, compute the minimum
         trapped rank $\rho_{\stratopp}^{\fsink}$ and targeted rank $\rho_{\stratopp}^{\tsink}$
         by~\Cref{lem:forest}.
  \State Accept if $\rho_{\stratopp}^{\fsink}>\rho_{\stratopp}^{\tsink}$, otherwise reject.
\end{algorithmic}
\end{algorithm}

\begin{theorem}
\label{thm:np-upper-bound}
The limit-sure reachability problem is $\mathsf{NP}$-complete.
\end{theorem}

\begin{proof}
Hardness follows from~\cite{dAF03,BHMST26}.
For the upper bound, we use~\Cref{algo:limit-sure-reachability} with $B=N^{N/2}$.

\emph{Running time.}
Each rank has $O(N\log(N+1))$ bits.
The guessed response has one action at each nonroot.
The graph and encoded ranks have polynomial size.
There are at most $|U|\,|A_{\opp}|$ stability checks.
Every check takes polynomial time by
\Cref{lem:np-bellman-test}.

\emph{Soundness.}
If a run accepts, the preprocessing conditions or
\Cref{lem:np-one-policy-certificate} give a Puiseux-winning strategy.
In the latter case, all residuals are nonnegative and the final
forest comparison holds.
The team wins limit surely by~\Cref{lem:puiseux}.

\emph{Completeness.}
If the team wins limit surely, \Cref{lem:small-ranks} gives winning
ranks bounded by $B$.
A run guesses them.
Their normalised-power realisation is winning because only ranks
matter.
It either accepts during preprocessing or accepts with a stable response by \cref{lem:np-one-policy-certificate}.
\end{proof}

\section{Safety}\label{sec:safety}
We now consider the safety objective $\obj=\Box C$, where
$C\subseteq\states$ is the set of safe states.
Unlike for reachability, we show limit-sure and almost-sure winning are
equivalent.
This extends the two-player result~\cite{CdAH12} to teams without
shared randomness.
The two-player proof constructs a
memoryless almost-sure winning optimal strategy from a fixed-point characterisation of the
game value. We instead use a compactness argument on the space of team
strategies. 
We only show that an almost-sure winning strategy exists if there exists a limit-sure winning strategy.

We equip $\Strats{\team}^{\mathrm h}$ with the product topology
on the action probabilities $\stratp(h)(a)$, which makes it a compact set.
The proof uses compactness to obtain a strategy
that meets all safety thresholds below~$1$.
We give the necessary topological results
in the full version.

\begin{theorem}
\label{thm:safety-limit-almost}
For a game with a safety objective, the team wins limit surely
if and only if it wins almost surely.
\end{theorem}

\begin{proof}[Proof sketch]
Almost-sure winning implies limit-sure winning.
For the converse, suppose that the team wins limit surely.
For each integer $m\geq2$, let $K_m$ be the set of team strategies
that guarantee safety probability at least $1-1/m$ against every
opponent strategy.
Each $K_m$ is nonempty by limit-sure winning.
It is closed because it is an intersection of closed sets,
one for each opponent strategy.
The sets decrease as $m$ increases.
Compactness therefore gives a strategy
$\stratteam^*\in\bigcap_{m\geq2}K_m$.
Against every opponent strategy, this strategy guarantees
safety probability at least $1-1/m$ for every $m\geq2$.
The probability must be $1$, so $\stratteam^*$ is almost-sure winning.
\end{proof}

Almost-sure and sure winning also coincide for
safety~\cite{BHMST26}.
Hence all three winning conditions are equivalent.
Moreover, sure safety can be decided in polynomial time~\cite{BHMST26}.
Thus, the limit-sure safety problem is also decidable in polynomial time.

\section{Individually Randomised ATL}\label{sec:IRATL}
Brice et al.~\cite[Section~6]{BHMST26} introduced \emph{Individually Randomised ATL} (IRATL) to express what coalitions can achieve with individual or shared sources of randomness. We recall its qualitative fragment and apply our results to its model-checking problem.

Let $G$ be a game structure without a fixed team, initial state or objective, and let $L\colon\states\to 2^Q$ label its states with atomic propositions from $Q$. State formulas and path formulas are given by
\[
\begin{aligned}
 \varphi &::= q\mid\neg\varphi\mid\varphi\vee\varphi
       \mid\langle\!\langle C\rangle\!\rangle_w^r\theta,\\
 \theta &::= \bigcirc\varphi\mid\Box\varphi
       \mid\varphi\,\mathsf U\,\varphi,
\end{aligned}
\]
where $q\in Q$, $C\subseteq\players$, $r\in\{\mathsf{ind},\mathsf{sh}\}$, and $w\in\{\mathsf{sure},\mathsf{almost},\mathsf{limit}\}$.

The logic as introduced also has a quantitative fragment that uses thresholds. However, we restrict our attention to the qualitative fragment with sure, almost-sure, and limit-sure winning.

The formula $\langle\!\langle C\rangle\!\rangle_w^r\theta$ holds at $s$ if the coalition $C$ wins the objective described by $\theta$ from $s$ with winning condition $w$ and randomisation type $r$.
For $r=\mathsf{ind}$, coalition strategies have the product form of \cref{def:strategies}; for $r=\mathsf{sh}$, the coalition chooses a distribution over $A_C$. Semantically, almost-sure and limit-sure winning are as in \cref{def:almost-sure,def:limit-sure}, and sure winning requires every play compatible with the chosen strategies and the transition supports to
satisfy the objective against every response of the opponents, even if the play has probability $0$.

On a play $s^0s^1\cdots$, the semantics of $\bigcirc\varphi$ requires $\varphi$ at $s^1$. The semantics of $\Box\varphi$ requires that $\varphi$ holds at every state, and $\varphi\,\mathsf U\,\psi$ requires $\psi$ at some $s^k$ and $\varphi$ at every $s^j$ with $j<k$. %

\begin{restatable}{theorem}{iratlTheorem}
\label{thm:iratl-model-checking}
Model checking qualitative IRATL on explicitly represented finite stochastic game structures is $\P^{\NP}$-complete. 
\end{restatable}
\begin{proof}[Proof sketch]
For each state subformula, we recursively compute the set of states where it is true.
Once these sets are known, formulas of the form $\bigcirc\varphi$ and $\Box\varphi$ can be checked in polynomial time. 
Evaluating if a formula of the form $\varphi_1\,\mathsf U\,\varphi_2$ holds reduces to solving reachability 
by making target states and states containing $\neg\varphi_1$ absorbing. 
For both shared and individual randomisation, sure reachability is in $\P$.
Similarly,  almost-sure reachability with individual randomisation is in $\NP$~\cite{BHMST26}.
Limit-sure reachability is in $\NP$ by \cref{thm:np-upper-bound}. 
Moreover, shared randomisation reduces to the classical two-player games that can be solved in~$\P$. There are at most $|\states|\cdot|\operatorname{Sub}(\varphi)|$ pairs of states and subformulas.
Therefore, they can be solved by a deterministic polynomial-time algorithm with an $\NP$ oracle. The hardness follows from an involved reduction from \emph{sequentially nested satisfiability} (SNSAT), which is $\P^\NP$-complete~\cite{LMS01}.
In fact, the hardness reduction requires only a team of size 2, and it holds even for game structures with deterministic transitions and for the fragment restricted to the sure and almost-sure modalities of the logic.
We give the details in the full version.
\end{proof}

\section{Conclusions}\label{sec:conclusions}
We showed that limit-sure reachability for teams without shared
randomness is in $\NP$, and hence $\NP$-complete, and that for safety
objectives limit-sure, almost-sure, and sure winning coincide and can be
decided in polynomial time. 

We focus now on some future directions of research.
Firstly, for two players, limit-sure winning is well understood for B\"uchi, co-B\"uchi and parity
conditions~\cite{10.1145/1970398.1970404}. Our proof for reachability, however, relies on memoryless team strategies being sufficient~\cite{BHMST26}, and we
do not know whether this is true for other objectives.

Secondly, there are models between full sharing and no sharing: subteams
that share a source of randomness, a correlation device used only
before the play~\cite{Aumann74,CG18}, a few shared random
bits~\cite{Newman91}, or a full structure of shared randomness~\cite{lics26}.

Thirdly, the certificate of \cref{sec:np-upper-bound} may also be useful outside team games. Perturbing ranks so that all forests get distinct weights,
guessing a stable response, and checking residual signs replace a check over exponentially many opponent
strategies by a polynomial-time test. Similar checks arise in the analysis of perturbed Markov chains~\cite{Solan2003}, and in
limit-sure problems for POMDPs with bounded memory~\cite{AsadiEtAl2025}, in particular once an adversary is added.

Finally, as future work, we propose extending PRISM-games~\cite{PRISMgames20} and MCMAS~\cite{MCMAS17}, which still merge a coalition into one player. As our certificates are small and checked
in polynomial time, these tools could support individually randomised coalitions by using a SAT or ILP solver to obtain the ranks using our proposed techniques.

\bibliographystyle{ACM-Reference-Format}
\bibliography{biblio}

\clearpage
\onecolumn
\appendix

\section{Preliminaries}\label{app:sec:prelims}
\begin{lemma}
  \label{lem:single-opponent}
  Let $\obj$ be a reachability or safety objective.
  For every team strategy
  $\stratteam\in\Strats{\team}^{\mathrm h}$,
  \[
    \inf_{\stratrest\in
      \Strats{\players\setminus\team}^{\mathrm h}}
      \Prb{\sinit}{\stratteam,\stratrest}(\obj)
    =
    \inf_{\stratopp\in\Strats{\opp}^{\mathrm h}}
      \Prb{\sinit}{\stratteam,\stratopp}(\obj).
  \]
\end{lemma}

\begin{proof}
Fix $\stratteam$.
Let $\mathcal D^{\mathrm h}$ be the set of deterministic
strategies $d\colon\states^{+}\to A_{\opp}$.
Each $d\in\mathcal D^{\mathrm h}$ can be implemented by the
original opponents without shared randomness.
At each history, each opponent plays its component of $d(h)$.
Conversely, the single opponent can implement any tuple of
strategies of the original opponents by using their product
distribution.
Hence,
\begin{align*}
  \inf_{\stratopp\in\Strats{\opp}^{\mathrm h}}
    \Prb{\sinit}{\stratteam,\stratopp}(\obj)
  &\leq
  \inf_{\stratrest\in
    \Strats{\players\setminus\team}^{\mathrm h}}
    \Prb{\sinit}{\stratteam,\stratrest}(\obj)\\
  &\leq
  \inf_{d\in\mathcal D^{\mathrm h}}
    \Prb{\sinit}{\stratteam,d}(\obj).
\end{align*}

For the reverse inequality, fix
$\stratopp\in\Strats{\opp}^{\mathrm h}$.
Independently for each history $h$, sample $d(h)$ according
to $\stratopp(h)$.
Since the set of histories is countable, these choices define
a product probability measure $\mu$ on $\mathcal D^{\mathrm h}$.
Each history is visited at most once along a play.
Thus, averaging over $d$ gives the same probability for every
finite state sequence as playing $\stratopp$.

For reachability, let $\obj_n$ require a visit to the target
by time $n$.
For safety, let $\obj_n$ require that all states through
time $n$ are safe.
In both cases, the probabilities of $\obj_n$ converge to
those of $\obj$.
Therefore,
\begin{align*}
  \Prb{\sinit}{\stratteam,\stratopp}(\obj)
  &= \lim_{n\to\infty}
     \Prb{\sinit}{\stratteam,\stratopp}(\obj_n)\\
  &= \lim_{n\to\infty}
     \int_{\mathcal D^{\mathrm h}}
       \Prb{\sinit}{\stratteam,d}(\obj_n)\,d\mu(d)
     &&\text{(equality on finite prefixes)}\\
  &= \int_{\mathcal D^{\mathrm h}}
       \Prb{\sinit}{\stratteam,d}(\obj)\,d\mu(d)
     &&\text{(bounded convergence)}\\
  &\geq
     \inf_{d\in\mathcal D^{\mathrm h}}
       \Prb{\sinit}{\stratteam,d}(\obj)
     &&\text{(an average is at least the infimum).}
\end{align*}
Taking the infimum over $\stratopp$ proves the reverse
inequality and hence the claim.
\end{proof}

\section{Reachability}\label{app:sec:reach}
We now focus on the reachability objective $\obj=\Reach{\target}$. We assume that each target state is absorbing, that is, $\delta(t, (\bar{a}, \bar{b}))(t)=1$ for all $t\in\target$ and $(\bar{a}, \bar{b})\in A_{\players}$. This assumption is without loss of generality, because a play satisfies $\Reach{\target}$ as soon as it visits $\target$. 

\begin{definition}[Limit-sure reachability problem]
  \label{app:def:problem}
  Given a game $\game$ with a reachability objective, decide whether the team wins limit surely.
\end{definition}

Brice et al.~\cite{BHMST26} show that deciding whether the team wins almost surely is $\mathsf{NP}$-complete for reachability objectives. 
They prove the $\mathsf{NP}$-hardness by a reduction from the $k$-clique problem. In the games it produces, the team either wins almost surely or has value at most $3/4$~\cite[Theorem~4]{BHMST26}. The same reduction therefore shows that the limit-sure reachability problem is $\mathsf{NP}$-hard.

We show in four steps that the problem belongs to $\mathsf{NP}$. First, we represent the family of team strategies that limit-sure winning requires by a single strategy whose probabilities are series in an infinitesimal $\epsilon$, which are called Puiseux series. Second, we show that whether this strategy wins depends only on the exponents of the leading terms of its probabilities, which we call \emph{ranks}. Third, we show that once both strategies of the team and the opponent are fixed, deciding whether the resulting Markov chain reaches $\target$ with probability infinitesimally close to $1$ takes polynomial time, by comparing directed minimum-weight spanning forests. Finally, we show that a certificate consisting of the team ranks and an optimal opponent strategy in an encoded structure has polynomial size and can be checked in polynomial time. Therefore, solving the main problem is in $\NP$.

First, we build on a result of Brice et al.~\cite{BHMST26}, which shows that memoryless team strategies suffice.

\begin{proposition}[{\cite[Theorem~1]{BHMST26}}]
  \label{app:prop:memoryless}
  Let $\game$ be a game with objective $\Reach{\target}$, and let $t \in [0,1]$. If some team strategy ensures that the play reaches $\target$ with probability greater than $t$ against every tuple of strategies of the opponents, then some memoryless team strategy does.
\end{proposition}

By \cref{app:prop:memoryless}, restricting the team to memoryless
strategies leaves the value unchanged.
Fixing such a strategy gives the MDP of \cref{def:mdp}.
For the reachability objective, this MDP has an optimal
deterministic memoryless opponent strategy.
Thus, in the rest of this section, we restrict attention to
memoryless team strategies and opponent responses in $D$.

\subsection{Puiseux Strategies}
\label{app:sec:puiseux}

Limit-sure winning allows the team strategy to depend on the required accuracy.
In \cref{example:introexample}, each agent crosses with probability
$\epsilon$ and hides with probability $1-\epsilon$. The team's
winning probability approaches $1$ as $\epsilon$ tends to zero.
We represent such a family by a single symbolic strategy whose probabilities
depend on a positive parameter $\epsilon$ standing for ``an arbitrarily small
positive number''. Puiseux series make this idea precise. Their use in the
analysis of stochastic games goes back to Bewley and Kohlberg~\cite{BewleyKohlberg76}.

\begin{definition}[Puiseux series]
  \label{app:def:puiseux-series}
  A \emph{Puiseux series} is a series $x = \sum_{i \ge i_0} c_i \epsilon^{i/m}$ with $i_0 \in \mathbb{Z}$, $m \in \mathbb{N}_{>0}$, and $c_i \in \mathbb{R}$ for each integer $i$. It is \emph{algebraic} if $Q(\epsilon, x) = 0$ for some nonzero polynomial $Q \in \mathbb{R}[\epsilon,X]$.
\end{definition}

In other words, a Puiseux series may have fractional exponents with a common denominator, but only finitely many negative exponents. For example, $x=1-\epsilon^{1/2}$ is a Puiseux series with the fractional exponent $1/2$. It is algebraic because it satisfies $(1-x)^2-\epsilon=0$. In our setting, such series describe probabilities that depend on a positive parameter tending to zero.

Let $\K$ denote the field of algebraic Puiseux series. We write a nonzero $x\in\K$ as $x=\sum_{i\geq i_0}c_i\epsilon^{i/m}$ with $c_{i_0}\neq0$, and call $c_{i_0}$ its \emph{leading coefficient}. We order $\K$ by
\[
  x>0\iff c_{i_0}>0,
  \qquad
  x<y\iff y-x>0.
\]
This order extends the order on $\mathbb{R}$ and makes $\epsilon$ a positive \emph{infinitesimal}, that is, $0<\epsilon<t$ for every real $t>0$.
For $x\in\K$ bounded in absolute value by a real number, its
\emph{standard part} is $\st(x)=c_0\in\mathbb{R}$, with $c_0=0$ if
$\epsilon^0$ does not occur in $x$.

The following classical properties justify reasoning with these symbolic quantities~\cite{BPR06}. Here a \emph{first-order formula} in the language of ordered fields is built from polynomial equalities and inequalities using logical connectives and quantifiers over field elements. A \emph{sentence} is such a formula with no free variables. An ordered field is \emph{real closed} if every positive element has a square root and every polynomial of odd degree has a root.

\begin{lemma}[Puiseux series properties]
  \label{app:lem:puiseux-calculus}
  \begin{enumerate}[label=(\roman*), itemsep=1pt]
    \item \emph{Ordered field and standard parts:} with the order defined above, $\K$ is a real closed ordered field extending $\mathbb{R}$. For every $x\in\K$ bounded in absolute value by a real number, $\st(x)$ is the unique $a\in\mathbb{R}$ such that $|x-a|<t$ for every real $t>0$.
    \item \emph{Transfer:} every first-order sentence in the language of ordered fields, with real parameters, holds over $\mathbb{R}$ if and only if it holds over $\K$.
    \item \emph{Evaluation:} every $x\in\K$ converges at $\epsilon=r$ for all sufficiently small real $r>0$, and we write $x(r)\in\mathbb{R}$ for its value. Let $\Phi$ be a first-order formula with real parameters and $\bar x=(x_1,\ldots,x_n)\in\K^n$. If $\Phi(\bar x)$ holds over $\K$, then $\Phi(\bar x(r))$ holds over $\mathbb{R}$ for all sufficiently small real $r>0$, where $\bar x(r)=(x_1(r),\ldots,x_n(r))\in\mathbb{R}^n$.
  \end{enumerate}
\end{lemma}

Real closedness is the property of $\K$ that makes transfer in part~(ii) hold. For our purposes, transfer lets us replace arbitrarily accurate real strategies by one symbolic strategy. Evaluation then turns that symbolic strategy back into real strategies by choosing a sufficiently small positive parameter.

\begin{definition}[Puiseux team strategy]
  \label{app:def:puiseux-strategy}
  A \emph{Puiseux team strategy} is a memoryless team strategy with probabilities in $\K$ rather than $\mathbb{R}$. It is a tuple $\stratteam = (\stratp)_{p \in \team}$ with $\stratp \colon \states \to \K^{A_p}$, $\stratp(s)(a) \ge 0$ for all $s\in\states$ and $a\in A_p$, and $\sum_{a\in A_p} \stratp(s)(a) = 1$ for all $s\in\states$.
\end{definition}

A \emph{Puiseux Markov chain} on $V$ is a Markov chain whose transition matrix is a map $P\colon V\times V\to\K$. Its reachability probabilities $\Prb{u}{P}(\Reach{R})$ are the least nonnegative solution of \eqref{eq:reach-system} over $\K$, which exists by \cref{app:lem:puiseux-calculus}(ii).

\begin{definition}[Puiseux-winning]
  \label{app:def:winning}
  A Puiseux team strategy $\stratteam$ is
\emph{Puiseux-winning} if
\[
  \st\bigl(\Prb{\sinit}{\stratteam,\stratopp}(\Reach{\target})\bigr)=1
\]
for every real-valued memoryless opponent strategy
$\stratopp\in D$.
\end{definition}

The following lemma reduces limit-sure winning to the existence of a single symbolic strategy.

\begin{lemma}
  \label{app:lem:puiseux}
  The team wins limit surely if and only if some Puiseux team strategy is Puiseux-winning.
\end{lemma}

\begin{proof}
  Fixing a memoryless team strategy leaves the opponent with a
  finite MDP, so responses in $D$ are sufficient. The same holds over $\K$ by
  \cref{app:lem:puiseux-calculus}(ii), using the first-order description of
  reachability probabilities by \eqref{eq:reach-system}.

  Suppose the team wins limit surely. By \cref{app:prop:memoryless}, the sentence
  \[
    \forall\lambda\in(0,1)\;
    \exists\stratteam\in\Strats{\team}\;
    \forall\stratopp\in D:\;
    \Prb{\sinit}{\stratteam,\stratopp}(\Reach{\target})>1-\lambda
  \]
  holds over $\mathbb{R}$. It is a first-order sentence, since strategy
  validity is expressed by polynomial constraints, reachability
  probabilities by the coordinatewise least nonnegative solution of
  \eqref{eq:reach-system}, and $D$ is finite.
  By \cref{app:lem:puiseux-calculus}(ii), it also holds over $\K$.
  Taking $\lambda=\epsilon$, we obtain a Puiseux team strategy
  $\stratteam$ satisfying
  \[
    1-\epsilon
    <\Prb{\sinit}{\stratteam,\stratopp}(\Reach{\target})
    \leq1
    \quad\text{for every }\stratopp\in D.
  \]
  These probabilities
  have standard part $1$ by \cref{app:lem:puiseux-calculus}(i), so
  $\stratteam$ is Puiseux-winning.

  Conversely, let $\stratteam$ be a Puiseux-winning team
  strategy and fix a real $\lambda\in(0,1)$. By
  \cref{app:lem:puiseux-calculus}(i),
  $\Prb{\sinit}{\stratteam,\stratopp}(\Reach{\target})>1-\lambda$ for
  every $\stratopp\in D$. Since $D$ is finite,
  \cref{app:lem:puiseux-calculus}(iii) gives a real $r>0$ such that the
  evaluation $\stratteam(r)\in\Strats{\team}$ satisfies
  \[
    \Prb{\sinit}{\stratteam(r),\stratopp}(\Reach{\target})>1-\lambda
    \quad\text{for every }\stratopp\in D.
  \]
  As responses in $D$ suffice and $\lambda$ was arbitrary,
  the team wins limit surely.
\end{proof}

Now we show that fractional exponents can also be eliminated.

\begin{corollary}
  \label{app:cor:integer-exponents}
  If the team wins limit surely, then some Puiseux-winning team strategy has every probability $\stratp(s)(a)$ of the form $\sum_{i} c_i\, \epsilon^{r_i}$ with coefficients $c_i\in\mathbb{R}$ and exponents $r_i\in\mathbb{N}$.
\end{corollary}

\begin{proof}
  Let $\stratteam$ be a Puiseux-winning team strategy, which exists by \cref{app:lem:puiseux}. We write $x_1, \dots, x_k\in\K$ for its probabilities $\stratp(s)(a)$ with $p \in \team$, $s \in \states$, and $a \in A_p$.

  No exponent is negative. Indeed, if $x_j \ne 0$ had leading exponent $q\in\mathbb{Q}$ with $q < 0$, then $x_j = \epsilon^{q}(c + y)$ with $c\in\mathbb{R}_{>0}$ and $y\in\K$ infinitesimal, and $\epsilon^{q}$ is the inverse of a positive infinitesimal, so $x_j$ would be larger than every real, contradicting $x_j \le 1$. Since exponents increase from the leading one, all exponents of $x_j$ are nonnegative.

  We choose $d\in\mathbb{N}_{>0}$ to be a common multiple of the denominators of the exponents appearing in $x_1, \dots, x_k$, so every exponent of every entry has the form $i/d$ with $i \in \mathbb{N}$. Substituting $\epsilon \coloneqq \eta^{d}$ turns each $x_j$ into a series in $\eta$ with nonnegative integer exponents. This substitution is an order-preserving field bijection of $\K$ that coincides with the identity on $\mathbb{R}$, so the transformed strategy is still Puiseux-winning.
  Renaming $\eta$ as $\epsilon$ gives the strategy required.
\end{proof}

\subsubsection{Ranks}
\label{app:sec:ranks}

Throughout this subsection, we consider Puiseux team
strategies with nonnegative integer exponents, as in
\cref{app:cor:integer-exponents}.

\begin{definition}[Rank]
  \label{app:def:rank}
  The \emph{rank} of a nonzero
  $x=\sum_{i\geq i_0}c_i\epsilon^{i}\in\K$ with $c_{i_0}\neq0$ is its
  smallest exponent $\rk(x)=i_0\in\mathbb{Z}$, and $\rk(0)=\infty$. For a
  Puiseux team strategy $\stratteam$, the rank of an action $a\in A_p$ of a
  player $p\in\team$ at a state $s\in\states$ is
  $r_{p,s,a}=\rk(\stratp(s)(a))$. In a Puiseux Markov chain $P$ on $V$, the
  rank of a transition $(u,v)\in V\times V$ is $\rk(P(u,v))$, and the rank
  of a set $F\subseteq V\times V$ of transitions is
  $\rk(F)=\sum_{(u,v)\in F}\rk(P(u,v))$.
\end{definition}

The rank measures how fast a probability vanishes as $\epsilon\to0$. We
suppose $\min\emptyset=\infty$ and $\infty\pm k=\infty$ for
$k\in\mathbb{Z}$. For all $x,y\in\K$,
\begin{align*}
  \rk(xy)&=\rk(x)+\rk(y),\\
  \rk(x+y)&=\min\{\rk(x),\rk(y)\}&&\text{if }x,y\geq0,\\
  \rk(x/y)&=\rk(x)-\rk(y)&&\text{if }y\neq0.
\end{align*}
The first rule holds because leading coefficients multiply, and the
second because positive leading coefficients cannot cancel. The third
follows from the first, since $(x/y)\,y=x$. For all $p\in\team$ and
$s\in\states$, we have $\sum_{a\in A_p}\stratp(s)(a)=1$, so the sum rule
gives $\min_{a\in A_p}r_{p,s,a}=0$. Conversely, let
$r_{p,s,a}\in\mathbb N\cup\{\infty\}$ for $p\in\team$, $s\in\states$, and
$a\in A_p$ satisfy $\min_{a\in A_p}r_{p,s,a}=0$ for all $p\in\team$ and
$s\in\states$. The Puiseux team strategy
\begin{equation}
  \label{app:eq:realisation}
  \stratp(s)(a)=
  \begin{cases}
    \epsilon^{r_{p,s,a}}\big/\sum_{a'\colon r_{p,s,a'}<\infty}\epsilon^{r_{p,s,a'}}
      & \text{if }r_{p,s,a}<\infty,\\
    0 & \text{if }r_{p,s,a}=\infty
  \end{cases}
\end{equation}
then has action ranks $r_{p,s,a}$, since each denominator has rank $0$.
\Cref{app:ex:ranks-robot} illustrates ranks in the game of
\cref{example:introexample}.

\begin{figure}[h]
  \centering
  \begin{tikzpicture}[x=1cm,y=1cm,font=\footnotesize,
      >={Stealth[length=1.7mm,width=1.5mm]}]
    \tikzset{
      gstate/.style={
        circle,draw,fill=white,
        minimum size=6.6mm,inner sep=0.5pt
      },
      gabs/.style={gstate,double,double distance=0.6pt},
      act/.style={
        circle,fill=black,inner sep=0pt,minimum size=1.5mm
      },
      plab/.style={font=\scriptsize,inner sep=1.2pt},
      alab/.style={font=\scriptsize\itshape,inner sep=1pt}
    }
    \node[anchor=east] at (-0.95,0) {\scriptsize
      \setlength{\tabcolsep}{2.5pt}%
      \begin{tabular}{@{}lcc@{}}
        & cross & hide \\ \midrule
        agent 1
          & $\epsilon$\,\rankbox{1}
          & $1-\epsilon$\,\rankbox{0} \\
        agent 2
          & $\epsilon$\,\rankbox{1}
          & $1-\epsilon$\,\rankbox{0}
      \end{tabular}};
    \node[gstate] (s0)   at (0,0)       {$\sinit$};
    \node[act]    (sc)   at (1.25,0.7)  {};
    \node[act]    (wt)   at (1.25,-0.7) {};
    \node[gabs]   (fail) at (3.6,0.7)   {$\sfail$};
    \node[gabs]   (goal) at (3.6,-0.7)  {$\sgoal$};

    \draw (s0) --
      node[alab,below right,pos=0.55] {$\block_1$} (sc);
    \draw (s0) --
      node[alab,above right,pos=0.55] {$\block_2$} (wt);

    \draw[->] (sc) --
      node[plab,above] {$\epsilon^2$\,\rankbox{2}} (fail);
    \draw[->] (sc) --
      node[plab,fill=white,pos=0.72]
        {$\epsilon(1-\epsilon)$\,\rankbox{1}} (goal);
    \draw[->] (wt) --
      node[plab,fill=white,pos=0.72]
        {$\epsilon^2$\,\rankbox{2}} (fail);
    \draw[->] (wt) --
      node[plab,below]
        {$\epsilon(1-\epsilon)$\,\rankbox{1}} (goal);

    \draw[->] (sc) to[out=135,in=90,looseness=1.2]
      node[plab,above,pos=0.45]
        {$1-\epsilon$\,\rankbox{0}} (s0);
    \draw[->] (wt) to[out=225,in=270,looseness=1.2]
      node[plab,below,pos=0.45]
        {$1-\epsilon$\,\rankbox{0}} (s0);
  \end{tikzpicture}
  \caption{The Puiseux-winning team strategy of
    \cref{app:ex:ranks-robot} (left) and the MDP it induces
    for the opponent (right). Boxes show ranks, and black
    dots are the actions of the opponent.}
  \Description{On the left, a table gives the team strategy:
    each of the two agents crosses with probability epsilon,
    of rank 1, and hides with probability 1 minus epsilon,
    of rank 0. On the right, the induced MDP: from the initial
    state the opponent chooses to scan the first or second
    robot, with the two actions drawn as black dots.
    Under either action, the play moves to the failure state
    with probability epsilon squared (rank 2), to the goal
    state with probability epsilon times 1 minus epsilon
    (rank 1), and back to the initial state with probability
    1 minus epsilon (rank 0). The goal and failure states
    are both double-circled.}
  \label{app:fig:ranks}
\end{figure}
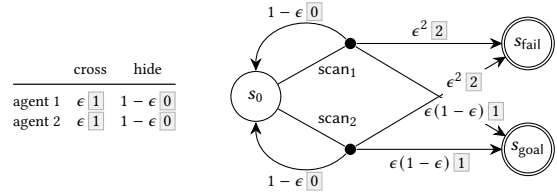

\begin{example}[Ranks of actions and transitions]
  \label{app:ex:ranks-robot}
  In the game of \cref{example:introexample}, let each agent
  cross with probability $\epsilon$ and hide with probability
  $1-\epsilon$ at $\sinit$.
  Crossing then has rank $1$ and hiding has rank $0$.
  \Cref{app:fig:ranks} shows this team strategy and the MDP
  it leaves for the opponent.
  Since the agents randomise independently, each transition
  probability is a sum of products of action probabilities,
  and the product and sum rules give its rank.
  For either opponent action $\block_i$, the play reaches
  $\sfail$ only if both agents cross, so this transition has
  rank $\rk(\epsilon\cdot\epsilon)=1+1=2$.
  It reaches $\sgoal$ only if the scanned agent hides and
  the other crosses, giving rank
  $\rk((1-\epsilon)\epsilon)=0+1=1$.
  It stays at $\sinit$ if both agents hide or only the scanned
  agent crosses, giving rank
  $\rk\bigl((1-\epsilon)^2+\epsilon(1-\epsilon)\bigr)
  =\min\{0,1\}=0$.
  Thus, success has smaller rank than failure.
  Although most rounds return to $\sinit$, success occurs
  before failure with probability $1-\epsilon$, which
  approaches $1$ as $\epsilon$ goes to zero.
\end{example}

To show that only the ranks matter, we express the absorption probabilities of a Markov chain through spanning forests of its support graph. The \emph{support graph} of a Markov chain on $V$ with transition matrix $P$
has vertex set $V$ and edge set $\{(u,v)\in V\times V\mid P(u,v)>0\}$. For a set $F$ of edges with at most one outgoing edge per vertex, a vertex $u$ \emph{leads to} a set $R$ of vertices if following the edges of $F$ from $u$ ends in $R$, and to a vertex $v$ if it leads to $\{v\}$.

\begin{definition}[Spanning forest]
  \label{app:def:spanning-forest}
  Let $P$ be a finite Markov chain on $V$ with a nonempty set $R\subseteq V$ of absorbing states, which we call \emph{roots}. A \emph{spanning forest} is a set $F\subseteq V\times V$ of edges $(u,v)$ with $P(u,v)>0$, with exactly one outgoing edge from each vertex in $V\setminus R$ and none from vertices in $R$, such that every vertex leads to $R$. Its weight is defined as $\wt(F) = \prod_{(u,v) \in F} P(u,v)$. We write $\Forests$ for the set of spanning forests and $\Forests_{R'}(u)\subseteq\Forests$ for those in which $u\in V$ leads to $R'\subseteq R$.
\end{definition}

\begin{lemma}[Functional forests]
  \label{app:lem:functional-forest}
  Let $V$ be a finite set of vertices, let $\emptyset\neq R\subseteq V$, and
  let $F$ be a set of edges between vertices of $V$ with exactly one edge
  leaving each vertex of $V\setminus R$ and none leaving $R$. Every vertex
  leads to $R$ if and only if $F$ has no cycle. In that case, the graph $(V,F)$ has $|R|$ connected components, each a tree with exactly one vertex in $R$.
\end{lemma}

\begin{proof}
  A directed cycle prevents its vertices from leading to $R$. Conversely,
  on a cycle with directions ignored, the restriction to one outgoing edge
  per vertex forces its edges to form a directed cycle. Thus, $F$ has no
  such cycle if and only if it has no directed cycle. Without directed
  cycles, following the edges of $F$ from any vertex ends in $R$, since
  $V$ is finite.

  Without such cycles, $(V,F)$ has $|V|-|F|=|R|$ connected components. Each component contains a
  vertex of $R$, since every vertex leads to $R$. As there are $|R|$
  components, each contains exactly one.
\end{proof}

\begin{lemma}[Forest formula]
  \label{app:lem:forest-formula}
  If every state has a directed path to $R$ in the support graph, then $\sum_{F \in \Forests} \wt(F) > 0$ and
  \[
    \Prb{u}{P}(\Reach{R'}) \;=\; \frac{\sum_{F \in \Forests_{R'}(u)} \wt(F)}{\sum_{F \in \Forests} \wt(F)}
    \qquad \text{for all } u \in V, \ R' \subseteq R.
  \]
\end{lemma}

\begin{proof}
  The lemma is a special case of the classical forest formula~\cite[Chapter~6, Lemma 3.3]{FreidlinWentzell12}. It transfers to $\K$ by \cref{app:lem:puiseux-calculus}.
\end{proof}

\begin{example}
  \label{app:ex:forest-formula}
  Consider the Markov chain on $V=\{1,2,3,4\}$, where $3$ and $4$
  are absorbing and
  \[
    P(1,2)=P(1,3)=\frac12,
    \qquad P(2,1)=\frac34,\quad P(2,4)=\frac14,
  \]
as depicted in \Cref{app:fig:forest-formula-chain}.
  
  \begin{figure}[htbp]
    \centering
    \begin{tikzpicture}[x=1cm,y=1cm,font=\footnotesize,
        >={Stealth[length=1.7mm,width=1.5mm]}]
      \tikzset{
        mcstate/.style={circle,draw,fill=white,minimum size=5.6mm,inner sep=0.5pt},
        mcroot/.style={mcstate,double,double distance=0.6pt},
        tstate/.style={circle,draw,fill=white,minimum size=4.2mm,inner sep=0pt,font=\scriptsize},
        troot/.style={tstate,double,double distance=0.5pt},
        tlab/.style={align=center,font=\scriptsize,inner sep=1pt},
        unused/.style={->,gray!45,line width=0.4pt,>={Stealth[length=1.2mm,width=1mm]}},
        chosen/.style={->,line width=1pt,>={Stealth[length=1.9mm,width=1.7mm]}},
        loopy/.style={->,>={Stealth[length=1.3mm,width=1.1mm]}},
        box/.style={draw=gray!50,fill=gray!7,rounded corners=3pt}
      }
      \node[mcstate] (s1) at (0,0) {$1$};
      \node[mcstate] (s2) at (1.5,0) {$2$};
      \node[mcroot] (s3) at (0,-1.05) {$3$};
      \node[mcroot] (s4) at (1.5,-1.05) {$4$};
      \draw[->] (s1) to[bend left=25] node[above,inner sep=1pt] {$\tfrac12$} (s2);
      \draw[->] (s2) to[bend left=25] node[below,inner sep=1pt] {$\tfrac34$} (s1);
      \draw[->] (s1) -- node[left,inner sep=1pt] {$\tfrac12$} (s3);
      \draw[->] (s2) -- node[right,inner sep=1pt] {$\tfrac14$} (s4);
      \draw[loopy] (s3) to[out=-115,in=-65,min distance=3.5mm] node[below,inner sep=1pt] {$1$} (s3);
      \draw[loopy] (s4) to[out=-115,in=-65,min distance=3.5mm] node[below,inner sep=1pt] {$1$} (s4);
      \foreach \x/\name in {2.7/A,4.33/B,5.96/C} {
        \draw[box] (\x-0.33,0.32) rectangle (\x+1.13,-1.95);
        \node[tstate] (\name1) at (\x,0) {$1$};
        \node[tstate] (\name2) at (\x+0.8,0) {$2$};
        \node[troot]  (\name3) at (\x,-0.75) {$3$};
        \node[troot]  (\name4) at (\x+0.8,-0.75) {$4$};
        \draw[unused] (\name1) to[bend left=25] (\name2);
        \draw[unused] (\name2) to[bend left=25] (\name1);
        \draw[unused] (\name1) -- (\name3);
        \draw[unused] (\name2) -- (\name4);
        \draw[unused] (\name3) to[out=-115,in=-65,min distance=2.5mm] (\name3);
        \draw[unused] (\name4) to[out=-115,in=-65,min distance=2.5mm] (\name4);
      }
      \draw[chosen] (A1) to[bend left=25] (A2);  \draw[chosen] (A2) -- (A4);
      \draw[chosen] (B2) to[bend left=25] (B1);  \draw[chosen] (B1) -- (B3);
      \draw[chosen] (C1) -- (C3);                \draw[chosen] (C2) -- (C4);
      \node[tlab] at (3.1,-1.55)  {$\wt(F)=\tfrac18$\\ $1$ leads to $4$};
      \node[tlab] at (4.73,-1.55) {$\wt(F)=\tfrac38$\\ $1$ leads to $3$};
      \node[tlab] at (6.36,-1.55) {$\wt(F)=\tfrac18$\\ $1$ leads to $3$};
    \end{tikzpicture}
    \caption{The Markov chain of \cref{app:ex:forest-formula} and its three spanning forests.}
    \Description{A Markov chain with transient states 1 and 2 and absorbing roots 3 and 4, drawn with double circles and self-loops of probability 1. State 1 moves to 2 with probability 1/2 and to 3 with probability 1/2. State 2 moves to 1 with probability 3/4 and to 4 with probability 1/4. Next to it are the three spanning forests rooted at 3 and 4, with the chosen edges in bold: first, edges 1 to 2 and 2 to 4, with weight 1/8, so 1 leads to 4; second, edges 2 to 1 and 1 to 3, with weight 3/8, so 1 leads to 3; third, edges 1 to 3 and 2 to 4, with weight 1/8, so 1 leads to 3.}
    \label{app:fig:forest-formula-chain}
  \end{figure}
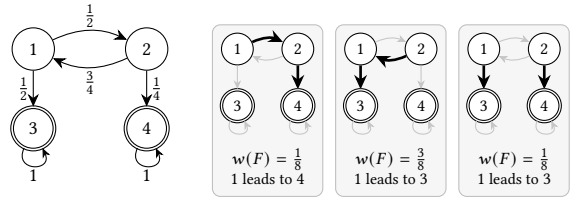
  
  We apply \cref{app:lem:forest-formula} with roots $R=\{3,4\}$ to
  compute the probability of reaching $3$ from state $1$.

  Each spanning forest selects one outgoing edge from each of $1$ and $2$.
  Selecting both $1\to2$ and $2\to1$ creates a cycle, so this choice
  is excluded. The three spanning forests, drawn in \cref{app:fig:forest-formula-chain}, are
  \[
    \begin{array}{c|c|c}
      F & \wt(F) & \text{Root to which }1\text{ leads} \\
      \hline
      \{1\to2,\;2\to4\} & \frac12\cdot\frac14=\frac18 & 4 \\
      \{1\to3,\;2\to1\} & \frac12\cdot\frac34=\frac38 & 3 \\
      \{1\to3,\;2\to4\} & \frac12\cdot\frac14=\frac18 & 3
    \end{array}
  \]
  The denominator sums the weights of all three forests, while the
  numerator sums only those in which $1$ leads to $3$. Hence
  \[
    \Prb{1}{P}(\Reach{\{3\}})
    =\frac{\frac38+\frac18}{\frac18+\frac38+\frac18}
    =\frac45.
  \]
\end{example}

Combining the forest formula with the rank rules gives the rank of an
absorption probability.

\begin{lemma}[Rank criterion]
  \label{app:lem:rank-criterion}
  Let $P$ be a Puiseux Markov chain on $V$ with a set $R\subseteq V$ of
  absorbing states such that every state has a path to $R$ in the support
  graph.
  For all $u\in V$ and $R'\subseteq R$,
  \[
    \rk\bigl(\Prb{u}{P}(\Reach{R'})\bigr)
    =\min_{F\in\Forests_{R'}(u)}\rk(F)-\min_{F\in\Forests}\rk(F).
  \]
\end{lemma}

\begin{proof}
  By \cref{app:lem:forest-formula} and the rank rules after \cref{app:def:rank},
  \begin{align*}
    \rk\bigl(\Prb{u}{P}(\Reach{R'})\bigr)
    &=\rk\Bigl(\textstyle\sum_{F\in\Forests_{R'}(u)}\wt(F)\Bigr)
      -\rk\Bigl(\textstyle\sum_{F\in\Forests}\wt(F)\Bigr)\\
    &=\min_{F\in\Forests_{R'}(u)}\rk(F)-\min_{F\in\Forests}\rk(F).
  \end{align*}
  The first equality uses the quotient rule, since the denominator is
  nonzero. The second uses the sum rule, since each weight
  $\wt(F)=\prod_{(v,w)\in F}P(v,w)$ is nonnegative and has rank $\rk(F)$
  by the product rule.
\end{proof}

\begin{definition}[Strongly connected components]
  A \emph{strongly connected component} (SCC) of a Markov chain $P$ on $V$
  is a maximal nonempty set $C\subseteq V$ such that any two states of $C$
  are reachable from one another in the support graph. An SCC $C$ is
  \emph{bottom} if $P(u,v)=0$ for all $u\in C$ and $v\in V\setminus C$.
\end{definition}

Once a Markov chain enters a bottom SCC, it stays there forever. A finite Markov chain reaches a bottom SCC almost surely. Since target states are absorbing, a strongly connected set of states that contains a target state is that state alone.
Let $\stratteam$ be a Puiseux team strategy and $\stratopp\in D$. We
write $R\subseteq\states$ for the union of the bottom SCCs of
$P_{\stratteam,\stratopp}$, and
$R'=\bigcup\{C\mid C\text{ is a bottom SCC with }C\cap\target=\emptyset\}$.
Making the states of $R$ absorbing gives the Puiseux Markov chain
$P\colon\states\times\states\to\K$ with
\[
  P(u,v)=
  \begin{cases}
    P_{\stratteam,\stratopp}(u,v) & \text{if }u\notin R,\\
    1 & \text{if }u=v\in R,\\
    0 & \text{otherwise}.
  \end{cases}
\]
On the spanning forests of $P$, we consider the condition
\[
  \min_{F\in\Forests_{R'}(\sinit)}\rk(F)
  >
  \min_{F\in\Forests}\rk(F).
  \tag{$\star$}\label{app:eq:rank-criterion}
\]
By \cref{app:lem:rank-criterion}, \eqref{app:eq:rank-criterion} says that the
failure probability $\Prb{\sinit}{P}(\Reach{R'})$ has positive rank.
Equivalently, $\wt(F)/\wt(F^*)$ is infinitesimal for every
$F\in\Forests_{R'}(\sinit)$ and every $F^*\in\Forests$ of minimum rank.

\begin{lemma}
  \label{app:lem:ranks-only}
  The strategy $\stratteam$ is Puiseux-winning if and only if
  \eqref{app:eq:rank-criterion} holds for every $\stratopp\in D$. In
  particular, if a Puiseux team strategy $\stratteam'$ has the same action
  ranks as $\stratteam$, then $\stratteam'$ is Puiseux-winning if and only
  if $\stratteam$ is.
\end{lemma}

\begin{proof}
  As in the proof of \cref{app:lem:puiseux}, it suffices to consider
  opponent strategies in $D$. We fix $\stratopp\in D$
  and define $P$, $R$, and $R'$ as above. The Markov chain reaches a bottom SCC almost
  surely. It therefore avoids $\target$ exactly when it reaches $R'$, so
  \[
    1-\Prb{\sinit}{\stratteam,\stratopp}(\Reach{\target})
    =\Prb{\sinit}{P}(\Reach{R'}).
  \]
  The strategy $\stratteam$ is Puiseux-winning if and only if this
  failure probability has standard part $0$ for every $\stratopp$. By the
  definition of the standard part, a bounded $x\in\K$ satisfies
  $\st(x)=0$ if and only if $\rk(x)>0$. Every state of $P$ has a path to
  $R$, so by \cref{app:lem:rank-criterion},
  $\rk\bigl(\Prb{\sinit}{P}(\Reach{R'})\bigr)>0$ is equivalent to
  \eqref{app:eq:rank-criterion}.

  The action ranks determine the transition ranks of
  $P_{\stratteam,\stratopp}$. Indeed, the team players randomise
  independently, so by the product rule a team profile $\bar a\in A_{\team}$
  at a state $s\in\states$ has rank
  $\rk(\stratteam(s)(\bar a))=\sum_{p\in\team}r_{p,s,a_p}$. A transition
  probability is a sum of such probabilities multiplied by positive reals.
  For all $s,t\in\states$, the sum rule therefore gives
  \begin{equation}
    \label{app:eq:transition-rank}
    \rk\bigl(P_{\stratteam,\stratopp}(s,t)\bigr)
    =\min\left\{
      \sum_{p\in\team}r_{p,s,a_p}
      \;\middle|\;
      \begin{array}{l}
        \bar a\in\Supp(\stratteam(s)),\\
        \delta\bigl(s,(\bar a,\stratopp(s))\bigr)(t)>0
      \end{array}
    \right\}.
  \end{equation}
  The action ranks thus determine the support, the bottom SCCs, $R$,
  $R'$, and both minima in \eqref{app:eq:rank-criterion}.
\end{proof}

\begin{proposition}[Computing forest minima]
  \label{app:prop:mdst-poly}
  Let $P$ be a Puiseux Markov chain on $V$, given by its transition ranks
  $\rk(P(v,w))\in\mathbb N\cup\{\infty\}$, and let $R\subseteq V$
  be a set of absorbing states. For
  $u\in V\setminus R$ and $R'\subseteq R$, \cref{app:algo:forest-min} computes
  $\min_{F\in\Forests_{R'}(u)}\rk(F)$ in polynomial time.
\end{proposition}

\begin{proof}
  We reduce the problem to an instance of the minimum-weight spanning arborescence problem. For this, we
  contract $R'$ into a new vertex $\hat v$, delete the vertices of
  $R\setminus R'$, and redirect every edge into $R\setminus R'$ to $u$.
  A set $F$ of edges of the support graph with one outgoing edge at every
  vertex of $V\setminus R$ becomes a set $\hat F$ with one outgoing edge at
  every vertex other than $\hat v$, and
  \[
    F\in\Forests_{R'}(u)
    \iff
    \text{every vertex leads to } \hat v \text{ in } \hat F.
  \]
  Indeed, if $F\in\Forests_{R'}(u)$ then a path of $F$ ending in
  $R\setminus R'$ continues in $\hat F$ from $u$ and ends in $\hat v$.
  Conversely, if every vertex leads to $\hat v$ in $\hat F$ then the path
  from $u$ uses no redirected edge, since such an edge would close a cycle
  at $u$. Hence, $u$ leads to $R'$ in $F$, and every path of $F$ ends in $R$.
  The correspondence preserves the sum of ranks. \Cref{app:algo:forest-min}
  therefore returns the minimum with one call to an oracle
  \texttt{MinDirectedSpanningForest}, which returns the minimum rank of a
  set $\hat F$ on the right, or $\infty$ if no such set exists. By
  \cref{app:lem:functional-forest}, these sets are the spanning trees directed
  towards $\hat v$. A minimum-weight spanning arborescence
  algorithm~\cite{Edmonds67,ChuLiu65,Tarjan77} on edges thus
  implements the oracle in polynomial time.
\end{proof}

\algtext*{EndIf}
\algtext*{EndFor}

\begin{algorithm}[H]
\caption{Forest Minimum \texttt{ForestMin}}
\label{app:algo:forest-min}
\begin{algorithmic}[1]
    \State \textbf{Input:} $V$, $\rk(P(v,w))$ for $v,w\in V$, $R$,
           $u\in V\setminus R$, $R'\subseteq R$
    \State \textbf{Output:} $\min_{F\in\Forests_{R'}(u)}\rk(F)$

    \State $\hat V\gets(V\setminus R)\cup\{\hat v\}$
    \State $\hat E\gets\emptyset$
    \ForAll{$v\in V\setminus R$ and $w\in V$ with $\rk(P(v,w))<\infty$}
        \If{$w\in R'$}
            \State $w'\gets\hat v$
                \Comment{Contract $R'$}
        \ElsIf{$w\in R\setminus R'$}
            \State $w'\gets u$
                \Comment{Redirect to $u$}
        \Else
            \State $w'\gets w$
        \EndIf
        \State $\hat E\gets\hat E\cup\{e_{v,w}\}$
            \Comment{Edge from $v$ to $w'$}
        \State $\rk(e_{v,w})\gets\rk(P(v,w))$
    \EndFor

    \State \Return $\texttt{MinDirectedSpanningForest}(\hat V,\hat E,\rk,\hat v)$
\end{algorithmic}
\end{algorithm}

Since $\Forests=\Forests_R(\sinit)$ for $\sinit$, both minima in
\eqref{app:eq:rank-criterion} can be computed in polynomial time by \cref{app:algo:forest-min}.

\subsubsection{Small ranks}
\label{app:sec:small-ranks}

By \cref{app:lem:ranks-only}, a Puiseux-winning strategy is described by its
action ranks. To guess such a strategy in nondeterministic polynomial time,
we need ranks with polynomially many bits. We now recall two
linear-algebra facts that help us show that such ranks exist.

\paragraph{Cramer's rule.}
For an invertible matrix $A\in\mathbb{R}^{k\times k}$ and a vector
$b\in\mathbb{R}^k$, the unique solution $x\in\mathbb{R}^k$ of $Ax=b$ satisfies
\[
  x_i=\frac{\det A_i}{\det A}\quad\text{for }1\leq i\leq k,
\]
where $A_i\in\mathbb{R}^{k\times k}$ is obtained from $A$ by replacing its
$i$th column with $b$.
If $A$ and $b$ have integer entries, then $\det A$ is a nonzero integer,
so $\abs{\det A}\geq1$ and $\abs{\det A}\,x\in\mathbb{Z}^k$.

\paragraph{Hadamard's inequality.}
For a matrix $B\in\mathbb{R}^{k\times k}$ with columns
$b_1,\ldots,b_k\in\mathbb{R}^k$,
\[
  \abs{\det B}\leq\prod_{i=1}^k\|b_i\|_2,
\]
where $\|b_i\|_2$ is the Euclidean length of column $b_i$.
In particular, every $B\in\{-1,0,1\}^{k\times k}$ has $\|b_i\|_2\leq\sqrt{k}$
for every $i$ and hence $\abs{\det B}\leq k^{k/2}$. Together with Cramer's
rule, we obtain the following bound. If $A\in\{-1,0,1\}^{k\times k}$ is
invertible and $b\in\{0,1\}^k$, then the solution $x\in\mathbb{R}^k$ of
$Ax=b$ satisfies
\begin{align*}
  \abs{\det A}\,x&\in\mathbb{Z}^k,\\
  \abs{\det A}\,\abs{x_i}&=\abs{\det A_i}\leq k^{k/2}
  \quad\text{for }1\leq i\leq k.
\end{align*}

\begin{lemma}
  \label{app:lem:small-ranks}
  Let $N=\abs{\{(p,s,a)\mid p\in\team,\ s\in\states,\ a\in A_p\}}$. If the team wins limit surely, then there exists a Puiseux-winning team strategy whose finite ranks are nonnegative integers bounded by $N^{N/2}$.
\end{lemma}

\begin{proof}
  We start with a Puiseux-winning team strategy and construct
  another Puiseux-winning strategy with the same support whose finite action ranks
  are bounded by $N^{N/2}$. We write \eqref{app:eq:rank-criterion} as linear
  inequalities in the action ranks, take a vertex of their solution set,
  and scale it to an integer vector satisfying the required bound.

  By \cref{app:lem:puiseux,app:cor:integer-exponents}, we choose without loss of generality a Puiseux-winning strategy
  with integer action ranks. Let
  $I=\{(p,s,a)\mid p\in\team,\ s\in\states,\ a\in A_p,\ \stratp(s)(a)>0\}$
  be its support. We denote its finite ranks by
  \[
    r^*=(r^*_{p,s,a})_{(p,s,a)\in I}\in\mathbb{N}^I.
  \]
  We keep this support fixed and set $k=|I|\leq N$.

  We first express forest ranks as sums of action ranks. We fix
  $\stratopp\in D$ and consider the Markov chain $P$ and the sets $R$ and
  $R'$ defined before \eqref{app:eq:rank-criterion}. The support graph of $P$,
  the sets $R$ and $R'$, and the forest sets $\Forests$ and
  $\Forests_{R'}(\sinit)$ depend only on the support of the team strategy,
  which we keep fixed. Let $F\in\Forests$ be a spanning forest of $P$.
  For each edge $(u,v)\in F$, we choose a supported team profile
  $\alpha(u)\in\Supp(\stratteam(u))$ with
  $\delta\bigl(u,(\alpha(u),\stratopp(u))\bigr)(v)>0$. Since $F$ has
  exactly one edge leaving each state of $\states\setminus R$, these profiles form
  a map $\alpha\colon\states\setminus R\to A_{\team}$, which we call a
  \emph{choice} for $F$.

  For a vector $r=(r_{p,s,a})_{(p,s,a)\in I}\in\mathbb{R}^I$ of candidate
  ranks, we define the \emph{cost} of $F$ and a choice $\alpha$ by
  \[
    C_{\stratopp}(F,\alpha;r)=\sum_{(u,v)\in F}\sum_{p\in\team}r_{p,u,\alpha(u)_p}.
  \]
  If the actions have ranks $r$, then $C_{\stratopp}(F,\alpha;r)$ is the
  rank of the term of $\wt(F)$ in which each edge $(u,v)\in F$ uses the
  profile $\alpha(u)$, and \eqref{app:eq:transition-rank} gives
  $\rk(F)=\min_{\alpha}C_{\stratopp}(F,\alpha;r)$ over all choices $\alpha$
  for $F$. The cost is linear in $r$ since the coefficient of $r_{p,s,a}$ is $1$
  if $s\notin R$ and $\alpha(s)_p=a$, and $0$ otherwise.

  For each $\stratopp\in D$, we choose a forest $F_{\stratopp}^*\in\Forests$ and
  a choice $\alpha_{\stratopp}^*$ for it whose cost is smallest under the
  original ranks $r^*$. By \eqref{app:eq:rank-criterion} in \cref{app:lem:ranks-only},
  every forest in $\Forests_{R'}(\sinit)$ has a strictly larger rank than
  $F_{\stratopp}^*$. Since these ranks are
  integers, the difference is at least $1$.

  We require the rank $r$ to satisfy
  \[
    C_{\stratopp}(F,\alpha;r)
    -C_{\stratopp}(F_{\stratopp}^*,\alpha_{\stratopp}^*;r)
    \geq1
  \]
  for every $\stratopp\in D$, every $F\in\Forests_{R'}(\sinit)$, and every
  choice $\alpha$ for $F$.
  We also require $r\geq0$, and $r_{p,s,a}=0$ whenever $r^*_{p,s,a}=0$.

  The original vector $r^*$ satisfies this system. Every solution $r$
  also satisfies \eqref{app:eq:rank-criterion}. Indeed, with forest ranks
  computed from the action ranks $r$, every $\stratopp\in D$ and
  $F\in\Forests_{R'}(\sinit)$ satisfy
  \begin{align*}
    \rk(F)&=\min_{\alpha}C_{\stratopp}(F,\alpha;r)
          \geq C_{\stratopp}(F_{\stratopp}^*,\alpha_{\stratopp}^*;r)+1\\
          &\geq\rk(F_{\stratopp}^*)+1
          \geq\min_{F'\in\Forests}\rk(F')+1.
  \end{align*}

  The system may have exponentially many constraints, but we use it only
  to prove that bounded ranks exist. We write the full system as
  \[
    Mr\geq d,
    \qquad M\in\{-1,0,1\}^{m\times k},
    \qquad d\in\{0,1\}^{m},
    \qquad k\leq N,
  \]
  where $m$ is the number of inequalities. We include nonnegativity of entries of $r$
  among these inequalities and express each equality as two inequalities.
  Its solution set
  \[
    Q=\{r\in\mathbb{R}^I\mid Mr\geq d\}
  \]
  is a nonempty polyhedron in the nonnegative orthant, so it has a vertex
  $z\in Q$.

  At a vertex, the constraints that hold with equality determine a unique
  point and their coefficient matrix has full rank $k$, so we can select
  $k$ linearly independent rows of $M$. Let $A\in\{-1,0,1\}^{k\times k}$
  consist of these rows and $b\in\{0,1\}^k$ of the corresponding entries
  of $d$, giving
  \[
    Az=b.
  \]
  Thus $A$ is an invertible $k\times k$ matrix.

  By Cramer's rule, $r'=\abs{\det A}\,z\in\mathbb{Z}^I$, and
  \[
    Mr'=\abs{\det A}\,Mz\geq\abs{\det A}\,d\geq d.
  \]
  The multiplier $\abs{\det A}$ is at least $1$ since $\det A$ is a nonzero integer, so every inequality remains satisfied
  and coordinates fixed to zero stay zero. By the bound after Hadamard's
  inequality,
  \[
    0\leq r'_{p,s,a}\leq k^{k/2}\leq N^{N/2}.
  \]

  The minimum rank at each player-state pair remains zero, so
  \eqref{app:eq:realisation}, with $r'_{p,s,a}=\infty$ for $(p,s,a)\notin I$,
  gives a Puiseux team strategy with ranks $r'$. It is Puiseux-winning by
  \cref{app:lem:ranks-only}.
\end{proof}

Since $N$ is polynomial in the size of the game and $\log_2(N^{N/2})=\frac{N}{2}\log_2 N$, ranks can be encoded with polynomially many bits. Guessing these ranks and checking \eqref{app:eq:rank-criterion} for every $\stratopp\in D$ gives membership in $\mathsf{NP}^{\mathsf{NP}}$ as \eqref{app:eq:rank-criterion} can be checked in polynomial time with \cref{app:algo:forest-min}. In the following sections, we improve the bound to $\mathsf{NP}$.

\subsection{Puiseux Markov Decision Processes}
\label{app:sec:puiseux-mdp}

By \cref{app:prop:mdst-poly}, we can check \eqref{app:eq:rank-criterion} for a
single $\stratopp\in D$ in polynomial time. However, the chain $P$ in
\eqref{app:eq:rank-criterion} depends on $\stratopp$ through its bottom SCCs,
and $D$ can be exponentially large. We therefore build one object for all
$\stratopp\in D$ from the MDP induced by $\stratteam$. 

\subsubsection{The induced Puiseux MDP}

We fix action ranks $r_{p,s,a}\in\mathbb N\cup\{\infty\}$ with
$\min_{a\in A_p}r_{p,s,a}=0$ for all $p\in\team$ and $s\in\states$, and let
$\stratteam$ be the Puiseux team strategy \eqref{app:eq:realisation}.

\begin{definition}[Induced Puiseux MDP]
  \label{app:def:induced-mdp}
  The \emph{Puiseux MDP induced by $\stratteam$} is the MDP $\imdp$ of
  \cref{def:mdp} for $\stratteam$, with transition function $P=P_{\stratteam}\colon\states\times A_{\opp}\to\K^{\states}$
  and transition ranks $\rk(s,b,t)=\rk(P(s,b)(t))$ for $s,t\in\states$ and
  $b\in A_{\opp}$.
\end{definition}

These ranks are given by \eqref{app:eq:transition-rank} with $b$ in place of
$\stratopp(s)$, so we can compute them in polynomial time. For
$\stratopp\in D$, we write
$\Prb{u}{\stratopp}=\Prb{u}{\stratteam,\stratopp}$ for the probabilities in
$\imdp$ after fixing $\stratopp$.

\begin{lemma}[Zero-rank transitions]
  \label{app:lem:zero-edges}
  For all $s\in\states$ and $b\in A_{\opp}$, there is $t\in\states$
  with $\rk(s,b,t)=0$.
\end{lemma}

\begin{proof}
  Since $\min_{a\in A_p}r_{p,s,a}=0$ for all $p\in\team$, some
  $\bar a^0\in A_{\team}$ has $\sum_{p\in\team}r_{p,s,a^0_p}=0$, so every
  $t\in\Supp(\delta(s,(\bar a^0,b)))$ satisfies $\rk(s,b,t)=0$ by
  \eqref{app:eq:transition-rank}.
\end{proof}

\subsubsection{Traps and the merged MDP}
\label{app:sec:traps}\label{app:sec:forests}

\begin{definition}[End component]
  \label{app:def:end-component}
  An \emph{end component} of $\imdp$ is a pair $(X,\Gamma)$ with
  $\emptyset\neq X\subseteq\states$ and $\emptyset\neq\Gamma(x)\subseteq A_{\opp}$
  for $x\in X$ such that $\Supp(P(x,b))\subseteq X$ for all $x\in X$ and
  $b\in\Gamma(x)$, and the graph
  $\bigl(X,\{(x,y)\mid P(x,b)(y)>0\text{ for some }b\in\Gamma(x)\}\bigr)$
  is strongly connected. It is a
  \emph{maximal end component} (MEC) if $X$ and every $\Gamma(x)$ are inclusion-maximal.
\end{definition}

\begin{lemma}[Computing MECs]
  \label{app:lem:mec-support}
  The MECs of $\imdp$ can be computed in polynomial time.
\end{lemma}

\begin{proof}
  End components are defined through the support only. We can therefore
  run the polynomial-time MEC algorithms for
  MDPs~\cite{CY95,deAlfaro97,CH14}.
\end{proof}

\begin{definition}[Trap]
  \label{app:def:trap}
  We write
  \[
    \Trap=\bigcup\{X\mid (X,\Gamma)\text{ is a MEC of }\imdp\text{ with }X\cap\target=\emptyset\},
  \]
  and obtain $\imdp^\sharp$ from $\imdp$ by making every state of $\Trap$
  absorbing.
\end{definition}

\begin{remark}
  \label{app:rem:not-only-bottom}
  Inside a MEC, the opponent can select actions that keep the play there. For example, if $s\notin\target$ and
  $P(s,b)(s)=1$ for some $b\in A_{\opp}$, then $(\{s\},\Gamma)$ with
  $\Gamma(s)=\{b\}$ is an end component, so $s\in\Trap$ even if another
  action moves to $\target$.
\end{remark}

\begin{lemma}[Trap criterion]
  \label{app:lem:trap}
  The Puiseux team strategy $\stratteam$ is Puiseux-winning if and only if
  \[
    \st\bigl(\Prb{\sinit}{\stratopp}(\Reach{\Trap})\bigr)=0
    \quad\text{in }\imdp^\sharp
  \]
  for every $\stratopp\in D$. Moreover,
  $\Prb{u}{\stratopp}(\Reach{(\Trap\cup\target)})=1$ in $\imdp^\sharp$ for
  all $u\in\states$ and $\stratopp\in D$.
\end{lemma}

\begin{proof}
  Fix $\stratopp\in D$. Every bottom SCC $C$ of the chain of $\imdp$ under
  $\stratopp$ satisfies $C=\{t\}$ for some $t\in\target$ or
  $C\subseteq\Trap$, since $(C,\Gamma)$ with $\Gamma(x)=\{\stratopp(x)\}$ is
  an end component. The play almost surely reaches a bottom SCC, so
  \[
    1-\Prb{\sinit}{\stratopp}(\Reach{\target})
    \leq\Prb{\sinit}{\stratopp}(\Reach{\Trap}),
  \]
  where the right side is the same in $\imdp$ and $\imdp^\sharp$. This
  inequality shows the first direction. For the second direction, let
  $\stratopp'\in D$ agree with $\stratopp$ outside $\Trap$ and satisfy
  $\stratopp'(x)\in\Gamma(x)$ for every MEC $(X,\Gamma)$ with
  $X\subseteq\Trap$ and every $x\in X$. A play under $\stratopp'$ never leaves
  $\Trap$ after entering it, so
  $\Prb{\sinit}{\stratopp'}(\Reach{\target})\leq1-\Prb{\sinit}{\stratopp}(\Reach{\Trap})$,
  which has standard part less than $1$ if
  $\st\bigl(\Prb{\sinit}{\stratopp}(\Reach{\Trap})\bigr)>0$. Finally, the
  states of $\Trap\cup\target$ are absorbing in $\imdp^\sharp$, and by the
  first argument every other bottom SCC would lie in $\Trap$.
\end{proof}

In $\imdp^\sharp$, we contract $\target$ into a new absorbing state $\tsink$
and $\Trap$ into a new absorbing state $\fsink$, and obtain a Puiseux MDP
$\mergedmdp$ on
\[
  V=U\cup\{\tsink,\fsink\},
  \qquad
  U=\states\setminus(\target\cup\Trap).
\]
By the sum rule, its transition ranks into the new states are
\[
  \rk(u,b,\tsink)=\min_{t\in\target}\rk(u,b,t),
  \qquad
  \rk(u,b,\fsink)=\min_{x\in\Trap}\rk(u,b,x)
\]
for $u\in U$ and $b\in A_{\opp}$. We write $n=|V|$. The \emph{edges} of
$\mergedmdp$ are the triples in
$E=\{(u,b,v)\in U\times A_{\opp}\times V\mid\rk(u,b,v)<\infty\}$, and an
edge $e=(u,b,v)$ has rank $\rk(e)=\rk(u,b,v)$. Unlike the chain $P$ in
\eqref{app:eq:rank-criterion}, $\mergedmdp$ does not depend on $\stratopp$.

\begin{definition}[Spanning root-forests]
  \label{app:def:root-forest}
  A \emph{spanning root-forest} of $\mergedmdp$ is a set $F\subseteq E$ with exactly
  one edge leaving each $u\in U$, such that every vertex leads to
  $\{\tsink,\fsink\}$. Its rank is $\rk(F)=\sum_{e\in F}\rk(e)$. It is
  \emph{trapped} if $\sinit$ leads to $\fsink$, \emph{targeted} if $\sinit$
  leads to $\tsink$, and \emph{consistent} with $\stratopp\in D$ if
  $F\subseteq\{(u,\stratopp(u),v)\mid u\in U,\ v\in V\}$. The forests consistent with
  $\stratopp$ are the spanning forests of the Markov chain $\mergedmdp_{\stratopp}$
  obtained from $\mergedmdp$ by fixing $\stratopp$. In $\mergedmdp_{\stratopp}$, the minimum
  ranks $\rho_{\stratopp}^{\fsink},\rho_{\stratopp}^{\tsink}\in\mathbb N\cup\{\infty\}$ of the
  trapped and of the targeted forests consistent with $\stratopp$ are
  \[
    \rho_{\stratopp}^{\fsink}=\min_{F\in\Forests_{\{\fsink\}}(\sinit)}\rk(F),
    \qquad
    \rho_{\stratopp}^{\tsink}=\min_{F\in\Forests_{\{\tsink\}}(\sinit)}\rk(F).
  \]
\end{definition}

Since $\min_{F\in\Forests}\rk(F)=\min\{\rho_{\stratopp}^{\fsink},\rho_{\stratopp}^{\tsink}\}$ in
$\mergedmdp_{\stratopp}$, the inequality $\rho_{\stratopp}^{\fsink}>\rho_{\stratopp}^{\tsink}$ is
\eqref{app:eq:rank-criterion} for $\mergedmdp_{\stratopp}$ with $R'=\{\fsink\}$, and
\cref{app:algo:forest-min} computes $\rho_{\stratopp}^{\fsink}$ and $\rho_{\stratopp}^{\tsink}$ in
polynomial time by \cref{app:prop:mdst-poly}.

\begin{corollary}[Forest criterion]
  \label{app:lem:forest}
  For every $\stratopp\in D$,
  \[
    \st\bigl(\Prb{\sinit}{\stratopp}(\Reach{\Trap})\bigr)>0
    \quad\text{in }\imdp^\sharp
    \quad\Longleftrightarrow\quad \rho_{\stratopp}^{\fsink}\leq \rho_{\stratopp}^{\tsink}.
  \]
\end{corollary}

\begin{proof}
  By \cref{app:lem:trap}, every state of $\mergedmdp_{\stratopp}$ has a path to
  $\{\tsink,\fsink\}$, and $\Prb{\sinit}{\stratopp}(\Reach{\Trap})$ in
  $\imdp^\sharp$ equals $\Prb{\sinit}{\stratopp}(\Reach{\{\fsink\}})$ in $\mergedmdp$.
  By \cref{app:lem:rank-criterion},
  \[
    \rk\bigl(\Prb{\sinit}{\stratopp}(\Reach{\{\fsink\}})\bigr)
    =\rho_{\stratopp}^{\fsink}-\min\{\rho_{\stratopp}^{\fsink},\rho_{\stratopp}^{\tsink}\}.
  \]
  A probability has positive standard part exactly when its rank is $0$,
  and the right side is $0$ exactly when $\rho_{\stratopp}^{\fsink}\leq \rho_{\stratopp}^{\tsink}$.
\end{proof}

\subsection{The \texorpdfstring{$\mathsf{NP}$}{NP} Upper Bound}
\label{app:sec:np-upper-bound}

We now prove that limit-sure reachability in the general game setting
belongs to $\mathsf{NP}$.
By \cref{app:lem:small-ranks}, a winning team strategy can be described by
action ranks with polynomially many bits.
For fixed team ranks, \cref{app:lem:trap,app:lem:forest} give a winning criterion.
It compares forest ranks for every opponent strategy.
For one opponent strategy, this comparison takes polynomial time
by \cref{app:prop:mdst-poly}.
However, there may be exponentially many opponent strategies, so checking them one by one takes exponential time.

Instead, we use a certificate with two parts.
The first part consists of the team's action ranks.
The second part is one opponent strategy in an encoded Puiseux MDP.
We check that this opponent strategy is \emph{stable}.
This means that the opponent cannot increase its value by choosing
a different action for the next move and then following the original
strategy.
We show that every stable strategy is optimal and that a stable
strategy exists.
Finally, because stability implies optimality, we propose an algorithm to check stability in polynomial time. This algorithm uses the structure of the encoded Puiseux MDP. 

\paragraph{Recalling the setting.}
Fix action ranks $r_{p,s,a}\in\mathbb N\cup\{\infty\}$ with
$\min_a r_{p,s,a}=0$ at every $(p,s)$, and let $\stratteam$ be the
strategy that realises them as in \cref{app:sec:puiseux-mdp}.
We use the transition ranks of \cref{app:def:induced-mdp}, the set
$\Trap$ of \cref{app:def:trap}, and the graph $\mergedmdp=(V,E)$ of \cref{app:sec:forests}.
The cases listed after \cref{app:lem:trap} will be handled directly in the
algorithm. Therefore, we assume
$\sinit\notin\target\cup\Trap$, $\Trap\neq\emptyset$, and that a
target is reachable from $\sinit$ in the support of $\imdp^\sharp$.

Recall that $V=U\cup\{\tsink,\fsink\}$, where
$U=\states\setminus(\target\cup\Trap)$. 
An edge $e=(u,b,v)$ has source $u\in U$, action label $b$, and rank
$\rk(e)$. There is at most one edge for each triple $(u,b,v)$.
Write $E(u,b)$ for the edges with source $u$ and label $b$.
This set is nonempty by \cref{app:lem:zero-edges}.
Throughout this subsection, an opponent strategy is a deterministic
memoryless map $\stratopp\colon U\to A_{\opp}$.

We use the forests of \cref{app:def:root-forest}. A vertex \emph{leads} to
a root if following the selected edges ends at that root. Write $\mathcal F_{\stratopp}$ for the forests
consistent with $\stratopp$, and $\mathcal F_{\stratopp}(v)$ for those
in which $v$ leads to $\fsink$.
By \cref{app:lem:trap,app:lem:forest},
\begin{equation}
  \stratteam\text{ is Puiseux-winning}
  \quad\Longleftrightarrow\quad
  \rho_{\stratopp}^{\fsink}>\rho_{\stratopp}^{\tsink}\text{ for every }\stratopp.
  \label{app:eq:winning-all-policies}
\end{equation}
Here $\rho_{\stratopp}^{\fsink}$ and $\rho_{\stratopp}^{\tsink}$ are the minimum ranks of
trapped and targeted forests from \cref{app:def:root-forest}, with
$\min\emptyset=\infty$.

\subsubsection{The encoded MDP}

\paragraph{Breaking ties between forests.}
Index the edges as $e_0,\ldots,e_{m-1}$ and define
\begin{align*}
  \widehat{\rk}(e_i)&=2^{m+1}\rk(e_i)+2^i,
  &\widehat{\rk}(J)&=\sum_{e\in J}\widehat{\rk}(e)
  \qquad\text{for all }J\subseteq E.
\end{align*}
The powers of two distinguish different sets of edges. Their total is
at most $2^m-1$, so they cannot reverse a strict comparison of the
original integer ranks. More precisely,
\begin{align*}
  0\leq\widehat{\rk}(J)-2^{m+1}\rk(J)\leq2^m-1.
\end{align*}
Thus $\rk(J)<\rk(J')$ implies
$\widehat{\rk}(J)<\widehat{\rk}(J')$.
If the encoded ranks are equal, the original ranks must be equal,
and then uniqueness of binary representation gives $J=J'$.

\paragraph{Transitions and values.}
The encoded MDP has state set $V$, action set $A_{\opp}$, and
absorbing roots $\tsink,\fsink$. For $u\in U$, put
\begin{align*}
  Z_{u,b}&=\sum_{e\in E(u,b)}\epsilon^{\widehat{\rk}(e)}.
\end{align*}
Its transition probabilities are
\begin{equation}
  \widetilde P(u,b,v)=
  \begin{cases}
    \epsilon^{\widehat{\rk}(u,b,v)}/Z_{u,b},&(u,b,v)\in E,\\
    0,&\text{otherwise.}
  \end{cases}
  \label{app:eq:generic-transition}
\end{equation}
These are nonnegative elements of $\K$ whose sum over $v$ is $1$.
The support is the same as in $\mergedmdp$. The encoded probabilities are
chosen for the certificate. Thus, they are not necessarily equal to the probabilities
of $\imdp^\sharp$. In particular, $\widehat{\rk}(e)$ is the exponent
of an unnormalised weight, not necessarily the rank of a transition
probability after division by $Z_{u,b}$.

For an opponent strategy $\stratopp$, let $q_v^{\stratopp}$ be its probability
of reaching $\fsink$ from $v$ in this encoded MDP, as defined by
the least nonnegative solution of \cref{eq:reach-system} over $\K$.
The opponent seeks to maximise these values.

By \cref{app:lem:trap}, every vertex has a support path to a root under
each fixed strategy. This path property is
unchanged in the encoded MDP. Selecting at every nonroot the first
edge of a shortest path to a root gives a forest, since the distance
to the roots strictly decreases. Hence $\mathcal F_{\stratopp}$ is
nonempty for every $\stratopp$.

Define the forest sums
\begin{align*}
  N_v^{\stratopp}
    &=\sum_{F\in\mathcal F_{\stratopp}(v)}
         \epsilon^{\widehat{\rk}(F)},
  &D^{\stratopp}
    &=\sum_{F\in\mathcal F_{\stratopp}}
         \epsilon^{\widehat{\rk}(F)}.
\end{align*}
Every forest uses one edge from each nonroot. Its probability weight
is therefore $\epsilon^{\widehat{\rk}(F)}/\prod_{u\in U}Z_{u,\stratopp(u)}$.
The common denominator cancels in \cref{app:lem:forest-formula}, giving
\begin{equation}
  q_v^{\stratopp}=\frac{N_v^{\stratopp}}{D^{\stratopp}},
  \qquad D^{\stratopp}>0.
  \label{app:eq:generic-forest-formula}
\end{equation}
Thus the value is the fraction of forest weight sending $v$ to the
failure root. In particular, $q_{\fsink}^{\stratopp}=1$ and
$q_{\tsink}^{\stratopp}=0$. By \cref{eq:reach-system},
\begin{equation}
  q_u^{\stratopp}
    =\sum_{v\in V}\widetilde P(u,\stratopp(u),v)q_v^{\stratopp}
    \qquad\text{for all }u\in U.
  \label{app:eq:fixed-policy-bellman}
\end{equation}

\subsubsection{One stable response suffices}

A strategy $\stratopp$ is \emph{stable} if
\begin{equation}
  q_u^{\stratopp}\geq
  \sum_{v\in V}\widetilde P(u,b,v)q_v^{\stratopp}
  \qquad\text{for all }u\in U\text{ and }b\in A_{\opp}.
  \label{app:eq:stable-policy}
\end{equation}
The right-hand side is the value of playing $b$ for one step and then
following $\stratopp$. Stability says that no such change improves
the value. Equality holds when $b=\stratopp(u)$ by
\cref{app:eq:fixed-policy-bellman}.

This subsection shows that if an opponent strategy is stable, then it is optimal. Moreover, a stable strategy exists.   

The next lemma gives a principle to compare two strategies.
It shows how inequalities in individual states imply an inequality
between their values at every state.

\begin{lemma}[Minimum principle]
  \label{app:lem:np-minimum-principle}
  Fix $\stratopp$ and $d\in\K^V$ with
  $d_{\tsink}=d_{\fsink}=0$. If
  $d_u\geq\sum_v\widetilde P(u,\stratopp(u),v)d_v$ for every
  $u\in U$, then $d_v\geq0$ for every $v\in V$.
\end{lemma}

\begin{proof}
  Suppose $\delta=\min_v d_v<0$ and let $W=\{v:d_v=\delta\}$.
  The roots are outside $W$. For $u\in W$,
  \begin{align*}
    \delta=d_u
    &\geq\sum_v\widetilde P(u,\stratopp(u),v)d_v
     \geq\sum_v\widetilde P(u,\stratopp(u),v)\delta=\delta.
  \end{align*}
  Equality implies that every positive-probability successor of $u$
  also belongs to $W$.
  Thus, no vertex of $W$ has a support path to a root, which is a contradiction.
\end{proof}

We will apply the minimum principle lemma to the difference between two value vectors.
This difference is zero at both roots as the value of any strategy is the same in the absorbing roots.
If it satisfies the inequalities in the lemma, it is nonnegative
at all states.
Thus, one strategy has at least the value of the other in every state.

\begin{lemma}[Stable strategies are optimal]
  \label{app:lem:np-stable-optimal}
  A stable strategy exists. Every stable strategy $\stratopp$ satisfies
  $q_v^{\stratopp}\geq q_v^{\stratopp'}$ for every strategy
  $\stratopp'$ and vertex $v$.
\end{lemma}

\begin{proof}
    Suppose $\stratopp$ is stable and fix an arbitrary opponent strategy
    $\stratopp'$.
    For every $v\in V$, put
    $d_v=q_v^{\stratopp}-q_v^{\stratopp'}$.
    Both strategies have value $0$ at $\tsink$ and $1$ at $\fsink$.
    Hence $d_{\tsink}=d_{\fsink}=0$.
    
    For every $u\in U$, we have
    \begin{align*}
      d_u
      &=q_u^{\stratopp}-q_u^{\stratopp'}\\
      &\geq
        \sum_{v\in V} \bigl( \widetilde P(u,\stratopp'(u),v)q_v^{\stratopp}\bigr)
        -q_u^{\stratopp'}
        &&\text{(by stability of $\stratopp$)}\\
      &=\sum_{v\in V}\widetilde P(u,\stratopp'(u),v)
           \bigl(q_v^{\stratopp}-q_v^{\stratopp'}\bigr)
        &&\text{(by \cref{app:eq:fixed-policy-bellman} for $\stratopp'$)}\\
      &=\sum_{v\in V}\widetilde P(u,\stratopp'(u),v)d_v.
    \end{align*}
    Thus $d$ satisfies all hypotheses of
    \cref{app:lem:np-minimum-principle} for the strategy $\stratopp'$.
    It follows that $d_v\geq0$ for every $v\in V$.
    Therefore $q_v^{\stratopp}\geq q_v^{\stratopp'}$ at every vertex.
    Since $\stratopp'$ was an arbitrary strategy, we conclude that $\stratopp$ is optimal.

  To prove existence, choose a strategy $\stratopp$ maximising
  $\sum_{v\in V}q_v^{\stratopp}$. Such a choice exists because there
  are finitely many positional strategies and $\K$ is totally ordered.
  Suppose it is not stable. Choose $u,b$ with
  $\sum_v\widetilde P(u,b,v)q_v^{\stratopp}>q_u^{\stratopp}$,
  and let $\stratopp'$ change only the action at $u$ to $b$.
  Put $d_v=q_v^{\stratopp'}-q_v^{\stratopp}$.
  At the roots $d_v=0$. At every $z\in U\setminus\{u\}$,
  \cref{app:eq:fixed-policy-bellman} gives
  $d_z=\sum_v\widetilde P(z,\stratopp'(z),v)d_v$. At $u$,
  \begin{align*}
    d_u
    &=\sum_v\widetilde P(u,b,v)q_v^{\stratopp'}-q_u^{\stratopp}\\
    &>\sum_v\widetilde P(u,b,v)
                    (q_v^{\stratopp'}-q_v^{\stratopp})
     =\sum_v\widetilde P(u,b,v)d_v.
  \end{align*}
  By \cref{app:lem:np-minimum-principle}, every $d_v$ is nonnegative.
  The strict inequality then gives $d_u>0$, contradicting maximality
  of the sum. Hence $\stratopp$ is stable.
\end{proof}

Now, we connect optimality in the encoded MDP to the original forest criterion in the main Puiseux MDP. 
For a nonzero series $x\in\K$, write $\nu(x)$ for its
smallest exponent, and put $\nu(0)=\infty$.
Products add these exponents and quotients subtract them.
In a sum of positive series, leading coefficients cannot cancel,
so its rank is the smallest rank of its summands.

\begin{lemma}[Ranks of encoded values]
  \label{app:lem:np-value-rank}
  Fix a deterministic memoryless opponent strategy $\stratopp$.
  Let $\rho_{\stratopp}^{\fsink},\rho_{\stratopp}^{\tsink}$ be its minimum trapped and
  targeted forest ranks, and $q_{\sinit}^{\stratopp}$ its encoded
  failure probability.
  \begin{enumerate}
    \item If $\rho_{\stratopp}^{\fsink}\leq \rho_{\stratopp}^{\tsink}$, then
      $q_{\sinit}^{\stratopp}>0$ and
      $\nu(q_{\sinit}^{\stratopp})\leq2^m-1$.
    \item If $\rho_{\stratopp}^{\fsink}>\rho_{\stratopp}^{\tsink}$, then
      $q_{\sinit}^{\stratopp}=0$ or
      $\nu(q_{\sinit}^{\stratopp})\geq2^m+1$.
  \end{enumerate}
\end{lemma}

\begin{proof}
  If only trapped forests exist, then
  $\rho_{\stratopp}^{\fsink}<\rho_{\stratopp}^{\tsink}=\infty$ and
  $q_{\sinit}^{\stratopp}=1$, so the first claim holds.
  If only targeted forests exist, then
  $\rho_{\stratopp}^{\tsink}<\rho_{\stratopp}^{\fsink}=\infty$ and
  $q_{\sinit}^{\stratopp}=0$, so the second claim holds.

  Suppose both kinds of forests exist.
  Let $a$ and $b$ be their minimum encoded ranks in the encoded MDP, for trapped and
  targeted forests, respectively.
  The encoding preserves strict rank comparisons.
  Thus, forests attaining $a$ and $b$ have ranks
  $\rho_{\stratopp}^{\fsink}$ and $\rho_{\stratopp}^{\tsink}$
  in the original Puiseux MDP, respectively.
  Each tie-breaking contribution lies between $0$ and $2^m-1$, so

  \begin{align*}
    2^{m+1}\rho_{\stratopp}^{\fsink}
    &\leq a\leq2^{m+1}\rho_{\stratopp}^{\fsink}+2^m-1,\\
    2^{m+1}\rho_{\stratopp}^{\tsink}
    &\leq b\leq2^{m+1}\rho_{\stratopp}^{\tsink}+2^m-1.
  \end{align*}

  Since a trapped forest exists, $N_{\sinit}^{\stratopp}$
  contains at least one term.
  Every term is positive in $\K$, so
  $N_{\sinit}^{\stratopp}>0$.
  The denominator $D^{\stratopp}$ is also positive.
  By \cref{app:eq:generic-forest-formula},
  \begin{align*}
    q_{\sinit}^{\stratopp}
    &=\frac{N_{\sinit}^{\stratopp}}{D^{\stratopp}}>0.
  \end{align*}
  The leading coefficients in both forest sums are positive
  and cannot cancel.
  Hence
  \begin{align*}
    \nu(N_{\sinit}^{\stratopp})&=a,\\
    \nu(D^{\stratopp})&=\min(a,b),\\
    \nu(q_{\sinit}^{\stratopp})
      &=a-\min(a,b)=\max(0,a-b).
  \end{align*}

  If $\rho_{\stratopp}^{\fsink}\leq \rho_{\stratopp}^{\tsink}$, then
  \begin{align*}
    a-b
    &\leq2^{m+1}(\rho_{\stratopp}^{\fsink}-\rho_{\stratopp}^{\tsink})+2^m-1\\
    &\leq2^m-1.
  \end{align*}
  Therefore $\nu(q_{\sinit}^{\stratopp})\leq2^m-1$,
  proving the first claim.

  If $\rho_{\stratopp}^{\fsink}>\rho_{\stratopp}^{\tsink}$, both minima are finite
  integers, so their difference is at least $1$.
  Thus
  \begin{align*}
    a-b
    &\geq2^{m+1}(\rho_{\stratopp}^{\fsink}-\rho_{\stratopp}^{\tsink})-(2^m-1)\\
    &\geq2^m+1.
  \end{align*}
  Consequently,
  $\nu(q_{\sinit}^{\stratopp})=a-b\geq2^m+1$,
  proving the second claim.
\end{proof}

The bounds separate strategies that satisfy the forest criterion
from those that violate it.
A violating strategy has a strictly larger encoded value at $\sinit$.
This allows us to check the criterion using one optimal strategy.

\begin{lemma}[One-response certificate]
  \label{app:lem:np-one-policy-certificate}
  Fix the team's action ranks and let $\stratteam$ be the
  Puiseux team strategy that realises them.
  The following statements are equivalent.
  \begin{enumerate}
    \item In the original game, for every deterministic memoryless
      opponent strategy $\stratopp'$,
      \begin{align*}
        \st\bigl(
          \Prb{\sinit}{\stratteam,\stratopp'}(\Reach{\target})
        \bigr)=1.
      \end{align*}
    \item In the encoded MDP, there exists a stable opponent
      strategy $\stratopp$ such that
      $\rho_{\stratopp}^{\fsink}>\rho_{\stratopp}^{\tsink}$.
  \end{enumerate}
\end{lemma}

\begin{proof}
  By \cref{app:lem:trap,app:lem:forest}, the first statement holds
  if and only if $\rho_{\stratopp'}^{\fsink}>\rho_{\stratopp'}^{\tsink}$ for every opponent strategy
  $\stratopp'$.

  Suppose the first statement holds.
  A stable strategy $\stratopp$ exists by
  \cref{app:lem:np-stable-optimal}.
  It satisfies $\rho_{\stratopp}^{\fsink}>\rho_{\stratopp}^{\tsink}$, proving the second
  statement.

  Conversely, fix a stable strategy $\stratopp$ satisfying
  $\rho_{\stratopp}^{\fsink}>\rho_{\stratopp}^{\tsink}$.
  Suppose some strategy $\stratopp'$ satisfies
  $\rho_{\stratopp'}^{\fsink}\leq \rho_{\stratopp'}^{\tsink}$.
  By \cref{app:lem:np-value-rank},
  \begin{align*}
    q_{\sinit}^{\stratopp'}&>0,
    &\nu(q_{\sinit}^{\stratopp'})&\leq2^m-1,
  \end{align*}
  whereas $q_{\sinit}^{\stratopp}$ is zero or has rank at least
  $2^m+1$.
  Thus $q_{\sinit}^{\stratopp'}>q_{\sinit}^{\stratopp}$.
  Indeed, their difference has the same positive leading term
  as $q_{\sinit}^{\stratopp'}$.
  This contradicts \cref{app:lem:np-stable-optimal}.
  Therefore every opponent strategy satisfies the forest
  criterion, proving the first statement.
\end{proof}

\subsubsection{Checking stability in polynomial time}

Fix a deterministic memoryless strategy $\stratopp$.
For a vertex $u\in U$ and an action $b\in A_{\opp}$, define the \emph{Bellman residual}
\begin{equation}
  R^{\stratopp}(u,b)
  =q_u^{\stratopp}-\sum_{v\in V}\widetilde P(u,b,v)q_v^{\stratopp}.
  \label{app:eq:bellman-residual}
\end{equation}
Stability requires every residual to be nonnegative.
In this subsection, we propose a polynomial algorithm to determine the sign of the Bellman residuals. Naively evaluating the forest sums is not possible, as checking all possible forests would take exponential time. Moreover,
the residual is zero for $b=\stratopp(u)$ by
\cref{app:eq:fixed-policy-bellman}, so we only check the case where $b\neq\stratopp(u)$.

For a forest $F\in\mathcal F_{\stratopp}$, let $\chi_F(v)$ be $1$ if the vertex $v$
leads to $\fsink$ in the forest $F$, and $0$ otherwise.
Let $\mathcal T(\stratopp,u,b)$ consist of the spanning trees that
select one edge from $E(z,\stratopp(z))$ at every $z\in U$ and
one additional edge from $E(u,b)$.
For such a tree $T$, let $e_b(T)$ be its unique edge from $E(u,b)$.
By \cref{app:lem:functional-forest},
$F=T\setminus\{e_b(T)\}$ belongs to $\mathcal F_{\stratopp}$.
Define $\operatorname{sgn}(T)=+1$ if $\chi_F(u)=1$, and $-1$ otherwise.

\begin{lemma}[Signed tree formula]
  \label{app:lem:np-signed-tree}
  For $u\in U$ and $b\neq\stratopp(u)$,
  \begin{equation}
    Z_{u,b}D^{\stratopp}R^{\stratopp}(u,b)
    =\sum_{T\in\mathcal T(\stratopp,u,b)}
       \operatorname{sgn}(T)\epsilon^{\widehat{\rk}(T)}.
    \label{app:eq:signed-tree-identity}
  \end{equation}
  Moreover, the residual is zero if and only if $\mathcal T(\stratopp,u,b)$ is empty.
  Otherwise, its sign is the sign of the unique tree of minimum rank.
\end{lemma}

\begin{proof}
  Recall that $q_v^{\stratopp}=N_v^{\stratopp}/D^{\stratopp}$.
Here $N_v^{\stratopp}$ is the sum of the weights of forests in which
$v$ leads to $\fsink$, and $D^{\stratopp}$ is the sum of all forests.
For an edge $e=(u,b,v)$, its transition probability is
\begin{align*}
  \widetilde P(u,b,v)
  &=\frac{\epsilon^{\widehat{\rk}(e)}}{Z_{u,b}},
  & Z_{u,b}
  &=\sum_{e\in E(u,b)}\epsilon^{\widehat{\rk}(e)}.
\end{align*}

Using the definition of the Bellman residual:
\begin{align*}
  R^{\stratopp}(u,b)
  &=q_u^{\stratopp}
    -\sum_{v\in V}\widetilde P(u,b,v)q_v^{\stratopp}\\
  &=\frac{N_u^{\stratopp}}{D^{\stratopp}}
    -\sum_{e=(u,b,v)\in E(u,b)}
      \frac{\epsilon^{\widehat{\rk}(e)}}{Z_{u,b}}
      \frac{N_v^{\stratopp}}{D^{\stratopp}}.
\end{align*}
Multiplying by
$Z_{u,b}D^{\stratopp}$ gives
\begin{align*}
  Z_{u,b}D^{\stratopp}R^{\stratopp}(u,b)
  &=Z_{u,b}N_u^{\stratopp}
    -\sum_{e=(u,b,v)\in E(u,b)}
      \epsilon^{\widehat{\rk}(e)}N_v^{\stratopp}\\
  &=\sum_{e=(u,b,v)\in E(u,b)}
      \epsilon^{\widehat{\rk}(e)}N_u^{\stratopp}
    -\sum_{e=(u,b,v)\in E(u,b)}
      \epsilon^{\widehat{\rk}(e)}N_v^{\stratopp}
      &&\text{(expand $Z_{u,b}$)}\\
  &=\sum_{e=(u,b,v)\in E(u,b)}
      \epsilon^{\widehat{\rk}(e)}
      (N_u^{\stratopp}-N_v^{\stratopp}).
      &&\text{(combine the two sums)}
\end{align*}
  By the definition of $\chi_F$,
  \begin{align*}
    N_u^{\stratopp}-N_v^{\stratopp}
    &=\sum_{F\in\mathcal F_{\stratopp}}
       (\chi_F(u)-\chi_F(v))\epsilon^{\widehat{\rk}(F)}.
  \end{align*}
  Now, we prove \Cref{app:eq:signed-tree-identity} by the following bijection.
  
  A pair $(F,e)$ contributes to the left side only when $u$ and $v$ lead to different roots.
  Its coefficient is $+1$ if $u$ leads to $\fsink$, and $-1$ otherwise.
  The forest has two connected components, each containing one root.
  Thus, adding $e$ joins them without creating a cycle.
  Also, $e\notin F$ since $b\neq\stratopp(u)$.
  Hence $F\cup\{e\}\in\mathcal T(\stratopp,u,b)$.

  Conversely, deleting $e_b(T)=(u,b,v)$ from a tree $T$ gives a
  consistent forest and separates $u$ from $v$.
  Each connected component contains one root by
  \cref{app:lem:functional-forest}.
  These operations are inverse, and
  \begin{align*}
    \widehat{\rk}(T)&=\widehat{\rk}(F)+\widehat{\rk}(e_b(T)),\\
    \operatorname{sgn}(T)&=\chi_F(u)-\chi_F(v).
  \end{align*}
  This proves \cref{app:eq:signed-tree-identity}.
  Distinct trees have distinct encoded ranks, so the minimum term
  cannot cancel and determines the sign.
  Since both $Z_{u,b}$ and $D^{\stratopp}$ are nonnegative, multiplying by them does not change the sign.
\end{proof}

Next, we focus on determining the sign of the minimum-rank tree. Define
\begin{align*}
  \mathcal G&=\{E(z,\stratopp(z)):z\in U\}\cup\{E(u,b)\},\\
  E'&=\bigcup_{G\in\mathcal G}G.
\end{align*}
These are $|U|+1=n-1$ nonempty, disjoint groups, where $n=|V|$.
Let $\mathcal I_1$ contain the subsets of $E'$ with at most one edge
from each group. Let $\mathcal I_2$ contain the subsets with no undirected cycle.
These are the independent sets of a partition matroid and a graphic
matroid, respectively. A common independent set belongs to both.
Give each edge the profit
\begin{align*}
  W&=\sum_{e\in E'}\widehat{\rk}(e),
      &&\text{(total encoded rank)}\\
  p(e)&=W+1-\widehat{\rk}(e).
      &&\text{(profit of an edge)}
\end{align*}
Weighted matroid intersection finds a common independent set of
maximum profit in polynomial time~\cite{Schrijver03}.

\algtext*{EndIf}
\algtext*{EndFor}

\begin{algorithm}[H]
\caption{Residual Sign \texttt{ResidualSign}}
\label{app:algo:residual-sign}
\begin{algorithmic}[1]
    \State \textbf{Input:} $\mergedmdp$, $\widehat{\rk}$, $\stratopp$,
           $u\in U$, $b\neq\stratopp(u)$
    \State \textbf{Output:} The sign of $R^{\stratopp}(u,b)$

    \State $T\gets
        \texttt{MinTree}(\mergedmdp,\widehat{\rk},\stratopp,u,b)$

    \If{$T=\bot$}
        \State \Return $0$
    \EndIf

    \State $e_b\gets$ the unique edge in $T\cap E(u,b)$
    \State $F\gets T\setminus\{e_b\}$

    \If{$u$ leads to $\fsink$ in $F$}
        \State \Return $+1$
    \Else
        \State \Return $-1$
    \EndIf
\end{algorithmic}
\end{algorithm}

\begin{algorithm}[H]
\caption{Minimum Tree \texttt{MinTree}}
\label{app:algo:minimum-tree}
\begin{algorithmic}[1]
    \State \textbf{Input:} $\mergedmdp$, $\widehat{\rk}$, $\stratopp$,
           $u\in U$, $b\neq\stratopp(u)$
    \State \textbf{Output:} A minimum encoded-rank tree in
           $\mathcal T(\stratopp,u,b)$, or $\bot$

    \State $\mathcal G\gets\{E(u,b)\}$
    \ForAll{$z\in U$}
        \State $\mathcal G\gets
            \mathcal G\cup\{E(z,\stratopp(z))\}$
    \EndFor
    \State $E'\gets\bigcup_{G\in\mathcal G}G$

    \State $\mathcal I_1\gets
        \texttt{PartitionMatroid}(\mathcal G)$
        \Comment{At most one edge per group}
    \State $\mathcal I_2\gets
        \texttt{GraphicMatroid}(V,E')$
        \Comment{No undirected cycle}

    \State $W\gets\sum_{e\in E'}\widehat{\rk}(e)$
    \ForAll{$e\in E'$}
        \State $p(e)\gets W+1-\widehat{\rk}(e)$
    \EndFor

    \State $J\gets
        \texttt{WeightedMI}(\mathcal I_1,\mathcal I_2,p)$

    \If{$|J|<|V|-1$}
        \State \Return $\bot$
    \Else
        \State \Return $J$
    \EndIf
\end{algorithmic}
\end{algorithm}

\begin{lemma}[Polynomial stability test]
  \label{app:lem:np-bellman-test}
  \Cref{app:algo:residual-sign}
  computes the sign of
  $R^{\stratopp}(u,b)$ in polynomial time.
  Consequently, stability of $\stratopp$ can be checked in polynomial time.
\end{lemma}

\begin{proof}
  A common independent set of size $n-1$ selects one edge from each
  group and is acyclic. It is connected since it has $n-1$ edges on
  $n$ vertices. Thus, these sets are $\mathcal T(\stratopp,u,b)$.

  For $J\subseteq E'$, we have $0\leq\widehat{\rk}(J)\leq W$ and
  \begin{align*}
    p(J)&=\sum_{e\in J}p(e)\\
        &=(W+1)|J|-\widehat{\rk}(J).
          &&\text{(sum the edge profits)}
  \end{align*}
  If $|J|>|J'|$, then
  \begin{align*}
    p(J)-p(J')
    &=(W+1)(|J|-|J'|)-\widehat{\rk}(J)+\widehat{\rk}(J')\\
    &\geq(W+1)-W=1.
  \end{align*}
  For equal sizes, $p(J)-p(J')=\widehat{\rk}(J')-\widehat{\rk}(J)$.
  Hence maximum profit gives maximum size and then minimum encoded rank.
  If $|J|<n-1$, no required tree exists.
  Otherwise, $J$ is the unique minimum tree.
  Following $F$ from $u$ reaches a root in at most $|U|$ steps.
  The returned sign is therefore correct by \cref{app:lem:np-signed-tree}.

  Membership in the two matroids is checked by counting edges in
  each group and testing for undirected cycles.
  Both checks take polynomial time.
  All ranks, profits, and their sums have polynomially many bits.
  Hence the procedure takes polynomial time.
  
  Running the procedure for every $u$ and $b\neq\stratopp(u)$ uses at most
  $|U|\,|A_{\opp}|$ calls.
  The strategy is stable when no call returns $-1$,
  since the omitted residuals are zero.
\end{proof}

\subsubsection{The nondeterministic algorithm}

\begin{theorem}[Membership in NP]
  \label{app:thm:np-upper-bound}
  The limit-sure reachability problem belongs to $\mathsf{NP}$.
\end{theorem}

\begin{proof}
  Let $N$ be the number of triples $(p,s,a)$ with a player $p\in\team$,
  a state $s\in\states$, and an action $a\in A_p$, and put $B=N^\frac{N}{2}$.
  Use the following nondeterministic algorithm.

  \algtext*{EndIf}
\algtext*{EndFor}

\begin{algorithm}[H]
\caption{Nondeterministic Algorithm for Limit-Sure Reachability}
\label{app:algo:limit-sure-reachability}
\begin{algorithmic}[1]
    \State \textbf{Input:} A game and a rank bound $B$
    \State \textbf{Output:} Accept or reject

    \ForAll{$p\in\team$ and $s\in\states$}
        \State Guess a nonempty support $\Gamma_{p,s}\subseteq A_p$.
        \ForAll{$a\in\Gamma_{p,s}$}
            \State Guess $r_{p,s,a}\in\{0,\ldots,B\}$.
        \EndFor
        \ForAll{$a\in A_p\setminus\Gamma_{p,s}$}
            \State $r_{p,s,a}\gets\infty$
        \EndFor
        \If{$\min_{a\in A_p}r_{p,s,a}\neq0$}
            \State \Return Reject
        \EndIf
    \EndFor

    \State Compute the induced transition ranks.
    \State Compute the MECs and $\Trap$.

    \If{$\sinit\in\target$ or $\Trap=\emptyset$}
        \State \Return Accept
    \EndIf
    \If{$\sinit\in\Trap$}
        \State \Return Reject
    \EndIf
    \If{no target is reachable from $\sinit$
        in the support of $\imdp^\sharp$}
        \State \Return Reject
    \EndIf

    \State Build $\mergedmdp$ and its encoded edge ranks $\widehat{\rk}$.
    \State Guess $\stratopp\colon U\to A_{\opp}$.

    \ForAll{$u\in U$}
        \ForAll{$b\in A_{\opp}\setminus\{\stratopp(u)\}$}
            \State $d\gets
                \texttt{ResidualSign}(\mergedmdp,\widehat{\rk},\stratopp,u,b)$
            \If{$d=-1$}
                \State \Return Reject
            \EndIf
        \EndFor
    \EndFor

    \State Compute $\rho_{\stratopp}^{\fsink}$ and $\rho_{\stratopp}^{\tsink}$
           using \cref{app:prop:mdst-poly}.
    \If{$\rho_{\stratopp}^{\fsink}>\rho_{\stratopp}^{\tsink}$}
        \State \Return Accept
    \Else
        \State \Return Reject
    \EndIf
\end{algorithmic}
\end{algorithm}

  \emph{Running time.}
  Each guessed rank has $O((N+1)\log(N+2))$ bits, and the
  guessed opponent strategy has polynomial size because it is memoryless and deterministic. 
  Transition ranks are computed as in \cref{app:def:induced-mdp}. MECs are computed in
  polynomial time by \cref{app:lem:mec-support}.
  The graph and the binary encodings of the weights $2^i$,
  $2^{m+1}\rk(e_i)+2^i$ have polynomial size.
  There are polynomially many calls to \cref{app:algo:residual-sign}. \cref{app:prop:mdst-poly} computes the last comparison in
  polynomial time. 
  
  \emph{Soundness.}
  Consider the realisation of the guessed ranks by the normalised powers of $\epsilon$
  used in \cref{app:sec:puiseux-mdp}. Acceptance in the preprocessing
  cases is correct by \cref{app:lem:trap}.
  Otherwise, the guessed response is stable. Therefore, the tested inequalities
  hold, and those for $b=\stratopp(u)$ are equalities by
  \cref{app:eq:fixed-policy-bellman}. The last comparison and
  \cref{app:lem:np-one-policy-certificate} show that the team strategy is
  Puiseux-winning. In either case, \cref{app:lem:puiseux} implies that
  the team wins limit surely.

  \emph{Completeness.}
  If the team wins limit surely, \cref{app:lem:small-ranks}
  gives a Puiseux-winning strategy with finite ranks at most
  $N^{N/2}\leq B$. Some run guesses these supports and ranks.
  Their normalised-power realisation is also Puiseux-winning by
  \cref{app:lem:ranks-only}. Thus, preprocessing does not reject. Next, \cref{app:lem:np-one-policy-certificate} gives a stable
  response satisfying the final forest comparison. A run guessing
  that response passes every check and accepts.
\end{proof}

\section{Safety}\label{app:sec:safety}
We now consider the safety objective $\obj=\Box C$, for a set
$C\subseteq\states$ of safe states. We show that, unlike for
reachability, limit-sure and almost-sure winning are equivalent. 
This is an extension of a known result for concurrent 
games with two players~\cite {CdAH12} to teams
without shared randomness. The two-player proof constructs a
memoryless almost-sure winning optimal strategy from a fixed-point characterisation of the
game value. We instead use a compactness argument on the space of team
strategies. 
We only show that an almost-sure winning strategy exists if there exists a limit-sure winning strategy.

In \Cref{app:sec:safety-topology}, we establish the topological lemmas used
in the main proof. We show that the space of team strategies is compact, and for every
opponent strategy, the set of team strategies that guarantee a given
safety probability is closed. In \Cref{app:sec:safety-equivalence}, we
use these lemmas to show that limit-sure winning implies
almost-sure winning.

\subsection{Topological Preliminaries}
\label{app:sec:safety-topology}

For $n\in\mathbb{N}$, let
\[
  \Box_n C
  =
  \{s^0s^1\cdots\in\states^\omega
  \mid s^k\in C \text{ for every } 0\leq k\leq n\}.
\]

The events $\Box_n C$ form a decreasing sequence and satisfy $\Box C=\bigcap_{n\geq0}\Box_n C$. We view a team strategy as the collection of all probabilities $\stratp(h)(a)$, for every player $p\in\team$, history $h\in\states^+$, and action $a\in A_p$. We use the usual product topology on these probabilities.
We equip $\Strats{\team}^{\mathrm h}$ with the product topology.

\begin{lemma}[Compactness of the strategy space]
\label{app:lem:safety-compact}
The space $\Strats{\team}^{\mathrm h}$ is compact.
\end{lemma}

\begin{proof}
For every player $p\in\team$ and history $h\in\states^+$, the value
$\stratp(h)$ belongs to the simplex $\Dist(A_p)$, which is compact since
$A_p$ is finite. By \cref{def:strategies},
\[
  \Strats{\team}^{\mathrm h}
  =
  \prod_{p\in\team}
  \bigl(\Dist(A_p)\bigr)^{\states^+}.
\]
Hence $\Strats{\team}^{\mathrm h}$ is compact because the product of compact sets is also compact by Tychonoff's theorem.
\end{proof}

We show that a lower bound on the safety probability defines
a closed set of strategies.
For a threshold $c\in[0,1]$ and a fixed tuple $\stratrest$ of strategies for the opponents, let
\[
  K(\stratrest,c)
  =
  \left\{
    \stratteam\in\Strats{\team}^{\mathrm h}
    \mid
    \Prb{\sinit}{\stratteam,\stratrest}(\Box C)\geq c
  \right\}.
\]

\begin{lemma}[Closed safety guarantees]
\label{app:lem:safety-closed}
For every $\stratrest\in\Strats{\opp}^{\mathrm h}$ and $c\in[0,1]$,
the set $K(\stratrest,c)$ is closed.
\end{lemma}

\begin{proof}
For $n\geq0$, define
\[
  f_n:
  \Strats{\team}^{\mathrm h}\times\Strats{\opp}^{\mathrm h}
  \longrightarrow [0,1],
  \qquad
  f_n(\stratteam,\stratrest)
  =
  \Prb{\sinit}{\stratteam,\stratrest}(\Box_n C).
\]
Fix $\stratrest$. For fixed $n$, the map
$\stratteam\mapsto f_n(\stratteam,\stratrest)$ depends only on the
strategy choices at histories of length at most $n+1$. There are finitely
many such histories, and the probability is a finite sum of products of
strategy probabilities and transition probabilities. Hence this map is
continuous.

Since $\Box_n C$ decreases to $\Box C$, continuity of probability measures
from above gives
\begin{align*}
  \Prb{\sinit}{\stratteam,\stratrest}(\Box C)
  &=
  \lim_{n\to\infty} f_n(\stratteam,\stratrest)
  &&\text{(since $\Box_n C\downarrow\Box C$)}
  \\
  &=
  \inf_{n\geq0} f_n(\stratteam,\stratrest).
\end{align*}
Therefore:
\[
  K(\stratrest,c)
  =
  \bigcap_{n\geq0}
  \left\{
    \stratteam\in\Strats{\team}^{\mathrm h}
    \mid
    f_n(\stratteam,\stratrest)\geq c
  \right\}.
\]
For every $n$, the set in the intersection is closed by continuity of
$\stratteam\mapsto f_n(\stratteam,\stratrest)$. Therefore,
$K(\stratrest,c)$ is closed.
\end{proof}

Now, we show an existential argument on closed nested sets.

\begin{lemma}[Nested closed sets]
\label{app:lem:nested-closed}
Let $X$ be compact and let $K_2\supseteq K_3\supseteq\cdots$ be nonempty
closed subsets of $X$. Then $\bigcap_{m\geq2}K_m\neq\emptyset$.
\end{lemma}

\begin{proof}
Suppose that $\bigcap_{m\geq2}K_m=\emptyset$.
For every $m\geq2$, let $U_m=X\setminus K_m$. Since $K_m$ is closed, $U_m$ is open, and \begin{align*} 
    X &=X\setminus\bigcap_{m\geq2}K_m &&\text{(the intersection is empty)} \\ &=\bigcup_{m\geq2}(X\setminus K_m) &&\text{(by De Morgan's law)} \\ &=\bigcup_{m\geq2}U_m. 
\end{align*}

Thus $(U_m)_{m\geq2}$ is an open cover of $X$. By the definition of compactness, every open cover of $X$ has a finite subcover. Hence there exist indices $m_1,\ldots,m_k$ such that $X=U_{m_1}\cup\cdots\cup U_{m_k}$.
Let $M=\max\{m_1,\ldots,m_k\}$. Since the sets $K_m$ are decreasing, the sets $U_m$ are increasing, so $U_{m_i}\subseteq U_M$ for every $i$. Hence $X=U_M$. Therefore $K_M=\emptyset$, which is a contradiction.
\end{proof}

\subsection{From Limit-Sure to Almost-Sure Winning}
\label{app:sec:safety-equivalence}

\begin{theorem}[Limit-sure safety equals almost-sure safety]
\label{app:thm:safety-limit-almost}
For a game with a safety objective, the team wins limit surely if and
only if it wins almost surely.
\end{theorem}

\begin{proof}
Almost-sure winning implies limit-sure winning by
\cref{def:almost-sure,def:limit-sure}. 

We prove the converse.
Suppose that the team wins limit surely. 
For every integer $m\geq2$, let
\[
  K_m
  =
  \left\{
    \stratteam\in\Strats{\team}^{\mathrm h}
    \;\middle|\;
    \Prb{\sinit}{\stratteam,\stratrest}(\Box C)
    \geq 1-\frac1m
    \text{ for every }
    \stratrest\in\Strats{\opp}^{\mathrm h}
  \right\}.
\]
Every $K_m$ is nonempty by limit-sure winning. Moreover,
\[
  K_m
  =
  \bigcap_{\stratrest\in\Strats{\opp}^{\mathrm h}}
  K\left(\stratrest,1-\frac1m\right),
\]
so $K_m$ is closed by \cref{app:lem:safety-closed} because an intersection of closed sets is always closed. The sets are nested ($K_2\supseteq K_3\supseteq K_4\supseteq\cdots
$), since
$1-\frac{1}{m+1}>1-\frac1m$.
Thus, $(K_m)_{m\geq2}$ is a nested sequence of nonempty closed subsets of the
compact space $\Strats{\team}^{\mathrm h}$. By \cref{app:lem:nested-closed},
$\bigcap_{m\geq2}K_m\neq\emptyset$. Fix
$\stratteam^*\in\bigcap_{m\geq2}K_m$.

For every $\stratrest\in\Strats{\opp}^{\mathrm h}$ and every $m\geq2$,
\[
  \Prb{\sinit}{\stratteam^*,\stratrest}(\Box C)
  \geq 1-\frac1m.
\]
Since the left-hand side does not depend on $m$ and is at most $1$, it
must be equal to $1$. Thus
\[
  \Prb{\sinit}{\stratteam^*,\stratrest}(\Box C)=1
  \qquad
  \text{for every }
  \stratrest\in\Strats{\opp}^{\mathrm h}.
\]
Hence $\stratteam^*$ is almost-sure winning.
\end{proof}

By \cite[Theorem~8 and Corollary~2]{BHMST26}, almost-sure and sure winning
also coincide for safety. Together with \cref{app:thm:safety-limit-almost}, we obtain that for safety objectives, limit-sure, almost-sure, and sure winning are all equivalent. 

Moreover, sure safety can be decided in polynomial time by computing the
largest set of safe states from which some team profile ensures that,
against every opponent action profile, all possible successors remain in the set.
Hence limit-sure, almost-sure, and sure safety are all decidable in polynomial
time.

\section{IRATL}\label{app:sec:iratl}
We prove \cref{thm:iratl-model-checking}. Fix the input structure
$(G,L)$ and write for a state formula $\psi$: 
\[
 \operatorname{Sat}(\psi)=\{s\in\states\mid(G,L),s\models\psi\}
\]
for the set $\operatorname{Sat}(\psi)$. 
We omit $(G,L)$ from the notation if it is clear from context. 
From the semantics of boolean logic, the following statements hold. 
\[
\begin{aligned}
 \operatorname{Sat}(q)&=\{s\in\states\mid q\in L(s)\},\\
 \operatorname{Sat}(\neg\psi)&=\states\setminus\operatorname{Sat}(\psi),\\
 \operatorname{Sat}(\psi\vee\chi)&=
     \operatorname{Sat}(\psi)\cup\operatorname{Sat}(\chi).
\end{aligned}
\]

We write $\operatorname{Sub}(\varphi)$ for the set of distinct state subformulas of $\varphi$, including $\varphi$ itself and the state formulas appearing as arguments of temporal operators. 
It has size at most the length of $\varphi$.

\begin{lemma}\label{lem:iratl-one-modality}
    Consider a formula $\langle\!\langle C\rangle\!\rangle_w^r\theta$.
    Given the set of states satisfying each state formula appearing in $\theta$, we can decide in polynomial time using an $\NP$ oracle whether this formula holds at a given state.
\end{lemma}

\begin{proof}
    First consider individual randomisation for a coalition $C$.
For $X\subseteq\states$, define its sure predecessor by
\[
 \mathsf{CPre}_C(X)=
 \bigl\{s\in\states\mid
  \exists\bar a\in A_C\ \forall\bar b\in A_{\players\setminus C},\
  \Supp(\delta(s,(\bar a,\bar b)))\subseteq X\bigr\}.
\]
This set can be computed in polynomial time.

We first consider individual randomisation since the complexity of the shared randomisation is a well-studied problem~\cite{dAHK98}.
\paragraph{Next.}
Let $T=\operatorname{Sat}(\psi)$. 
For sure, almost-sure and limit-sure, the winning set for $\bigcirc\psi$ is $\mathsf{CPre}_C(T)$.
If $s\in \mathsf{CPre}_C(T)$, clearly there is a sure winning strategy from $s$. 
Conversely, suppose $s\notin\mathsf{CPre}_C(T)$, and let $\eta>0$ be the smallest positive transition probability in $G$. 
For every team profile $\bar a$, some opponent profile has probability
at least $\eta$ of reaching $\states\setminus T$ in the next step. If the opponent chooses uniformly and independently, every opponent profile has probability $1/|A_{\players\setminus C}|$. Hence, against any choice of the team, the probability of reaching $\states\setminus T$ is at least $\eta/|A_{\players\setminus C}|>0$. This rules out limit-sure winning, and therefore also almost-sure winning.

\paragraph{Always.}
Let $U=\operatorname{Sat}(\psi)$. By \cref{thm:safety-limit-almost}, limit-sure safety implies almost-sure safety. 
An almost-sure safety strategy can be made deterministic by choosing an action that has positive probability by the strategy for each team player at each history.
If the resulting strategy admitted an unsafe play, its finite prefix ending at the first unsafe state would have positive probability under the original strategy  which leads to a contradiction. Thus all three of sure, limit-sure and almost-sure safety coincide.
Their winning set can be computed by repeatedly computing $Z_i$ as below until a fixpoint is reached. 
\[
 Z_0=U,\qquad Z_{i+1}=U\cap\mathsf{CPre}_C(Z_i).
\]
Indeed, $Z_i$ is the set from which the team can surely remain in $U$ for $i$ transitions. At the fixed point, the witnessing strategy keeps the play inside the set forever. 
There are at most $|\states|$ many times this can decrease, so the computation is polynomial.

\paragraph{Until.}
For $\varphi_1\,\mathsf U\,\varphi_2$, let $U=\operatorname{Sat}(\varphi_1)$, $T=\operatorname{Sat}(\varphi_2)$, and $F=\states\setminus(U\cup T)$.
We make every state in $T\cup F$ absorbing and use the reachability objective $\Reach{T}$. 
The sets $U$ and $T$ are computed in the original structure. and remain fixed. Strategies in the two games  are the same until a vertex in $T\cup F$ is visited. 
Reaching $T$ before $F$ satisfies both the original until objective and the modified reachability objective and reaching $F$ first violates both the original until objective and the modified reachability objective.  If a play remains in $U\setminus T$ forever, then it violates both objectives.
This also covers initial states in $T$ or $F$, and preserves winning probabilities and sure winning. Hence the reduction preserves sure, almost-sure and limit-sure objectives.

Since randomisation is useless for sure reachability, this reduces the problem with sure reachability to two-player case for which there are polynomial algorithms provided by previous works~\cite{dAHK07}.
Almost-sure reachability is in $\NP$~\cite[Theorem~6]{BHMST26}, and limit-sure reachability is in $\NP$ by \cref{thm:np-upper-bound}.

\paragraph{Shared randomisation.}
We can merge the coalition $C$ of players into one player that has the action set $A_C$. The opponent also can be always treated as one player. The resulting two-player concurrent game has polynomial-time
algorithms for qualitative reachability~\cite{dAHK07}.
\end{proof}

Before we state the lemma and prove the lower bound for the qualitative fragment, we state the problem from which we reduce, which is \emph{sequentially nested SAT}.
This problem is complete for the complexity class $\P^\NP$. 
    \begin{definition}[$\mathsf{SNSAT}$]\label{def:iratl-snsat}
        An instance $I$ consists of an integer $n\geq 1$, distinct Boolean variables $x_1\dots,x_n$, pairwise disjoint finite sets $Y_1,\dots,Y_n$ of Boolean variables, disjoint from  the $x_i$s and 3-CNF formulas $F_1,\dots,F_n$ such that 
        \[
            \operatorname{Var}(F_i)\subseteq \{x_1,\dots,x_{i-1}\}\cup Y_i
        \]
    We inductively define $v_1,\ldots,v_n \in \{0,1\}$ as:
    \[
            v_i = 1 \quad \text{if}
            \quad \exists\, \nu_i \ : Y_i\to\{0,1\} \text{ such that }F_i(v_1,\ldots,v_{i-1},\nu_i) = 1) .
    \]
    The decision problem is to check if $v_n = 1$.
    \end{definition}
Without loss of generality, we assume that an instance of the problem such that each 3-CNF formula $F_i$ has at least one clause and that each clause has three literals per clause. 
\begin{lemma}\label{lem:iratlQualitativeHardness}
    Model checking qualitative IRATL on explicitly represented finite stochastic game structures is $\P^\NP$-hard.
\end{lemma}

\begin{proof}
We provide a reduction from the sequentially nested satisfiability problem of Laroussinie, Markey and Schnoebelen~\cite{LMS01} which is defined in \cref{def:iratl-snsat}. The precise 3-CNF version above is stated in the proof of \cite[Proposition~3.10]{LMO08}, where its $\P^{\NP}$-completeness is
shown. %

We reduce from $\mathsf{SNSAT}$, which consists of a sequence of satisfiability questions in which each formula may depend on the the satisfaction of the previous questions. 
We first give some intuition. 
We construct a game in which two cooperating players independently choose a clause and a literal within it. The opponent challenges a pair of clauses, checking that the players agree when selecting from the same clause and never select complementary literals. Given a satisfying assignment, the team can choose a true literal from every clause and pass every challenge. Conversely, an almost-sure winning strategy must cover every clause and make mutually consistent choices, from which a satisfying assignment can be extracted.

To capture dependencies between questions, a selected literal referring to an earlier answer leads to a state representing that answer or its negation. Nested IRATL formulas then evaluate these states using the preceding formula layer. Some additional states in the construction of the game-gadget ensures that the resulting formula at the positive state for the final question therefore decides the original SN-SAT instance.

\paragraph{Construction of the labelled game structure.}
Let $m\geq 1$ be the number of clauses of each $F_i$.
We assume that each $F_i$ has the same number of clauses without loss of generality.  
In $F_i$, let the literal $\ell_{1=i,c,r}$ denote the literal at clause $c$ and position $r\in \{1,2,3\}$. For instance, for a clause $(a \lor \neg b \lor c )$, the literal $a$ has position $1$ and $\neg b$ position 2. We also write $[d]$ to denote the set $\{1,\dots,d\}$ of positive integers up to $d$. 

The players are in the set $\players = \{a,b,o\}$.
For the input $\mathsf{SNSAT}$ instance that consists of the integer $n$, we define the game instance $G_I$ with $4n + 4$ many states as follows, with $2n+2$ states in $S_0$ and $2n+2$ states in $S_1$ defined as follows: 
\[
\begin{aligned}
   S_0&=\{\mathsf t,\mathsf f\}
         \cup\{s_i^+,s_i^-\mid i\in[n]\},\\
 S_1&=\{\widehat u\mid u\in S_0\},\qquad
 \states=S_0\mathbin{\dot\cup}S_1.  
\end{aligned}
\]

 The action sets available at each state is
\[
 A_a=A_b=[m]\times[3],\qquad A_o=[m]^2\times[2].
\]
Intuitively, each player's action selects a clause and location thereby choosing a literal, while the opponent that selects the action $(u,v,h)$ challenges a clause $u$ for $a$ and $v$ for clause $b$ and inspect $a$ if $h = 1$ and $b$ if $h=2$.

Finally, we define the state variables $Q=\{q,p,z\}$ and define the labelling completely as follows:
\[
 L_I(\mathsf t)=\{q,z\},\qquad
 L_I(u)=\{q\}\ (u\in S_0\setminus\{\mathsf t\}),\qquad
 L_I(\widehat u)=\{p\}\ (u\in S_0).
\]

We now define the transition function. From all states except the states that are of the form $s_i^-$, every action profile leads to another state deterministically, and independent of the actions and we first define that. Since this part is deterministic, we define a function
$D:\states\times A_a\times A_b\times A_o\to\states$ and then define
\[
 \delta(s,\alpha)(w)=
 \begin{cases}1&\text{if }w=D(s,\alpha),\\0&\text{otherwise.}\end{cases}
\]
For every action profile  $\alpha$, we define
\[
\begin{aligned}
 D(\widehat u,\alpha)&=u &&(u\in S_0),\\
 D(\mathsf t,\alpha)&=\widehat{\mathsf f},\qquad
 D(\mathsf f,\alpha)=\widehat{\mathsf t},\\
 D(s_i^+,\alpha)&=\widehat{s_i^-} &&(i\in[n]).
\end{aligned}
\]
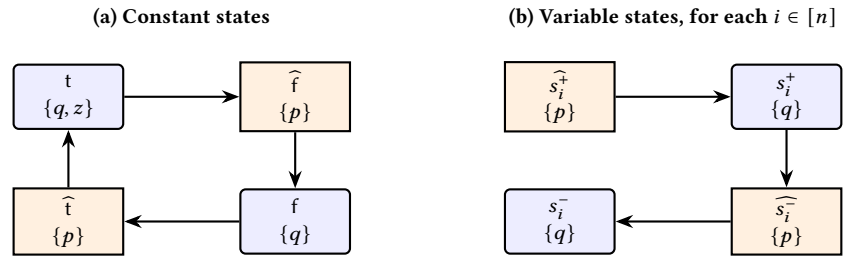
\begin{figure}[htbp]
\centering
\begin{tikzpicture}[x=1cm,y=1cm,font=\small,>=Stealth,
  base/.style={draw,rounded corners=2pt,fill=blue!7,minimum width=1.45cm,
               minimum height=0.85cm,align=center},
  copy/.style={draw,fill=orange!12,minimum width=1.45cm,
               minimum height=0.85cm,align=center},
  every path/.style={thick}]
 \node[font=\small\bfseries] at (1.5,1.05) {(a) Constant states};
 \node[base] (T) at (0,0) {$\mathsf t$\\$\{q,z\}$};
 \node[copy] (HF) at (3,0) {$\widehat{\mathsf f}$\\$\{p\}$};
 \node[base] (F) at (3,-1.65) {$\mathsf f$\\$\{q\}$};
 \node[copy] (HT) at (0,-1.65) {$\widehat{\mathsf t}$\\$\{p\}$};
 \draw[->] (T) -- (HF);
 \draw[->] (HF) -- (F);
 \draw[->] (F) -- (HT);
 \draw[->] (HT) -- (T);
 \node[font=\small\bfseries] at (8,1.05) {(b) Variable states, for each $i\in[n]$};
 \node[copy] (HP) at (6.5,0) {$\widehat{s_i^+}$\\$\{p\}$};
 \node[base] (P) at (9.5,0) {$s_i^+$\\$\{q\}$};
 \node[copy] (HM) at (9.5,-1.65) {$\widehat{s_i^-}$\\$\{p\}$};
 \node[base] (M) at (6.5,-1.65) {$s_i^-$\\$\{q\}$};
 \draw[->] (HP) -- (P);
 \draw[->] (P) -- (HM);
 \draw[->] (HM) -- (M);
\end{tikzpicture}
\caption{States and their proposition labels. 
Arrow represents that every action profile and has transition probability $1$ to the state. }
\Description{Two panels of states, each drawn as a box with its set of propositions. Base states are blue and hatted copies are orange. Panel (a), constant states: a cycle from t, labelled q and z, to the copy of f, labelled p, to f, labelled q, to the copy of t, labelled p, and back to t. Panel (b), variable states for each i: a chain from the copy of s_i plus, labelled p, to s_i plus, labelled q, to the copy of s_i minus, labelled p, to s_i minus, labelled q. Every arrow is taken with probability 1 under every action profile.}
\label{fig:iratl-state-labels}
\end{figure}
It remains to define $D$ at $s_i^-$.
For a literal of $F_i$, we define a state for an literal an intermediary mapping $\tau_i$, using which we define $D$:
\[
 \tau_i(\ell)=
 \begin{cases}
  \mathsf t&\text{if $\ell=y$ or $\ell=\neg y$ for $y\in Y_i$},\\
  s_j^+&\text{if $\ell=x_j$ for $j<i$},\\
  s_j^-&\text{if $\ell=\neg x_j$ for $j<i$}.
 \end{cases}
\]
Given an action profile of all players, 
$\alpha=((c,r),(d,t),(u,v,h))$, we write
\[ \ell_a=\ell_{i,c,r},\quad\ell_b=\ell_{i,d,t}.
\]
We say that the literals are complementary if one is a variable and the other its negation. We now define the remaining transitions from state~$s_i^-$.
\begin{enumerate}
 \item If $(c,d)\neq(u,v)$, then $D(s_i^-,\alpha)=s_i^-$.
 \item If $(c,d)=(u,v)$ and either
       $(c=d\text{ and }r\neq t)$ or $\ell_a,\ell_b$ are
       complementary, then $D(s_i^-,\alpha)=\widehat{\mathsf f}$.
 \item Otherwise, $D(s_i^-,\alpha)=\widehat{\tau_i(\ell_a)}$ if $h=1$,
       and $D(s_i^-,\alpha)=\widehat{\tau_i(\ell_b)}$ if $h=2$.
\end{enumerate}
This comparison is very similar to the  lower bound hardness construction
in Brice et al.~\cite{BHMST26}. Similar to their work, where the players guess vertices that form $k$-clique and the opponent checks if their guesses are consistent, here the players select literals and the opponent check if these choices are consistent.
Figure~\ref{fig:iratl-clause-test} shows these three cases.
\begin{figure}[h]
\centering
\begin{tikzpicture}[x=1cm,y=1cm,font=\small,>=Stealth,
  base/.style={draw,rounded corners=2pt,fill=blue!7,minimum width=1.35cm,
               minimum height=0.8cm,align=center},
  copy/.style={draw,fill=orange!12,minimum width=1.35cm,
               minimum height=0.8cm,align=center},
  edge/.style={->,thick},
  tag/.style={font=\footnotesize,fill=white,inner sep=2pt}]
 \node[base] (M) at (0,0) {$s_i^-$\\$\{q\}$};
 \node[copy] (HF) at (6,2.7) {$\widehat{\mathsf f}$\\$\{p\}$};
 \node[copy] (HT) at (6,0.9) {$\widehat{\mathsf t}$\\$\{p\}$};
 \node[copy] (HP) at (6,-0.9) {$\widehat{s_j^+}$\\$\{p\}$};
 \node[copy] (HN) at (6,-2.7) {$\widehat{s_j^-}$\\$\{p\}$};
 \node[base] (F) at (9,2.7) {$\mathsf f$\\$\{q\}$};
 \node[base] (T) at (9,0.9) {$\mathsf t$\\$\{q,z\}$};
 \node[base] (P) at (9,-0.9) {$s_j^+$\\$\{q\}$};
 \node[base] (N) at (9,-2.7) {$s_j^-$\\$\{q\}$};
 \draw[edge] (M) to[out=145,in=215,looseness=5]
   node[tag,left,align=center] {(1) Index\\doesn't\\match} (M);
 \draw[edge] (M) -- node[tag,pos=0.63,sloped,above] {(2) Inconsistent choices} (HF);
 \draw[edge] (M) -- node[tag,pos=0.63,sloped,above] {(3) $\ell=y$ or $\neg y$, $y\in Y_i$} (HT);
 \draw[edge] (M) -- node[tag,pos=0.63,sloped,below] {(3) $\ell=x_j$, $j<i$} (HP);
 \draw[edge] (M) -- node[tag,pos=0.63,sloped,below] {(3) $\ell=\neg x_j$, $j<i$} (HN);
 \draw[edge] (HF) -- node[tag,above] {all profiles} (F);
 \draw[edge] (HT) -- node[tag,above] {all profiles} (T);
 \draw[edge] (HP) -- node[tag,above] {all profiles} (P);
 \draw[edge] (HN) -- node[tag,above] {all profiles} (N);
\end{tikzpicture}
\caption{All transition from state $s_i^-$.}
\Description{Transitions from state s_i minus, labelled q. Case 1: if the index does not match, a self-loop back to s_i minus. Case 2: if the choices are inconsistent, a transition to the copy of f. Case 3: a transition to the copy of t if the literal is y or not y for some y in Y_i, to the copy of s_j plus if the literal is x_j with j less than i, and to the copy of s_j minus if the literal is not x_j with j less than i. Each copy, labelled p, moves under all action profiles to the corresponding base state: f labelled q, t labelled q and z, s_j plus labelled q, and s_j minus labelled q.}
\label{fig:iratl-clause-test}
\end{figure}
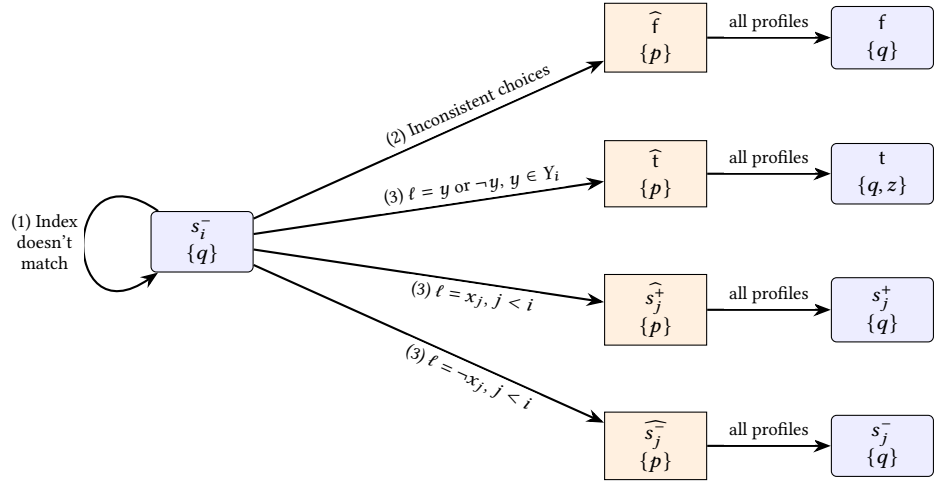

\paragraph{Construction of the state and formula.}
Let the coalition be $C=\{a,b\}$ and define the IRATL state formulas inductively by
\[
\begin{aligned}
 \Phi_0&=z,\\
 \Phi_{k+1}&=
 \neg\langle\!\langle C\rangle\!\rangle_{\mathsf{almost}}^{\mathsf{ind}}
 \left(q\,\mathsf U\,
  \left(p\wedge
   \langle\!\langle\varnothing\rangle\!\rangle_{\mathsf{sure}}^{\mathsf{ind}}
   \bigcirc\Phi_k\right)\right).
\end{aligned}
\tag{$\ddagger$}\label{eq:iratl-snsat-formulas}
\]
The output of the reduction is the model-checking instance
$\bigl((G_I,L_I),s_n^+,\Phi_{2n}\bigr)$.

We note that the construction is polynomial in the size of $\mathsf{SNSAT}$ instance.
\paragraph{A consequence of the next and until semantics.}
For any IRATL state formula $\psi$, we use the following abbreviation for the state formula to avoid some cumbersome notation
\[
 \chi_\psi=p\wedge
   \langle\!\langle\varnothing\rangle\!\rangle_{\mathsf{sure}}^{\mathsf{ind}}
   \bigcirc\psi.
\]
By the labelling and the sure-next semantics in \cref{sec:IRATL},
\[
 u\not\models\chi_\psi\quad(u\in S_0),\qquad
 \widehat u\models\chi_\psi\quad\Longleftrightarrow\quad u\models\psi.
\]
Indeed, every profile at $\widehat u$ has the single successor $u$.
Every state in $S_0$ satisfies $q$, and none in $S_1$ does.
Consequently, a play starting in $S_0$ satisfies $q\,\mathsf U\,\chi_\psi$ exactly when its first visit to $S_1$ occurs
at a state $\widehat u$ with $u\models\psi$. 
A play that does not visit $S_1$ cannot satisfy the strong-until objective. If its first state in $S_1$ fails $\chi_\psi$, the objective also fails, since $q$ needs to be visited earlier for until to hold.

We now show correctness at $s_i^-$.
\begin{proposition}\label{lemma:correctnessatSiMinus}
Fix $i$, an IRATL state formula $\psi$, and
$\beta_1,\ldots,\beta_{i-1}\in\{0,1\}$ such that
\[
\begin{aligned}
 &\mathsf{t}\models\psi,\qquad \mathsf{f}\not\models\psi,\\
 &s_j^+\models\psi\ \Longleftrightarrow\ \beta_j=1,
 \qquad
 s_j^-\models\psi\ \Longleftrightarrow\ \beta_j=0
 \quad(j<i).
\end{aligned}
\]
We claim that 
\[
 s_i^-\models
 \langle\!\langle C\rangle\!\rangle_{\mathsf{almost}}^{\mathsf{ind}}
       (q\,\mathsf U\,\chi_\psi)
 \quad\Longleftrightarrow\quad
 \exists\nu_i:Y_i\to\{0,1\}\;
       F_i(\beta_1,\ldots,\beta_{i-1},\nu_i)=1.
\]
\end{proposition}
\begin{proof}
$(\mathbf{\impliedby}).$
Fix a satisfying assignment and
choose one true literal position $r_c$ in each clause $c$.
At state $s_i^-$, each of $a,b$ independently plays the memoryless strategy 
$(c,r_c)$ with probability $1/m$ for each $c\in[m]$ and assigns
probability zero to every other action and plays arbitrary distributions
at other states.
Selected literals agree within a clause, are never complementary, and all selected $x_j$-literals have their prescribed truth values.
Whenever the clause indices match the opponent's indices, the successor
therefore satisfies $\chi_\psi$, by the definition of $\tau_i$.
For every opponent action, conditional on still being at $s_i^-$,
matching has probability $1/m^2$.  Using induction on $N$, we can bound the probability of staying at $s_i^-$ for $N$ transitions by $(1-1/m^2)^N$.
The probability of staying there forever is therefore zero. Every first visit to $S_1$ satisfies $\chi_\psi$. The until semantics proves the left-hand side of the proposition.

$(\mathbf{\implies}).$ 
Suppose that $(s_i^- \models \langle\langle C\rangle\rangle_{\mathrm{almost}}^{\mathrm{ind}}(q,U,\chi_\psi))$. Let
\[
R_\psi={\widehat u\in S_1\mid u\models\psi},
\]
where satisfaction of $\psi$ is evaluated in the original structure. From $s_i^-$, every transition either remains at $s_i^-$ or visits $S_1$. Moreover, the until objective is satisfied exactly when the first state reached in $S_1$ belongs to $R_\psi$.

Therefore, making every state in $S_1$ absorbing and taking $R_\psi$ as the reachability target preserves the winning probability from the state $s_i^-$ for every strategy profile. By the memoryless sufficiency theorem for almost-sure reachability~\cite[Theorem~5]{BHMST26}, there exists a memoryless almost-sure winning team strategy in the modified game. Its action at $s_i^-$ also ensures the original until objective is almost surely.

For $d\in\{a,b\}$ and $c\in[m]$, let
\[
 R_d(c)=\{r\in[3]\mid\sigma_d(s_i^-)(c,r)>0\}
\]
be the literal positions player $d$ selects with positive probability in clause $c$.

Every $R_d(c)$ is nonempty. Otherwise, if $R_a(c)=\varnothing$, the opponent can play the action $(c,1,1)$ at every visit to $s_i^-$, forcing the game to stay forever at the state.
If $R_b(c)=\varnothing$, the opponent can instead play $(1,c,1)$, thus also preventing the team from winning almost surely. 
Now the opponent can choose the same clause for both players by choosing the action $(c,c,1)$.
Since the team cannot coordinate their randomness, $(r,t)\in R_a(c)\times R_b(c)$ also occurs with positive probability. 
If $r\ne t$, the next state is $\widehat{\mathsf f}$, which does not satisfy $\chi_\psi$ and hence does not satisfy the until objective. 
Thus there is a unique position $r_c$ such that $R_a(c)=R_b(c)=\{r_c\}$.

The selected literals $\ell_{i,c,r_c}$ are pairwise noncomplementary since the action $(c,d,1)$ played by the opponent to two clauses with complementary
selected literals reaches the state $\widehat{\mathsf f}$ with positive probability. 
Moreover, every selected literal over an earlier variable $x_j$ is true under $\beta$. 
If such a literal in clause $c$ were false and if the opponent plays $(c,c,1)$, the play would reach $\widehat{\tau_i(\ell_{i,c,r_c})}$ with positive probability, and this
state fails $\chi_\psi$ by the assumptions on $\psi$ and $\beta$.
Both cases contradict almost-sure winning.

Let us assign each selected local literal to be true, and assign all remaining variables in $Y_i$ arbitrarily. 
This defines a consistent valuation $\nu_i$, since no local variable is selected with both signs.
Every clause has its selected literal true under $(\beta,\nu_i)$, so $F_i(\beta_1,\ldots,\beta_{i-1},\nu_i)=1$.
This proves the forward direction.
\end{proof}

\paragraph{Correctness of the reduction.}

\begin{proposition}
\label{prop:iratl-snsat-induction}
For the constructed structure $(G_I,L_I)$ and formulas $\Phi_k$,
\[
\begin{aligned}
 s_i^-\models\Phi_k&\Longleftrightarrow v_i=0
                  &&\text{for every }k\geq2i-1,\\
 s_i^+\models\Phi_k&\Longleftrightarrow v_i=1
                  &&\text{for every }k\geq2i,
\end{aligned}
\qquad i\in[n].
\]
\end{proposition}
\begin{proof}
For every
$k\geq0$,
\begin{equation}\label{eq:iratl-snsat-fixed-states}
\begin{aligned}
 \mathsf t\models\Phi_{k+1}&\Longleftrightarrow\mathsf f\not\models\Phi_k,\\
 \mathsf f\models\Phi_{k+1}&\Longleftrightarrow\mathsf t\not\models\Phi_k,\\
 s_i^+\models\Phi_{k+1}&\Longleftrightarrow s_i^-\not\models\Phi_k.
\end{aligned}
\end{equation}
Since $\Phi_0=z$, we have $\mathsf t\models\Phi_0$ and
$\mathsf f\not\models\Phi_0$. From the first two equivalences, we have $\mathsf t\models\Phi_k$ and $\mathsf f\not\models\Phi_k$ for all $k\geq0$.
We now induct on the index $i$.

\paragraph{Base case: $i=1$.}
Fix any $k\geq1$. The formula $F_1$ has no earlier variables, so the fact that 
$\mathsf t\models\Phi_k$ and $\mathsf f\not\models\Phi_k$ is enough for \cref{lemma:correctnessatSiMinus} with $\psi=\Phi_{k-1}$.
The team can satisfy the corresponding almost-sure until objective from $s_1^-$ exactly when $F_1$ is satisfiable, that is, when $v_1=1$.
Since $\Phi_k$ is the inverse of this objective,
\[
 s_1^-\models\Phi_k\quad\Longleftrightarrow\quad v_1=0
 \qquad(k\geq1).
\]
For any $k\geq2$, the third equivalence in
\cref{eq:iratl-snsat-fixed-states} implies that
\[
 s_1^+\models\Phi_k
 \quad\Longleftrightarrow\quad s_1^-\not\models\Phi_{k-1}
 \quad\Longleftrightarrow\quad v_1=1.
\]

\paragraph{Induction step: $i>1$.}
Assume both equivalences in the proposition hold for every $j<i$.
Let $k\geq2i-1$. For each $j<i$,
\[
 k-1\geq2i-2\geq2j.
\]
Thus the induction hypothesis ensures that $\Phi_{k-1}$ satisfies both:
\[
 s_j^+\models\Phi_{k-1}\Longleftrightarrow v_j=1,
 \qquad
 s_j^-\models\Phi_{k-1}\Longleftrightarrow v_j=0.
\]
Together with the fact that 
$\mathsf t\models\Phi_k$ and $\mathsf f\not\models\Phi_k$,  this shows that for \cref{lemma:correctnessatSiMinus} with
$\psi=\Phi_{k-1}$ and $\beta_j=v_j$.
That equivalence identifies the almost-sure until objective at $s_i^-$
with satisfiability of $F_i$ after fixing the earlier answers, which is
exactly $v_i=1$. By the definition of $\Phi_k$, we obtain  
\[
 s_i^-\models\Phi_k\quad\Longleftrightarrow\quad v_i=0
 \qquad(k\geq2i-1).
\]
Finally, for $k\geq2i$. Since $k-1\geq2i-1$, the equivalence just
proved and \cref{lemma:correctnessatSiMinus} gives
\[
 s_i^+\models\Phi_k
 \quad\Longleftrightarrow\quad s_i^-\not\models\Phi_{k-1}
 \quad\Longleftrightarrow\quad v_i=1.
\]
This completes the induction.
\end{proof}
\Cref{prop:iratl-snsat-induction} shows that 
\[
 I\in\mathrm{SNSAT}
 \quad\Longleftrightarrow\quad v_n=1
 \quad\Longleftrightarrow\quad (G_I,L_I),s_n^+\models\Phi_{2n}.
\]
This completes the proof of the lemma.
\end{proof}

\iratlTheorem*
\begin{proof}[Proof of~\cref{thm:iratl-model-checking}]
We evaluate the input formula $\varphi$ from the inside out, computing for each state subformula the set of states where it evaluates to true. 
\begin{itemize}
  \item For an atom, this set is given by the labelling. 
  \item For negation, we take the complement of the set already computed for its argument. 
  \item For disjunction, we take the union of the two argument sets. 
  \item When we evaluate a subformula of this form $\langle\!\langle C\rangle\!\rangle_w^r\theta$, the satisfying sets of its arguments are already computed. 
  \item We use \cref{lem:iratl-one-modality} to determine, from each state, whether the coalition can ensure the objective under the specified winning condition and type of randomisation.
\end{itemize}
We can use an induction  on the formula to show that this procedure computes $\operatorname{Sat}(\varphi)$ correctly.

We now show that the problem is in $\P^\NP$, by showing how to make polynomially many queries to the $\NP$ oracle. 
There are at most $|\varphi|$ state subformulas, and each subformula asking whether a coalition can ensure an objective requires at most one oracle query per state. 
Each query specifies a winning problem and an instance of that problem, and asks whether the coalition wins from the chosen state.
Each query asks which of the winning problems in \cref{lem:iratl-one-modality} is to be decided, and these problems are all in $\NP$.
Constructing the queries and computing the remaining states  that satisfy the  state-formula takes polynomial time. 
The procedure is therefore deterministic polynomial time with an $\NP$ oracle, proving membership in $\P^{\NP}$.

\paragraph{Hardness.}
The hardness follows from \cref{lem:iratlQualitativeHardness} even with at most three players and for a fragment that contains deterministic transitions and only sure and almost-sure modalities of the logic.

\end{proof}

\end{document}